\documentclass[11pt]{article}

\usepackage{amsmath}
\usepackage{amsthm}
\usepackage{amssymb}
\usepackage{graphicx}
\usepackage{color}
\usepackage{tikz}
\usepackage{wrapfig}
\usepackage{mathrsfs}
\usepackage{algorithmic}
\usepackage{xr} 
\usepackage{dsfont} 

\usepackage[top=1in, bottom=1in, left=0.83in, right=0.8in]{geometry}
\usepackage{setspace}
\usepackage{titling}

\newcommand{\indep}{\rotatebox[origin=c]{90}{$\models$}}

\newcommand{\dx}[1]{\ \text{d} #1}
\newcommand{\E}{\mathbb{E}}

\newcommand{\indicator}[1]{\mathds{1}\{ #1 \}}

\newcommand{\x}{\mathbf{x}}
\newcommand{\T}{\mathbf{t}}

\newtheorem{lem}{Lemma}
\newtheorem{defn}{Definition}
\newtheorem*{defn*}{Definition}
\newtheorem{prop}{Proposition}

\usepackage[numbers,sort&compress]{natbib}
\title{Risk ratios for contagious outcomes}

\author{Olga Morozova$^1$, Ted Cohen$^1$, Forrest W. Crawford$^{2,3,4}$ \\[1em]
\small 1. Department of Epidemiology of Microbial Diseases, Yale School of Public Health \\
\small 2. Department of Biostatistics, Yale School of Public Health \\
\small 3. Department of Ecology \& Evolutionary Biology, Yale University \\
\small 4. Yale School of Management}

\date{}

\begin{document}
\maketitle


\begin{abstract}

\noindent 
The risk ratio is a popular tool for summarizing the relationship between a binary covariate and outcome, even when outcomes may be dependent. Investigations of infectious disease outcomes in cohort studies of individuals embedded within clusters -- households, villages, or small groups -- often report risk ratios.  Epidemiologists have warned that risk ratios may be misleading when outcomes are contagious, but the nature and severity of this error is not well understood. In this study, we assess the epidemiologic meaning of the risk ratio when outcomes are contagious. We first give a structural definition of infectious disease transmission within clusters, based on the canonical susceptible-infective epidemic model. From this standard characterization, we define the individual-level ratio of instantaneous risks (hazard ratio) as the inferential target, and evaluate the properties of the risk ratio as an estimate of this quantity.  We exhibit analytically and by simulation the circumstances under which the risk ratio implies an effect whose direction is opposite that of the true individual-level hazard ratio. In particular, the risk ratio can be greater than one even when the covariate of interest reduces both individual-level susceptibility to infection, and transmissibility once infected. We explain these findings in the epidemiologic language of confounding and relate the direction bias to Simpson's paradox.  \\[1em]
\end{abstract}


\section{Introduction}

Risk ratios are often recommended for summarizing the relationship between a covariate and an outcome in epidemiology \citep{sinclair1994clinically,davies1998can,bracken1999avoidable,skov1998prevalence,jewell2003statistics,Greenland2004Model,liberman2005much,katz2006relative,rothman2008modern,lumley2006relative}.  Risk ratios, sometimes called ``prevalence ratios'' or ``prevalence proportion ratios'' are simple and easy to compute \citep{mcnutt2003estimating,zou2004modified,spiegelman2005easy}, either by aggregating individual-level outcomes, or obtained directly from population-level surveillance data. When outcomes may exhibit dependence within clusters, ``robust" standard errors are available \citep{zou2013extension}.  Many researchers report risk ratios in studies of infectious disease outcomes within clusters or single communities of interacting individuals \citep{jackson2011serologically, bower2016effects, dowell1999transmission, araujo2015risk, seward2004contagiousness, kim2012secondary}.

The risk ratio is known to have desirable robustness properties when outcomes are dependent \citep{zou2013extension,lumley2006relative}, but infectious disease epidemiologists have repeatedly warned that when outcomes are contagious, simplistic summaries of risk may be misleading \citep{longini1982estimating,longini1988statistical,koopman1991assessing,halloran1995causal,halloran1997study,chick2001bias,eisenberg2003bias,koopman2004modeling,pitzer2012linking,ohagan2014estimating,sharker2017estimation}.  Though they are often assumed to be time-invariant, risk ratios can vary over time in both observational \citep{goldstein2017temporally} and experimental \citep{scott2014timing} studies.  Error can arise when analytical methods do not separate the effects of a covariate on susceptibility to infection from infectiousness once infected \citep{halloran1995causal,kenah2015semiparametric}.  Several authors have suggested that epidemiologists must take exposure to infection into account when assessing risk factors for infectious disease outcomes \citep{longini1988statistical, rampey1992discrete, halloran1994exposure, rhodes1996counting, kenah2015semiparametric}.  However, to our knowledge, none have explained formally why the risk ratio may not be a satisfactory measure of association under contagion, and how its properties depend on the covariate of interest and the epidemiologic features of disease transmission.  

In this paper, we investigate the properties of the risk ratio when outcomes are contagious within clusters.  We first introduce a canonical definition of infectious disease contagion, based on the widely used susceptible-infective epidemic model \citep{anderson1992infectious,andersson2012stochastic}.  This structural description of disease transmission formalizes the epidemiologic intuition that a susceptible individual's risk of infection at a given time depends both on their own traits, and those of their infectious contacts \citep{rhodes1996counting,kenah2013non,kenah2015semiparametric}.  We define the inferential target as the ratio of instantaneous individual-level risks (hazards) of infection under a one-unit change in the value of a covariate \citep{halloran1997study,ohagan2014estimating}. Because the risk ratio is a measure of association between a covariate and outcome, investigators may expect that it provides a reasonable summary of the individual-level relationship between the covariate and susceptibility to infection, a property we call ``direction-unbiasedness''.  We show that this intuition is often correct when the covariate is jointly independent within clusters, the outcome is not contagious, or when the covariate does not affect transmissibility.  However, the risk ratio is in general not direction-unbiased when contagion is present.  

We characterize the epidemiologic features of infectious disease transmission that may lead investigators to report seriously misleading risk ratio estimates in the simplest setting of clusters of size two.  Further analytic results and simulations provide insight into the risk ratio under contagion in clusters of larger size and in randomized trials. Finally, we explain these results in the familiar epidemiologic context of bias induced by confounding.


\section{Setting}

Consider a collection of clusters (e.g. households, workplaces, villages), with $n_i$ subjects in cluster $i$.  Let $Y_{ij}(t)$ be the binary indicator of infection for subject $j$ in cluster $i$ on or before time $t \ge 0$.  Let $T_i$ be the time at which outcomes in cluster $i$ are observed and recorded by researchers.  We consider a single time-invariant binary covariate $x_{ij}$ for subject $j$ in cluster $i$.  The risk ratio is defined as
\begin{equation}
  RR = \frac{\E[Y_{ij}(T_i)|x_{ij}=1]}{\E[Y_{ij}(T_i)|x_{ij}=0]} .
  \label{eq:rr}
\end{equation}
The risk ratio is implicitly a function of the observation time $T_i$ for each cluster $i$ \citep{smith1984assessment}.  


\subsection{Data-generating process}

\begin{figure}
  \centering
 \includegraphics[width=\textwidth]{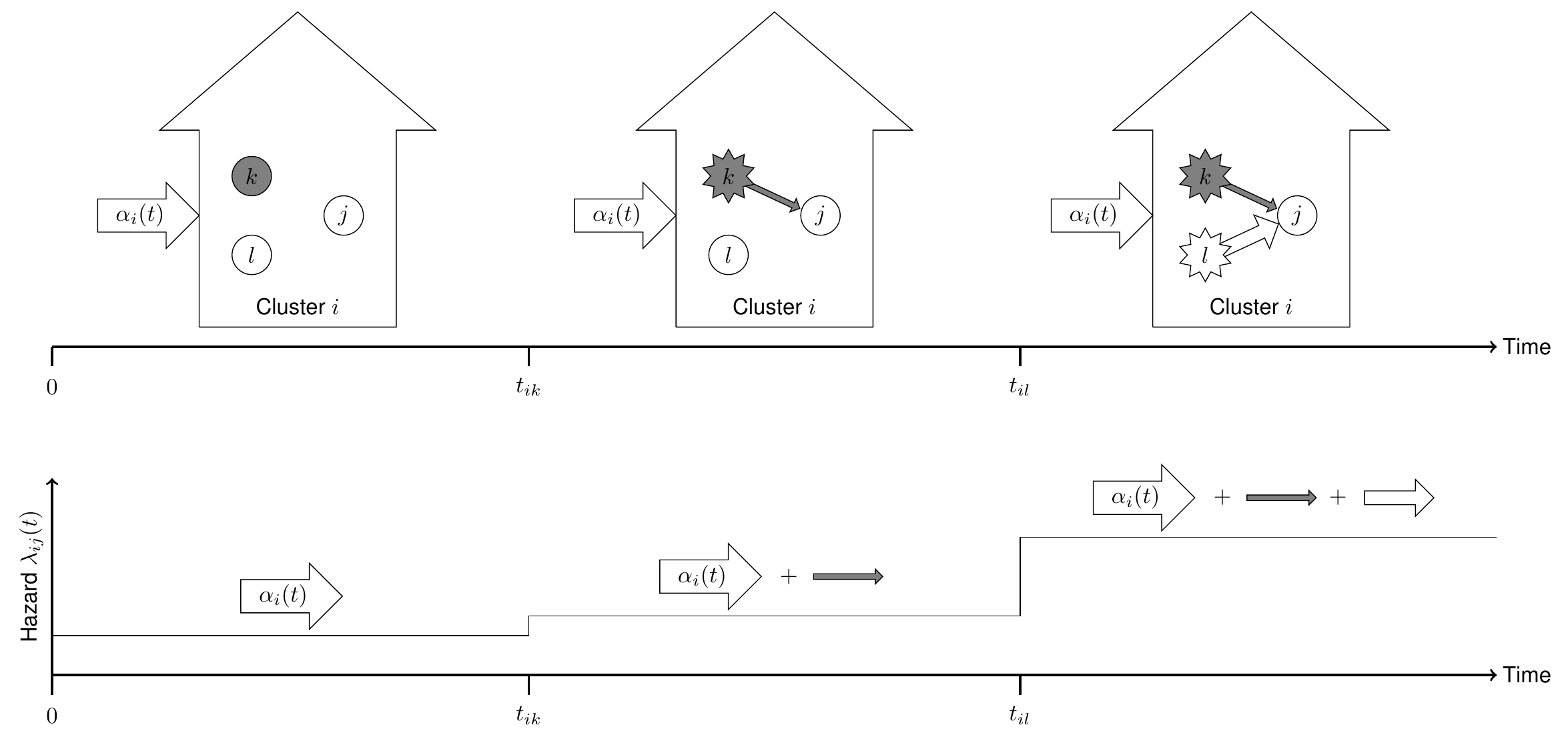}
 \caption{Schematic illustration of hazards in the data-generating process for a cluster of size three. Gray color indicates $x=1$.  Before the first infection, subject $j$ experiences only an exogenous (community) force of infection $\alpha_i(t)$, because neither $k$ nor $l$ is infected.  After $k$ is infected, the hazard to $j$ increases in proportion to the infectiousness of $k$, which is a function of $x_{ik}=1$. Likewise, after $l$ is infected, the hazard to $j$ increases in proportion to the infectiousness of $l$, with $x_{il}=0$. Below, the sum of hazards experienced by subject $j$ is shown over time.  }
  \label{fig:haz}
\end{figure}

We describe a data-generating process based on the canonical susceptible-infective model of infectious disease contagion within clusters \citep{anderson1992infectious,andersson2012stochastic,rhodes1996counting,kenah2013non,kenah2015semiparametric}, then characterize the hazard ratio corresponding to a one-unit change in a covariate associated with susceptibility to disease. The susceptible-infective model captures the intuition that the risk to a susceptible individual at a given time is given heuristically by 
\begin{equation}
  \text{risk of infection} = (\text{susceptibility})\times (\text{force of infection}) 
\label{eq:heuristic}
\end{equation}
where ``susceptibility'' is a function of the subject's own characteristics, and ``force of infection'' summarizes the risk transmitted by that subject's infectious contacts. 

To formalize this risk, let $t_{ij}$ be the minimum of the infection time of subject $j$ in cluster $i$ and the observation time $T_i$, so that $Y_{ij}(t)=0$ for $t \le t_{ij}$, and $Y_{ij}(t)=1$ for $t_{ij} < t \le T_i$. A subject $j$ in cluster $i$ is called susceptible at time $t$ if $Y_{ij}(t)=0$. 
Consider the possible sources of transmission to a susceptible subject $j$ in cluster $i$.  First, $j$ may be infected by exposure to an exogenous source of infection (sometimes called the community force of infection if clusters are households).  Let $\tau_{ij}^e$ be the waiting time for $j$ to be infected from this exogenous source, and let $\lambda_{ij}^e(t)$ be the hazard of this event at time $t$. Second, suppose another subject $k$ in cluster $i$ becomes infected at a time $t_{ik}\in [0,T_i)$, which is defined similarly to $t_{ij}$ as the minimum of the infection time of subject $k$ and cluster $i$ observation time $T_i$. Suppose subject $j$ is not infected at time $t_{ik}$, $Y_{ij}(t_{ik})=0$.  Let $\tau_{ij}^k$ be the waiting time (measured since $t_{ik})$ for $k$ to transmit the infection to $j$, and let $\lambda_{ij}^k(t)$ be the hazard of this event at time $t > t_{ik}$.  
For each cluster $i$ and susceptible subject $j$, 
the total hazard experienced by a susceptible individual $j$ 
is the sum of these hazards,
\begin{equation}
		\lambda_{ij}(t) = \lambda_{ij}^e(t) + \sum_{k=1}^{n_i} \lambda_{ij}^k(t) Y_{ik}(t).
		\label{eq:haz_short}
\end{equation}
The additive form of \eqref{eq:haz_short} arises because $j$ experiences \emph{competing risks} of infection: from the exogenous source, and from each of their infectious contacts. The transmission hazard from subject $k$, $\lambda_{ij}^k(t)$, is only present if $k$ is infected, that is, $Y_{ik}(t)=1$.   Under this simple generative process, subjects may not be re-infected. 

We assume for simplicity that the hazards $\lambda_{ij}^e(t)$ and $\lambda_{ij}^k(t)$ are Cox-type models: each decomposes into the product of a possibly time-varying force of infection and a function of covariates. Let $\lambda_{ij}^e(t) = \alpha_i(t)e^{x_{ij}\beta}$ where $\alpha_i(t)$ is the possibly time-varying exogenous force of infection to cluster $i$, and $\beta$ is a susceptibility parameter corresponding to the binary covariate $x$.  Likewise, when $t>t_{ik}$, let $\lambda_{ij}^k(t) = \omega_{ikj}(t-t_{ik}) e^{x_{ij}\beta + x_{ik}\gamma}$ where $\omega_{ikj}(t-t_{ik})$ is the possibly time-varying force of infection from subject $k$ to subject $j$ in cluster $i$, and $\gamma$ is an infectiousness parameter corresponding to the binary covariate $x$.  Then the total infection hazard to susceptible subject $j$ in cluster $i$ at time $t$ becomes 
\begin{equation}
  \lambda_{ij}(t) = e^{x_{ij}\beta} \left(\alpha_i(t) + \sum_{k=1}^{n_i} Y_{ik}(t) \omega_{ikj}(t-t_{ik}) e^{x_{ik}\gamma} \right).
  \label{eq:haz}
\end{equation}
The multiplicative relationship between susceptibility $e^{x_{ij}\beta}$ and the total force of infection in \eqref{eq:haz} mirrors the heuristic description of infection risk given by \eqref{eq:heuristic}. When $x_{ij}=x$ is constant across individuals, $\alpha_i(t)=0$, and $\omega_{ikj}(t-t_{ik})=\omega$, the process becomes the standard continuous-time Markov susceptible-infective model within clusters.  The formulation of the hazard of infection in \eqref{eq:haz} mirrors a transmission model proposed for semi-parametric relative risk regression \citep{kenah2015semiparametric}. The model captures temporal changes in post-infection transmission via the functional form of $\omega_{ikj}(t-t_{ik})$, which can accommodate latency or other changes in infectiousness over time.  

The hazard ratio ($HR$) is the ratio of these instantaneous risks under different values of the covariate $x$, holding individual-level force of infection constant:
\begin{equation}
 HR = \frac{\lambda_{ij}(t|x_{ij}=1)}{\lambda_{ij}(t|x_{ij}=0)} = e^{\beta}.
 \label{eq:hr}
 \end{equation}
 The hazard ratio summarizes the individual-level association between the covariate $x$ and susceptibility to infection at time $t$ \citep{smith1984assessment,halloran1997study,ohagan2014estimating}.  
 
 We emphasize that we do not treat the data-generating process characterized by \eqref{eq:haz} as an inferential model. We have not specified the possibly time-varying hazards $\alpha_i(t)$ and $\omega_{ikj}(t)$, nor showed that any feature of the process is identified by a particular observation scenario. Instead, \eqref{eq:haz} characterizes the transmission dynamics of infection by which the observable data are assumed to be generated.  Figure \ref{fig:haz} shows a schematic depiction of the data-generating process in a household of size three, and Table \ref{table:par} summarizes the parameters that define this process. 

\begin{table}
\caption{Summary of parameters in the data-generating process. }
\label{table:par} \centering%
\begin{tabular}{l l}
{} & {} \\ \hline
Notation & Definition \\ 
{} & {} \\ \hline
$n_i$ & Size of cluster $i$ \\
{} & {} \\ 
$T_i$ & Observation time for cluster $i$ \\
{} & {} \\ 
$x_{ij}$ & Covariate of interest, time-invariant \\
{} & {} \\ 
$Y_{ij}(t)$ & Binary indicator of infection by time $t$ \\
{} & {} \\ 
$\beta$ & Susceptibility parameter for covariate $x$ \\ 
{} & {} \\ 
$\gamma$ & Infectiousness parameter for covariate $x$ \\
{} & {} \\ 
$\alpha_i(t)$ & Exogenous force of infection, a function of time \\
{} & {} \\ 
$\omega_{ikj}(t)$ & Force of infection from infectious $k$ to susceptible $j$, \\
& a function of time since infection of $k$ \\
\hline
\end{tabular}
\end{table}

 It seems reasonable to expect the risk ratio \eqref{eq:rr} for the binary variable $x$, as a \emph{marginal} or \emph{population-level} measure of association, to be meaningful for assessment of the ratio of \emph{conditional} risks \eqref{eq:hr} experienced by an individual. 
 Since the hazard ratio \eqref{eq:hr} evaluated at a time $t$ is time-invariant, we might expect the risk ratio, as a cross-sectional measure of association at time $t$, to have similar properties. To make this notion more formal, we define a general property that we would like the risk ratio to satisfy.

\begin{defn}[Direction-unbiasedness of risk ratio]
If $HR<1$, then $RR<1$,
if $HR=1$, then $RR=1$, and
if $HR>1$, then $RR>1$.
  \label{defn:unbiased}
\end{defn}

\noindent When Definition \ref{defn:unbiased} holds, the ratio of marginal risks \eqref{eq:rr} may be regarded as a reasonable surrogate for the ratio of individual conditional risks \eqref{eq:hr}.  Definition \ref{defn:unbiased} is a relatively weak requirement: it does not imply monotonicity in the risk ratio as a function of the hazard ratio, nor any particular functional relationship between the two. We say that for a particular study design and values of parameters in \eqref{eq:haz}, the risk ratio is \emph{direction-unbiased} if Definition \ref{defn:unbiased} holds.


\section{Results}


\subsection{Clusters of size two}

We first consider a simple parametric version of \eqref{eq:haz} with two-person clusters and balanced covariate values for which a variety of precise analytic results can be derived. This setting is based on a two-person infectious disease contagion model introduced previously \citep{vanderweele2011bounding,vanderweele2012components,ogburn2014vaccines}, and serves to illustrate the potential for the risk ratio \eqref{eq:rr} to give a misleading summary of association under contagion. Clusters of size two appear in empirical study designs, including HIV transmission in couples \citep{carpenter1999rates, biraro2013hiv}, and mother-to-child transmission of \textit{Staphylococcus aureus} \citep{regev2009parental, leshem2012transmission}.
Consider the data-generating process \eqref{eq:haz}, where each cluster $i$ consists of exactly two subjects: $n_i = 2$.  Assume also that the covariate is balanced within the cluster, subject 1 has $x_{i1} = 1$ and subject 2 has $x_{i2} = 0$; all subjects are uninfected at baseline, $Y_{ij}(0) = 0$; and follow-up time is constant, $T_i = T$ for all $i$.  Additionally, assume there is a constant exogenous force of infection $\alpha_i(t)=\alpha$, and constant within-cluster contagion $\omega_{ikj}(t-t_{ik}) = \omega$ per susceptible $j$ and infected $k$.  The hazards of infection experienced by subjects $1$ and $2$ in cluster $i$ become $\lambda_{i1} (t) = e^\beta (\alpha + \omega Y_{i2}(t))$ and $\lambda_{i2} (t) = \alpha + \omega e^\gamma Y_{i1}(t)$ respectively.  The following results establish the properties of the risk ratio as an approximation to the hazard ratio in several relevant special cases. Proofs of all results are given in the Supplement.

When there is no within-cluster contagion, as is the case in studies of non-transmissible outcomes, the risk ratio is direction-unbiased. 
\begin{prop}[No within-cluster contagion]
  Suppose $\alpha>0$ and $\omega=0$.  For any $T>0$, the risk ratio is direction-unbiased. 
\label{prop:biasnoclustercont}
\end{prop}

\noindent Define the ``null'' hypothesis under the data-generating process as $\beta = 0$, so that all subjects are equally susceptible to infection.  When outcomes are measured after enough time has elapsed, the direction of the risk ratio is entirely determined by the infectiousness coefficient $\gamma$. 

\begin{prop}[Under the null]
  Suppose $\beta=0$ and $T>0$. Then if $\gamma<0$, $RR>1$; if $\gamma>0$, $RR<1$; and if $\gamma=0$, $RR=1$. 
\label{prop:biasundernull}
\end{prop}
\noindent If the covariate $x$ does not alter infectiousness given infection, direction-unbiasedness holds. 

\begin{prop}[Homogenous infectiousness]
  Suppose $\gamma=0$.  For any $T>0$, the risk ratio is direction-unbiased.  
\label{prop:biasnoinfect}
\end{prop}

\noindent However, when susceptibility $\beta$ and infectiousness $\gamma$ have the same sign, $\gamma$ is sufficiently large in absolute value, and the follow-up time $T$ is large enough, direction bias may occur. 

\begin{prop}[Direction bias] 
    Suppose 
    either  $\beta<0$ and $e^\gamma < \min\{e^{2\beta},e^\beta + \frac{\alpha}{\omega} (e^\beta -1)\}$, or $\beta>0$ and $e^\gamma > \max\{e^{2\beta},e^\beta + \frac{\alpha}{\omega} (e^\beta -1)\}$.  Then there exists $t^*> 0$ such that for all $T>t^*$, the risk ratio is not direction-unbiased. 
\label{prop:biasacrossnull}
\end{prop}
\noindent Figure \ref{fig:hazbias} illustrates Result \ref{prop:biasacrossnull}.

\begin{figure}
\centering
\includegraphics[width=\textwidth]{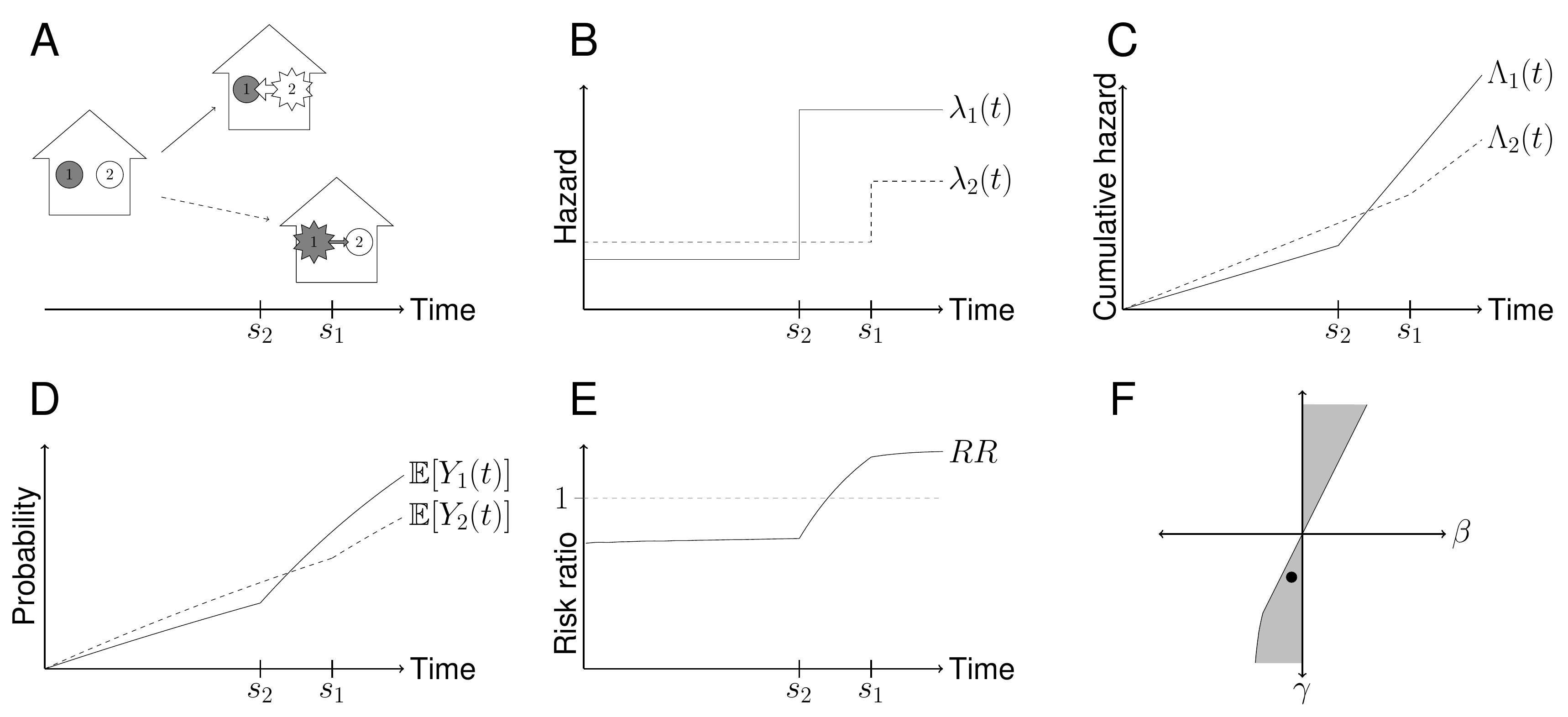} 
\caption{Illustration of how the risk ratio can give misleading results under contagion in a cluster of size two. Suppose the conditions of Result \ref{prop:biasacrossnull} hold with $\beta<0$, subject 1 (gray) has $x_1=1$, subject 2 (white) has $x_2=0$.  (A) Either subject 2 becomes infected first (at time $s_2$), or subject 1 is infected first (at time $s_1$); depending on which is infected first, the other experiences a change in their hazard of infection. (B) If 1 is infected first, then 2 experiences only a small increase in hazard, because $\gamma<0$.  Alternatively, if 2 is infected first, then 1 experiences a large increase in hazard because $x_2=0$. (C) The relationship between the cumulative hazards in these scenarios, and hence the relationship between the expected infection outcomes (D), is eventually reversed at some time after $s_2$.  Therefore the risk ratio (E) eventually rises above one.  Panel (F) shows the region of $(\beta,\gamma)$ parameter space in which direction bias may occur, where $\beta$ and $\gamma$ are plotted on the same scale. A black dot indicates the values of $\beta$ and $\gamma$ used in this illustration. }
\label{fig:hazbias}
\end{figure}

Direction-unbiasedness under Definition \ref{defn:unbiased} does not imply zero bias.  Figure \ref{fig:exact_main} shows the expected value of $\log[RR]$ across values $\beta$ and $\gamma$ for several values of $\omega / \alpha$. To make results comparable in every sub-figure, the observation time $T$ is selected so that cumulative incidence at time $T$ when $\beta = 0$ and $\gamma = 0$ is held constant at approximately 0.15.  
The Supplement provides an exact expression for the bias of $\log[RR]$ and similar plots for a wider range of parameters $\alpha$ and $\omega$.  As an approximation to the hazard ratio, the risk ratio is always biased unless $\beta = 0$ and either $\omega = 0$ or $\gamma=0$ holds. For all other combinations of parameters, whenever the risk ratio is direction-unbiased, it is biased towards the null of $\beta=0$. 

\begin{figure}
\centering
\includegraphics[scale=0.88]{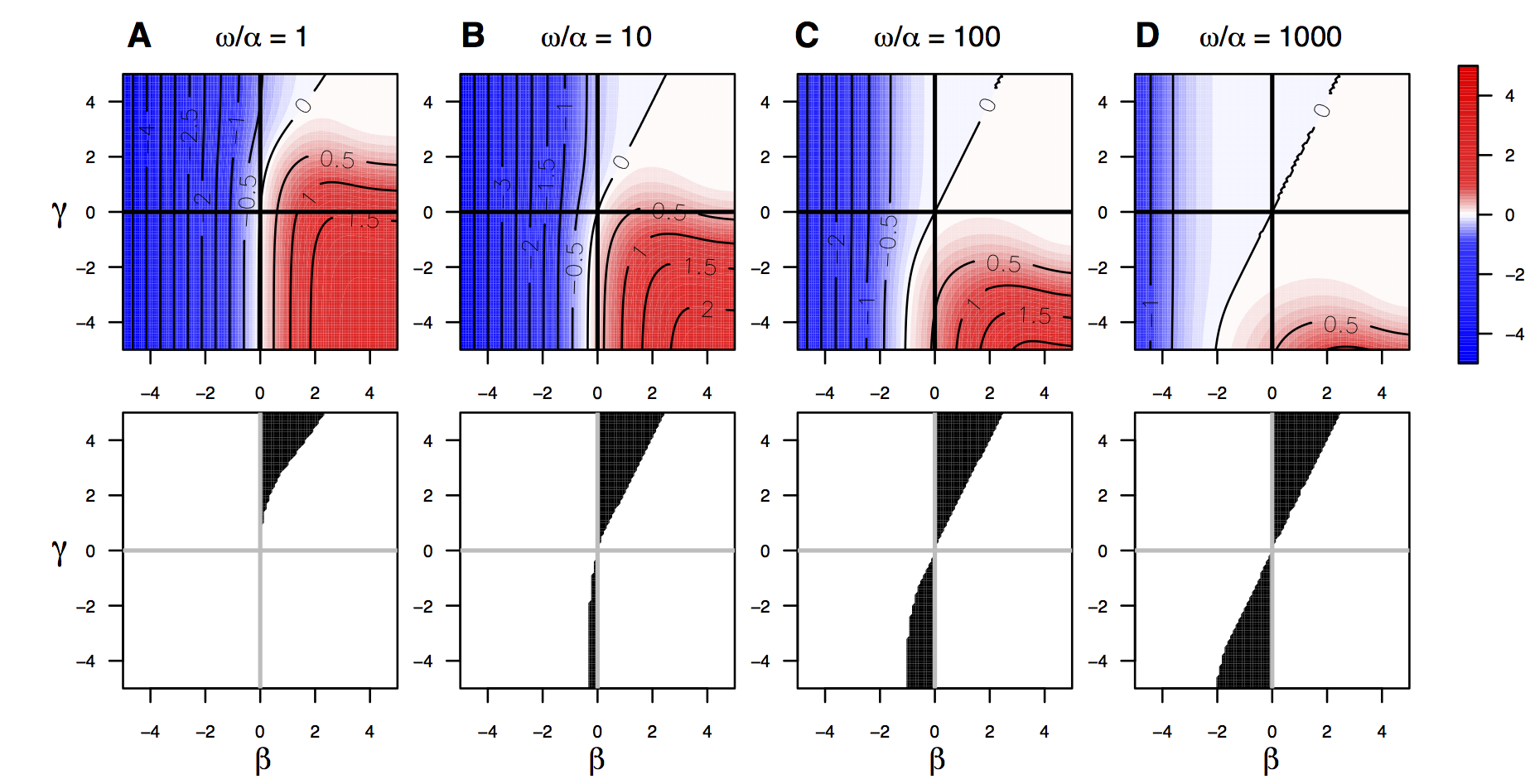}
\caption{Computed $\log[RR]$ (top row) and region of direction bias (bottom row) as a function of $\beta$ and $\gamma$ in clusters of size two with exactly one subject per cluster with $x=1$. For any given ratio $\omega / \alpha$ observation time $T$ was chosen such that cumulative incidence when $\beta=0$ and $\gamma=0$ is approximately 0.15. } 
\label{fig:exact_main}
\end{figure}


\subsection{General clusters}

Most cluster cohort studies of infectious diseases involve variable cluster sizes and a more complex design. Several factors may influence the behavior of the risk ratio in empirical studies, such as epidemiologic features like $\alpha_i(t)$ and $\omega_{ikj}(t)$, and aspects of study design such as experimental assignment of the covariate $x$, the duration and variability of observation time $T_i$, cluster size distribution, or selection of clusters with or without infected individuals at baseline. 

When there is no within-cluster contagion and the covariate is independent of the exogenous force of infection and observation time, the risk ratio is direction-unbiased.  Let $\x_i=(x_{i1},\ldots,x_{in_i})$ be the vector of covariate values in cluster $i$.

\begin{prop}[No within-cluster contagion]
  Suppose $\omega_{ikj}(t)=0$ for all $t$ and $\x_i \indep \{\alpha_i(t), n_i, T_i\}$.  Then the risk ratio is direction-unbiased. 
\label{prop:biasnoclustercontgen}
\end{prop}

\noindent Joint independence of within-cluster covariates guarantees direction-unbiasedness for any parameter values.  

\begin{prop}[Independent $\x$]
  Suppose the covariates $\x_i=(x_{i1},\ldots,x_{in_i})$ are jointly independent and $\x_i \indep \{\alpha_i(t),\omega_{ikj}(t), n_i, T_i\}$. Then the risk ratio is direction-unbiased. 
\label{prop:biasindependentx}
\end{prop}

The risk ratio is not generally direction-unbiased when the joint distribution of the covariate $x$ is dependent.  For example, direction unbiasedness may not hold under two common randomization schemes: ``block randomization'' within clusters, when a fixed number of subjects per cluster have $x=1$ with $\Pr(\x_i)=\binom{n_i}{k_i}^{-1}\indicator{\sum_j x_{ij}=k_i}$, and ``cluster randomization'' with $x_{ij}=1$ for all $j$ in some subset of clusters, and $x_{ij}=0$ for all $j$ in the remaining subset. In general, when the joint distribution of $\x_i$ is not independent, or when there is heterogeneity in $n_i$, $\alpha_i(t)$, or $\omega_{ikj}(t)$ across clusters, the risk ratio need not be direction unbiased, even when $\gamma=0$.
Dependence in $\x_i$ may occur in observational studies, where $\x_i$ may be dependent due to shared environment, genetic factors, or other forms of dependence within clusters.

\subsection{Simulation results}

Analytical expressions for the bias of the risk ratio as an approximation to the hazard ratio \eqref{eq:hr} are intractable in general. However, simulations can provide further insight under particular epidemiologic and study design parameters.  In simulations, we vary the distribution of covariates $\x_i$, cluster size $n_i$, observation time $T_i$, infected cluster members at baseline, and values of force of infection parameters $\alpha$ and $\omega$, which are assumed to be constant over time and clusters. 
A comprehensive set of simulation results and R code \citep{R2017} for replicating the simulations appear in the Supplement.

\begin{figure}
\centering
\includegraphics[scale=0.6]{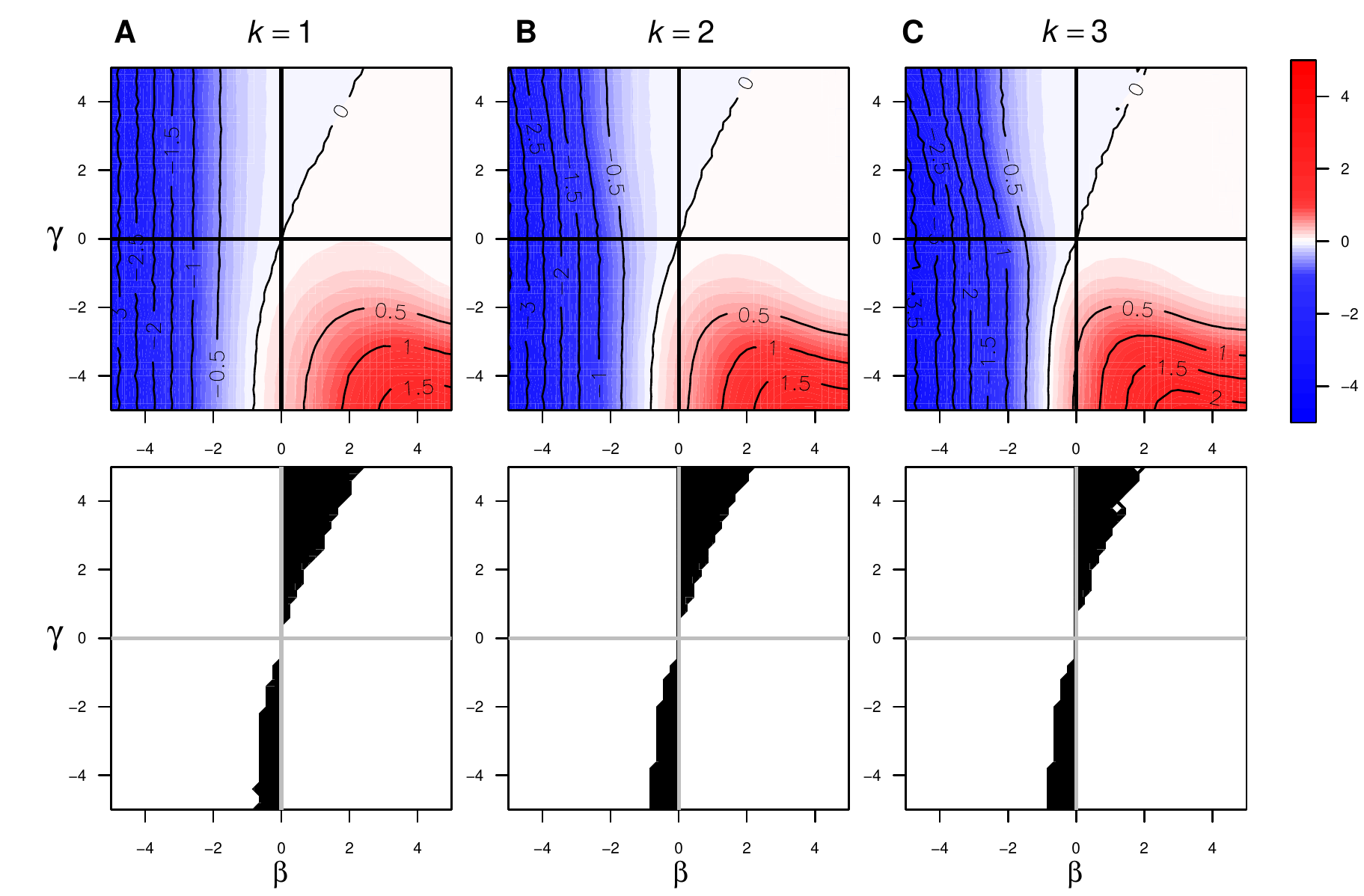}
\caption{$\log[RR]$ (top row) and region of direction bias (bottom row) as a function of $\beta$ and $\gamma$, when cluster size is constant ($n_i=4$ for all $i$) and block randomized distribution of $x$: $\sum_{j=1}^{n_i} x_{ij} = k$. In all plots $\alpha = 0.0001$, $\omega = 0.01$, $N=500$, $T_i = 450$, and no subjects are infected at time zero.}
\label{fig:sim_main_set1}
\end{figure}

Some properties of the two person-cluster case hold in more complex scenarios. Figure \ref{fig:sim_main_set1} shows results for clusters of size four and block randomized distribution of $x$, such that each cluster has exactly $k$ subjects with $x=1$, $k = 1, 2, 3$. The behavior of the bias in Figure \ref{fig:sim_main_set1} mimics that of the two-person cluster case. We demonstrated analytically in Result \ref{prop:biasnoclustercontgen} that direction-unbiasedness under no within-cluster contagion holds under independence of $\x_i$ and cluster level parameters $\alpha_i(t)$, $n_i$ and $T_i$. The simulation shows that direction bias results under constant cluster size and block randomized $x$ are similar to that of two-person cluster case for sufficiently large observation times $T_i$. 

It follows from Result \ref{prop:biasindependentx} that the risk ratio is direction-unbiased under independent Bernoulli assignment of $\x_i$.  In practical intervention trials, many studies in small clusters employ block or cluster randomization. Simulation results show that both of these methods result in the risk ratio having direction bias in some region of $(\beta, \gamma)$ space.  Figure \ref{fig:sim_main_set1} shows that block randomized distribution of $\x_i$ results in direction bias in regions where $\beta$ and $\gamma$ have the same sign and $\gamma$ is more extreme than $\beta$. Figure \ref{fig:sim_main_set2} illustrates cluster randomized distribution of $\x_i$, showing direction bias in regions where $\beta$ and $\gamma$ have opposite sign.

\begin{figure}
\centering
\includegraphics[scale=0.6]{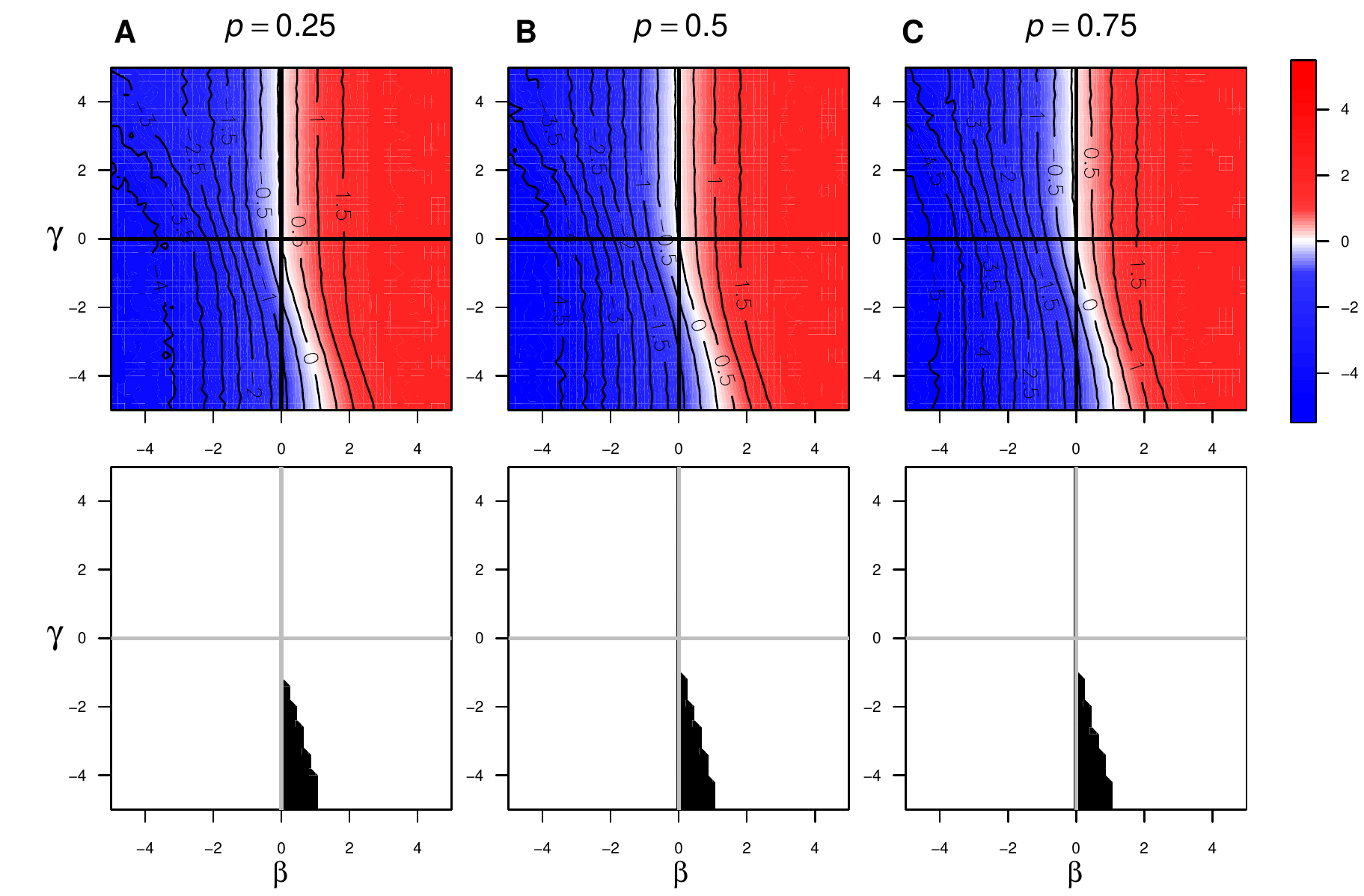}
\caption{$\log[RR]$ (top row) and region of direction bias (bottom row) as a function of $\beta$ and $\gamma$ when cluster size is constant and $x$ is cluster randomized: proportion $p$ of clusters have $\sum_{j=1}^{n_i} x_{ij}= 4$, and remaining $1-p$ have $\sum_{j=1}^{n_i} x_{ij}= 0$. In all plots $\alpha = 0.0001$, $\omega = 0.01$, $N=500$, $n_i = 4$, $T_i = 450$, and no subjects are infected at time zero.}
\label{fig:sim_main_set2}
\end{figure}

When cluster size $n_i$ varies, bias patterns can change substantially with the nature of dependence in the distribution of $\x_i$.  Even under block randomization, the direction bias pattern differs depending on allocation proportion, and generally worsens with imbalance between $\Pr(x_{ij} = 1)$ and $\Pr(x_{ij} = 0)$. 
Figure \ref{fig:sim_main_set3} illustrates direction bias under variable cluster sizes with exactly one subject per cluster having $x=1$, and Figure \ref{fig:sim_main_set4} shows balanced block randomized $x$ under variable cluster sizes.  It is not necessary for $\gamma$ to be more extreme than $\beta$, nor must these parameters have the same sign, to observe direction bias. While regions where the risk ratio exhibits direction bias become smaller in Figure \ref{fig:sim_main_set4} compared to Figure \ref{fig:sim_main_set3}, in both cases direction bias is present when $\gamma=0$. Thus the desirable property of direction-unbiasedness under homogenous infectiousness disappears when cluster sizes vary. 

\begin{figure}
\centering
\includegraphics[scale=0.8]{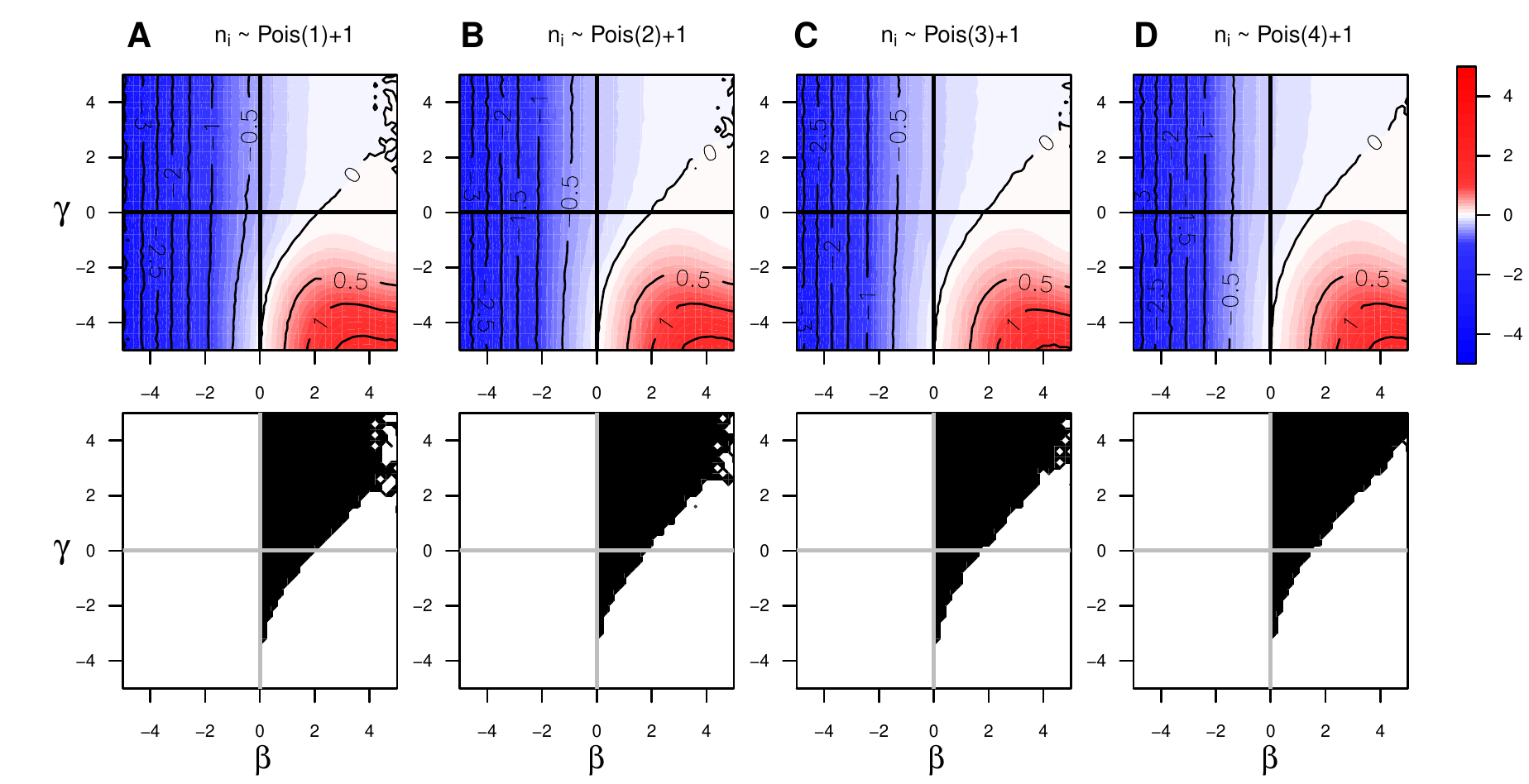}
\caption{$\log[RR]$ (top row) and region of direction bias (bottom row) as a function of $\beta$ and $\gamma$ when cluster size $n_i \sim \text{Pois}(\mu)+1$ and $x$ is block randomized such that $\sum_{j=1}^{n_i} x_{ij}=1$ for all $i$. In all plots $\alpha = 0.0001$, $\omega = 0.01$, $N=500$, no subjects are infected at time zero, and observation time $T_i$ is chosen such that cumulative incidence when $\beta=0$ and $\gamma=0$ is approximately 0.15.}
\label{fig:sim_main_set3}
\end{figure}

\begin{figure}
\centering
\includegraphics[scale=0.8]{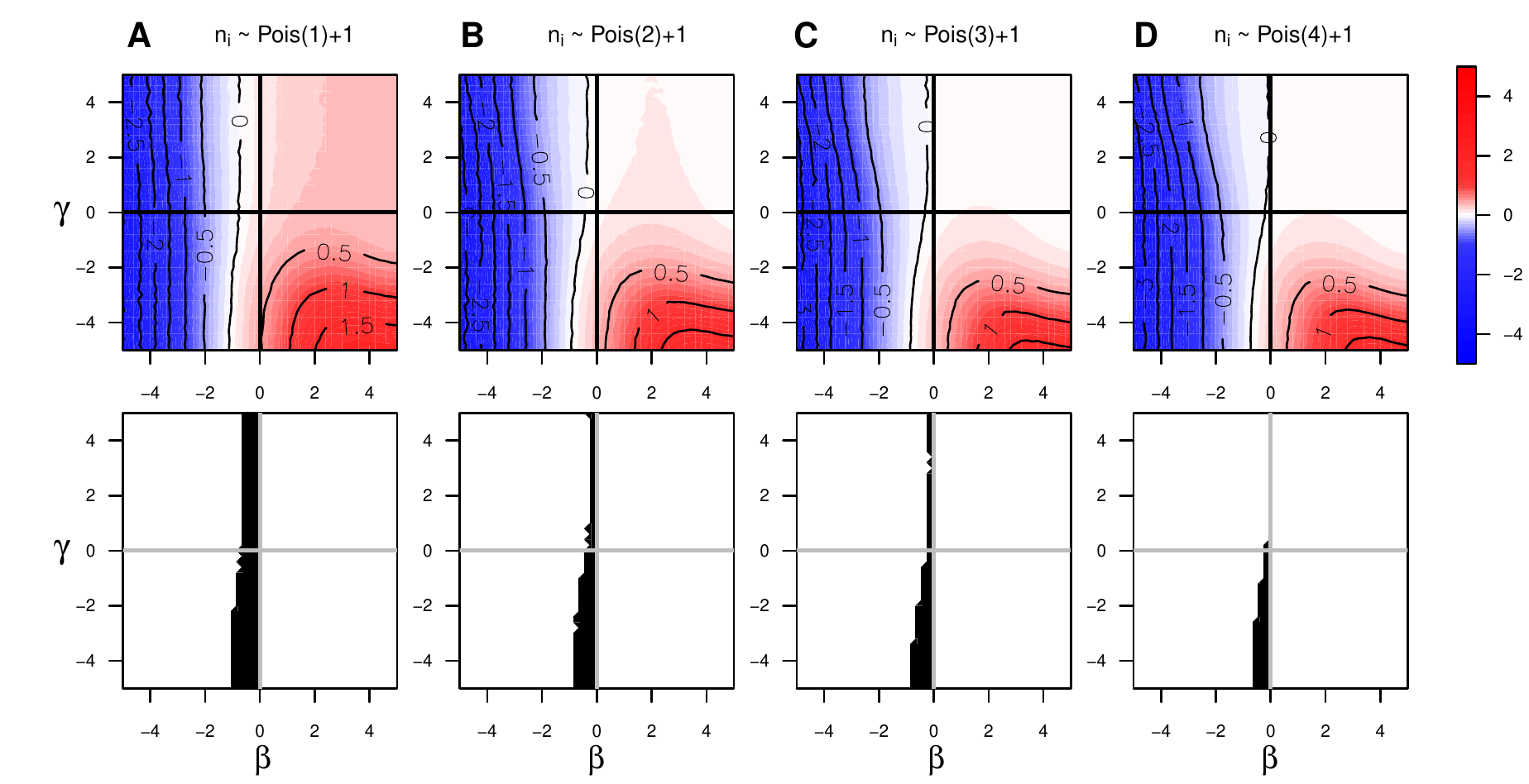}
\caption{$\log[RR]$ (top row) and region of direction bias (bottom row) as a function of $\beta$ and $\gamma$ when cluster size $n_i \sim \text{Pois}(\mu)+1$ and $x$ is block randomized such that $\sum_{j=1}^{n_i} x_{ij}=\lfloor n_i/2 \rfloor$ for all $i$. In all plots $\alpha = 0.0001$, $\omega = 0.01$, $N=500$, no subjects are infected at time zero, and observation time $T_i$ is chosen such that cumulative incidence when $\beta=0$ and $\gamma=0$ is approximately 0.15.}
\label{fig:sim_main_set4}
\end{figure}

The duration of observation influences the size of the region in $(\beta, \gamma)$ parameter space where the risk ratio exhibits direction bias. Longer observation times increase the region of direction bias under block randomized distribution of $x$ and reduce the size of this region under cluster randomized distribution of $x$. When the distribution of $x$ is jointly independent, the risk ratio is always direction unbiased; however, increasing the observation time increases the absolute value of the bias. Variable compared with fixed observation time does not meaningfully influence the behavior of the bias under any study design (see Supplement for details).

In real-world cohort studies of infectious disease outcomes, researchers often select clusters (e.g. households) based on infection outcomes detected at baseline (sometimes called ``index'' cases), especially for diseases with low overall prevalence or community force of infection, and risk ratios are computed at follow-up for cluster members susceptible at baseline.  Simulation results, given in the Supplement, show that when subjects are infected at baseline, resulting direction bias depends on the distribution of $x$ among infected and uninfected subjects at baseline.


\section{Discussion}

We have applied a standard and widely accepted measure of association to outcomes generated by a canonical susceptible-infective model of infectious disease contagion.  When the distribution of a covariate is dependent within clusters and associated with both susceptibility to infection and transmissibility once infected, the risk ratio for that covariate may imply an aggregate effect whose direction is opposite that of its individual-level effect on susceptibility to infection.  
  Several other measures of association -- the risk difference, odds ratio, attributable risk, secondary attack rate, and some measures of vaccine efficacy -- may suffer from direction bias in the sense of Definition \ref{defn:unbiased} under similar conditions. Statisticians and epidemiologists have warned that na\"ive summaries of association may be biased under contagion \citep{koopman1991assessing,halloran1997study,eisenberg2003bias,ohagan2014estimating,staples2015incorporating}.  Our finding of direction bias in the risk ratio can be readily understood in terms of concepts already familiar to epidemiologists:
\begin{itemize}
\item Confounding and omitted variable bias \citep{rothman2008modern}: When the covariate $x$ is dependent within clusters, other subjects' covariate values can be regarded as a common cause of both a given subject's covariate value (via dependence), and that subject's outcome (via contagion). Omitting or failing to condition on this common cause can result in bias. 
\item Simpson's paradox \citep{greenland1999confounding,arah2008role,pearl2009causality,pearl2014comment}: When the joint distribution of $x$ in the cluster is dependent, conditioning on $x$ alone induces a non-causal association between the covariate and infection, via exposure to the outcomes of other subjects.  Since this spurious association may be different for subjects with $x=1$ and $x=0$, reversal in the apparent direction of the effect of $x$ on the outcome may occur.  
\item Ecological bias \citep{greenland1989ecological,koopman1994ecological}: Aggregating subjects by their covariate $x$ obscures individual-level differences in their exposure to infection. A causal interpretation of the risk ratio attributes outcomes in the group of subjects with $x=1$ (or $x=0$) to the value of their covariate, when those outcomes may actually be partly due to individual-level exposure to a cluster member with $x=0$ (or $x=1$). 
\item Endogeneity and measurement error \citep{manski1993identification,zohoori1997econometric,jurek2006exposure}: the right-hand side of \eqref{eq:haz} shows that the hazard experienced by a subject is a function of both their covariate value $x$, and the covariates and outcomes of other subjects in the same cluster. These other outcomes are themselves functions of those subjects' covariates $x$. Therefore the residual error in a prediction of a given subject's outcome from their own covariate value $x$ alone is correlated with the covariates $x$ of other cluster members, and $x$ is therefore an endogenous variable. 
\end{itemize}

The risk ratio is a valid statistical estimand: it summarizes the marginal association between $x$ and infection.  However, if investigators are interested in the \emph{causal} direct (i.e. susceptibility) effect of treatment on the person who receives it \citep{halloran1995causal,halloran1997study}, the risk ratio may give a very misleading estimate of this quantity.  
One striking consequence of Result \ref{prop:biasacrossnull} is that direction bias can occur even when subjects are exchangeable and treatment (i.e. $x=1$) is randomized and balanced within each cluster.  The primary factor driving these results is contagion; direction bias can occur even in the absence of unmodeled within-cluster heterogeneity, imbalance in covariate values, or heterogeneity in contact patterns \citep{koopman1991assessing}.

Whether direction bias occurs in a particular empirical investigation depends on the epidemiologic features (i.e. $\alpha_i(t)$, $\omega_{ikj}(t)$) of the disease under study, the distribution of cluster size $n_i$ and observation time $T_i$, and the distribution of $x$ within clusters.  
When a disease is only weakly contagious within clusters
or when within-cluster transmissibility is less than the exogenous force of infection,
direction bias may be less likely to occur.  This may be the situation in many cohort studies of infectious diseases. Result \ref{prop:biasindependentx} may justify the use of risk ratio in experimental studies with simple (Bernoulli) randomization of $x$, or possibly observational studies in which covariates of interest are independent or only weakly dependent within clusters.  In a wide variety of empirical dependence settings in which infection is only weakly contagious, the risk ratio may be a reasonable estimator of the ratio of instantaneous risks \eqref{eq:hr}.  However, studies of highly contagious diseases with weakly effective interventions (e.g. Ebola) may benefit from more careful analysis. 

In this paper the data-generating process 
is represented by a standard susceptible-infective model with subject-specific covariates and an exogenous force of infection. This setting provides a simple generative model that 
incorporates features of infectious disease contagion relevant to the properties of the risk ratio. However, this model does not capture several important aspects of infectious disease dynamics, including recovery, removal, re-infection, and multiple infections. Addition of further realistic features to the data-generating model seems unlikely to reduce the bias in the risk ratio as an estimate of the hazard ratio, and instead may exacerbate its undesirable properties. 

Remedies for the pathologies of the risk ratio under contagion are within reach. Epidemiologists have developed deterministic and stochastic models of infectious disease transmission in groups that take exposure to infection into account \citep{anderson1992infectious,andersson2012stochastic,rampey1992discrete}. Several researchers have developed inferential approaches that capture infectious disease transmission dynamics and permit adjustment for individual-level factors \citep{longini1982household,haber1988models,rampey1992discrete,kenah2015semiparametric}.  Some analyses of contagious outcomes adjust for variables that may be correlated with exposure to infectiousness \citep{fine1997household,huang2014effect,martinez2016infectiousness,staples2016leveraging}.  It remains an open question whether standard regression adjustment using a summary of infection outcomes of other individuals can deliver risk ratio estimates that are direction-unbiased. \\

\textbf{Sources of financial support:} This work was supported by grants R36 DA042643 from NIDA, R01 DA015612 from NIDA, DP2 OD022614 from NICHD, R01 AI112438-03 from NIAID, the Yale Center for Clinical Investigation, and the Center for Interdisciplinary Research on AIDS. Computing support was provided by the Yale Center for Research Computing and the W. M. Keck Biotechnology Laboratory, as well as grants RR19895 and RR029676-01 from NIH.

\textbf{Acknowledgements:} We are grateful to
Peter M. Aronow, 
Xiaoxuan Cai,
Edward H. Kaplan,
Joseph Lewnard,
Marc Lipsitch,
A. David Paltiel, 
Harvey Risch,
Daniel Weinberger,
and
Jon Zelner
for helpful discussion and comments. 

\textbf{Replication of results:} This paper uses only simulated data.  Replication code is available from the authors upon request.


\bibliographystyle{spbasic}

\bibliography{epibias}


\title{Supplement}
\author{}
\date{}

\maketitle

\section*{Risk ratio in clusters of size two}

Consider a cluster of two subjects, both uninfected at baseline, with $x_{1} = 1, x_{2} = 0$. The hazard functions for these subjects are
\[ \lambda_{1} (t) = e^\beta [\alpha + \omega Y_{2}(t)] \]
\[ \lambda_{2} (t) = \alpha + \omega e^\gamma Y_{1}(t) \]
and we are interested in understanding the properties of the risk ratio evaluated at time $t$, 
\[ RR = \frac{\E[Y_{1}(t)]}{\E[Y_{2}(t)]}.  \]
First, let $T_{1}$ and $T_{2}$ be the infection times of subjects 1 and 2, and let $S=\min\{T_{1},T_{2}\}$ be the time of first infection.  Let $I$ be the identity of the first infected subject.  The random variables $S$ has density 
\[ f(s) = \alpha (e^\beta + 1) \exp[-\alpha(e^\beta + 1) s] \]
and 
\[ \Pr(I=1) = e^\beta/(1+e^\beta) \]
Furthermore $S$ and $I$ are independent. By the law of iterated expectations, we expand 
\begin{equation*}
\begin{split}
 \E[Y_1(t)] &= \E_S [ \E_I [ Y_1(t) | S ]] \\
            &= \E_S \left[ \sum_{j\in\{1,2\}} \E[Y_1(t) | I=j, S=s] \Pr(I=j | S=s) \right] = \E_S \left[ \sum_{j\in\{1,2\}} \E[Y_1(t) | I=j, S=s] \Pr(I=j) \right] \\
            &= \E_S \left[ \Pr(I=1) + \E[Y_1(t) | I=2, S=s] \Pr(I=2) \right] \\
            &= \E_S \left[ \frac{e^\beta}{1+e^\beta} + \E[Y_1(t) | I=2, S=s] \frac{1}{1+e^\beta} \right] \\
\end{split}
\end{equation*}

In the above expectation with respect to $S$, it is implicit that $s<t$. The remaining inner expectation is 
\begin{equation*}
\begin{split}
 \E[Y_1(t) | I=2, S=s] &= \Pr(T_1<t | I=2, S=s) \\
                       &= \Pr(T_1<t | T_1>s) \\
                       &= 1 - \exp[-e^\beta (\alpha + \omega)(t - s)] 
 \end{split}
 \end{equation*}
 by the memoryless property of the exponential distribution. 
Putting these pieces together, 
\begin{equation*}
\begin{split}
   \E[Y_{1}(t)] &= \int_0^\infty \indicator {s<t} \left[ \frac {e^\beta}{1 + e^\beta} + \frac {1}{1 + e^\beta} (1 - \exp[-e^\beta (\alpha + \omega)(t - s)]) \right] \alpha (e^\beta + 1) \exp[-\alpha(e^\beta + 1) s ] \dx{s} \\
   & = \alpha \int_0^t \left[ e^\beta + 1 - \exp[-e^\beta (\alpha + \omega)(t - s)] \right] \exp[-\alpha(e^\beta + 1) s ] \dx{s} \\
   & = \alpha (e^\beta + 1)\int_0^t \exp[-\alpha(e^\beta + 1)s] \dx{s} - \alpha \exp[-e^\beta(\alpha+\omega)t] \int_0^{t} \exp[(e^\beta \omega - \alpha) s] \dx{s}  
   \end{split}
   \end{equation*}
 When $e^\beta \omega \neq \alpha$, 
   \begin{equation*}
   \begin{split}
  \E[Y_1(t)]  & = \frac{\alpha (e^\beta + 1)}{-\alpha(e^\beta + 1)} \Big[\exp[-\alpha(e^\beta + 1) t ] - 1 \Big] - \frac{\alpha}{e^\beta \omega - \alpha} \exp[-e^\beta (\alpha + \omega) t] \Big[\exp[(e^\beta \omega - \alpha) t ] - 1 \Big] \\
   &= 1 - \exp[-\alpha(e^\beta + 1) t ] - \frac{\alpha}{e^\beta \omega - \alpha} \exp[-\alpha(e^\beta + 1) t ] + \frac{\alpha}{e^\beta \omega - \alpha} \exp[-e^\beta (\alpha + \omega) t] \\
   & = \frac {e^\beta \omega}{\alpha - e^\beta \omega} \exp[-\alpha(e^\beta + 1) t ] - \frac {\alpha}{\alpha - e^\beta \omega} \exp[-e^\beta (\alpha + \omega) t] +1\\
  \end{split}
\end{equation*}
and when $e^\beta \omega = \alpha$, 
   \begin{equation*}
   \begin{split}
  \E[Y_1(t)]  & = 1 - \exp[-\alpha(e^\beta + 1) t ] - \alpha t \exp[-e^\beta (\alpha + \omega) t] \\
  & = 1 - \exp[-\alpha(e^\beta + 1) t ] (1+\alpha t) \\
  \end{split}
\end{equation*}
Similarly for $\E[Y_2(t)]$, if $\alpha e^\beta \neq \omega e^{\gamma}$, 
\begin{equation*}
  \E[Y_{2}(t)] = \frac {\omega e^\gamma}{\alpha e^\beta - \omega e^\gamma} \exp[-\alpha(e^\beta + 1) t ] -  \frac {\alpha e^\beta}{\alpha e^\beta - \omega e^\gamma} \exp[- (\alpha + \omega e^\gamma) t ] + 1 
   \end{equation*}
and if $\alpha e^\beta = \omega e^{\gamma}$, 
   \begin{equation*}
  \E[Y_2(t)]  = 1 - \exp[-\alpha(e^\beta + 1) t ] (1+\alpha e^\beta t) \\
\end{equation*}
Therefore the ratio of expectations is: 
\begin{equation}
 RR = 
 \begin{cases}
   \frac{\frac {e^\beta \omega}{\alpha - e^\beta \omega} \exp[-\alpha(e^\beta + 1) t ] - \frac {\alpha}{\alpha - e^\beta \omega} \exp[-e^\beta (\alpha + \omega) t] +1}{\frac {\omega e^\gamma}{\alpha e^\beta - \omega e^\gamma} \exp[-\alpha(e^\beta + 1) t ] -  \frac {\alpha e^\beta}{\alpha e^\beta - \omega e^\gamma} \exp[- (\alpha + \omega e^\gamma) t ] + 1},  & e^\beta \omega \neq \alpha,  \alpha e^\beta \neq \omega e^{\gamma} \\[1em]
 \frac{1 - \exp[-\alpha(e^\beta + 1) t ] (1+\alpha t)}{\frac {\omega e^\gamma}{\alpha e^\beta - \omega e^\gamma} \exp[-\alpha(e^\beta + 1) t ] -  \frac {\alpha e^\beta}{\alpha e^\beta - \omega e^\gamma} \exp[- (\alpha + \omega e^\gamma) t ] + 1},  &  e^\beta \omega = \alpha,  \alpha e^\beta \neq \omega e^{\gamma} \\[1em]
 \frac{\frac {e^\beta \omega}{\alpha - e^\beta \omega} \exp[-\alpha(e^\beta + 1) t ] - \frac {\alpha}{\alpha - e^\beta \omega} \exp[-e^\beta (\alpha + \omega) t] +1}{1 - \exp[-\alpha(e^\beta + 1) t ] (1+\alpha e^\beta t) },  & e^\beta \omega \neq \alpha,  \alpha e^\beta = \omega e^{\gamma} \\[1em]
 \frac{1 - \exp[-\alpha(e^\beta + 1) t ] (1+\alpha t)}{1 - \exp[-\alpha(e^\beta + 1) t ] (1+\alpha e^\beta t) },  & e^\beta \omega = \alpha,  \alpha e^\beta = \omega e^{\gamma} .
 \end{cases}
\label{eq:exactrr}
\end{equation}
In some of the proofs that follow, it will be useful to consider the risk difference $\E[Y_1]-\E[Y_2]$, whose sign is the same as that of the $RD^*$, where
\small
\noindent \begin{equation}
 RD^* = 
 \begin{cases}
   {\frac{\omega(e^{2 \beta} - e^\gamma) \exp[-\alpha(e^\beta+1)t] + (\omega e^\gamma - \alpha e^\beta) \exp[-e^\beta(\alpha + \omega)t] + e^\beta (\alpha - \omega e^\beta) \exp[-(\alpha + \omega e^\gamma)t] }{(\alpha - \omega e^\beta)(\alpha e^\beta - \omega e^\gamma)}}, & e^\beta \omega \neq \alpha,  \alpha e^\beta \neq \omega e^{\gamma} \\[1em]
 {\frac{e^\beta \exp[-(\alpha + \omega e^\gamma)t]  - (e^\beta + t (\alpha e^\beta - \omega e^\gamma)) \exp[-\alpha(e^\beta+1)t] }{\alpha e^\beta - \omega e^\gamma}}, & e^\beta \omega = \alpha,  \alpha e^\beta \neq \omega e^{\gamma} \\[1em] 
 { \frac{(1 + t e^\beta (\alpha - \omega e^\beta)) \exp[-\alpha(e^\beta+1)t] - \exp[-e^\beta(\alpha + \omega)t]  }{\alpha - \omega e^\beta}}, &  e^\beta \omega \neq \alpha,  \alpha e^\beta = \omega e^{\gamma} \\[1em]
 {t (e^\beta - 1)}, & e^\beta \omega = \alpha,  \alpha e^\beta = \omega e^{\gamma} . \\
 \end{cases}
\label{eq:rd}
\end{equation}

\section*{Proofs}  

\subsection*{Households of size two}

We first state and prove a simple Lemma that will ease exposition in what follows. 
\begin{lem}
  Suppose $0<a<b<c$.  Then 
  \[ (c-b) (e^{-a}-e^{-b}) - (b-a) (e^{-b}-e^{-c}) > 0  . \]
  \label{lem1}
\end{lem}
\begin{proof}
  Let $f(x)=e^{-x}$, so $f'(x)={\rm d}f(x)/{\rm d}x=-e^{-x}$. By the mean value theorem, there exist $x_1\in(a,b)$ and $x_2\in(b,c)$ such that 
  \[ f'(x_1) = -e^{-x_1} = \frac{e^{-b}-e^{-a}}{b-a} \quad\text{and}\quad  f'(x_2) = -e^{-x_2} = \frac{e^{-c}-e^{-b}}{c-b}. \]
  But since $x_1<x_2$, it follows that $-e^{-x_1}<-e^{-x_2}$ and so $f'(x_1)<f'(x_2)$.  Therefore 
  \[ \frac{e^{-b}-e^{-a}}{b-a} < \frac{e^{-c}-e^{-b}}{c-b}, \]
  and rearranging this inequality gives $(c-b) (e^{-a}-e^{-b}) - (b-a) (e^{-b}-e^{-c}) > 0$, as claimed.
\end{proof}

\subsubsection*{Result \ref{prop:biasnoclustercont}: No within-cluster contagion}

\begin{proof}
  Suppose $\alpha>0$ and $\omega=0$. We only need to consider the first case in \eqref{eq:rd}, and the sign of this expression at time $T$ is the same as that of
  \[ \exp[-\alpha T] - \exp[-\alpha e^\beta T] . \]
Since $\alpha$ and $T$ are non-negative, the risk ratio is less than one for every $t\in(0,T]$ when $\beta<0$, one when $\beta=0$, and greater than one when $\beta>0$. Therefore the risk ratio is direction-unbiased. 
\end{proof}

\subsubsection*{Result \ref{prop:biasundernull}: Under the null}

\begin{proof}
  Suppose $\beta=0$. The sign of \eqref{eq:rd} is the same as the sign of $RD_{\beta=0}^*$, where 
\noindent \begin{equation}
 RD_{\beta=0}^* = 
 \begin{cases}
   {\frac{\omega(1 - e^\gamma) \exp[-\alpha t] + (\omega e^\gamma - \alpha) \exp[- \omega t] + (\alpha - \omega) \exp[- \omega e^\gamma t] }{(\alpha - \omega)(\alpha - \omega e^\gamma)}}, & \alpha \neq \omega,  \alpha \neq \omega e^{\gamma} \\[1em]
   {\frac{\exp[- \alpha e^\gamma t]  - (1 + t \alpha (1 - e^\gamma)) \exp[-\alpha t] }{\alpha(1 - e^\gamma)}}, & \alpha = \omega,  \alpha \neq \omega e^{\gamma} \\[1em]
   { \frac{(1 + t (\omega e^\gamma - \omega)) \exp[-\omega e^\gamma t] - \exp[- \omega t]  }{\omega (e^\gamma - 1)}}, &  \alpha \neq \omega,  \alpha = \omega e^{\gamma} \\[1em]
   {0}, & \alpha = \omega,  \alpha = \omega e^{\gamma} . \\[1em]
 \end{cases}
\label{eq:rdnull}
\end{equation}
First, note that when $\gamma=0$, $RD_{\beta=0}^*=0$, so $RR=1$.  

Now suppose $\gamma \neq 0$.  The proof is divided into cases for $\gamma < 0$ and $\gamma > 0$. These cases are further divided into several sub-cases defined by the relationship between the parameters of the model.\\[1em] 

\noindent \textbf{Case 1:} Let $\gamma < 0$.  We will show that for any $t>0$, expression in \eqref{eq:rdnull} is positive, and hence $RR>1$.\\
\noindent\textbf{Sub-case 1.1:} Suppose $0 < \alpha < \omega e^\gamma < \omega$. The denominator of \eqref{eq:rdnull} is positive, and the expressions in the numerator have the following signs:
\[  \omega(1 - e^\gamma) > 0, \quad \omega e^\gamma - \alpha > 0, \quad \text{and} \quad \alpha - \omega < 0 . \]
Multiplying the numerator of \eqref{eq:rdnull} by $t>0$ gives the following expression:
\[ (\omega t - \omega e^\gamma t) \exp[-\alpha t] + (\omega e^\gamma t - \alpha t) \exp[-\omega t] - (\omega t - \alpha t) \exp[-\omega e^\gamma t] .\]
Splitting $\omega t - \alpha t$ into $(\omega t - \omega e^\gamma t) + (\omega e^\gamma t - \alpha t)$ and rearranging, the numerator of \eqref{eq:rdnull} equals:
\[ (\omega t - \omega e^\gamma t) \left( \exp[-\alpha t] - \exp[-\omega e^\gamma t]\right) - (\omega e^\gamma t - \alpha t) \left(\exp[-\omega e^\gamma t] - \exp[-\omega t] \right) \]
Let $a=\alpha t$, $b=\omega e^\gamma t$, and $c=\omega t$. By Lemma \ref{lem1}, the numerator of \eqref{eq:rdnull} is positive for any $t>0$, so $RR>1$. \\[1em]

\noindent\textbf{Sub-case 1.2:} Suppose $0 < \omega e^\gamma < \alpha < \omega$. The denominator of \eqref{eq:rdnull} is negative, and the expressions in the numerator have the following signs:
\[  \omega(1 - e^\gamma) > 0, \quad \omega e^\gamma - \alpha < 0, \quad \text{and} \quad \alpha - \omega < 0 . \]
Multiplying the numerator of \eqref{eq:rdnull} by $t>0$ and rearranging gives the following expression:
\[ (\alpha t - \omega e^\gamma t) \left( \exp[-\alpha t] - \exp[-\omega t]\right) - (\omega t - \alpha t) \left(\exp[-\omega e^\gamma t] - \exp[-\alpha t] \right) \]
By Lemma \ref{lem1}, the numerator of \eqref{eq:rdnull} is negative for any $t>0$, so $RR>1$.\\

\noindent\textbf{Sub-case 1.3:} Suppose $0 < \omega e^\gamma < \omega < \alpha$. The denominator of \eqref{eq:rdnull} is positive, and the expressions in the numerator have the following signs:
\[  \omega(1 - e^\gamma) > 0, \quad \omega e^\gamma - \alpha < 0, \quad \text{and} \quad \alpha - \omega > 0 . \]
Multiplying the numerator of \eqref{eq:rdnull} by $t>0$ and rearranging gives the following expression:
\[ (\alpha t - \omega t) (\exp[-\omega e^\gamma t] - \exp[-\omega t]) - (\omega t - \omega e^\gamma t) (\exp[-\omega t] - \exp[-\alpha t] ) \]
By Lemma \ref{lem1}, numerator of \eqref{eq:rdnull} is positive for any $t>0$, so $RR>1$.\\ 

\noindent\textbf{Sub-case 1.4:} Suppose $\alpha - \omega = 0$. Since $e^\gamma < 1$, the denominator of \eqref{eq:rdnull} is positive. Dividing the numerator by $\exp[-\alpha t]$ and rearranging gives: 
\[ \exp[\alpha t (1-e^\gamma)] - (1+ \alpha t (1-e^\gamma)),\]
which is positive for any $t>0$, since $\exp[a] > 1+a$ for $a \ne 0$, so $RR>1$.\\

\noindent\textbf{Sub-case 1.5:} Suppose $\alpha - \omega e^\gamma = 0$. Since $e^\gamma < 1$, the denominator of \eqref{eq:rdnull} is negative. Dividing the numerator by $\exp[-\omega e^\gamma t]$ and rearranging gives: 
\[ (1+ \omega t (e^\gamma - 1)) - \exp[\omega t (e^\gamma - 1)] ,\]
which is negative for any $t>0$, since $\exp[a] > 1+a$ for $a \ne 0$, so $RR>1$. \\ [1em]

\noindent \textbf{Case 2:} Let $\gamma > 0$. We will show that for any $t>0$, the expression in \eqref{eq:rdnull} is negative, and hence $RR<1$. 

\noindent\textbf{Sub-case 2.1:} Suppose $0 < \alpha < \omega < \omega e^\gamma$. The denominator of \eqref{eq:rdnull} is positive, and the expressions in the numerator have the following signs:
\[  \omega(1 - e^\gamma) < 0, \quad \omega e^\gamma - \alpha > 0, \quad \text{and} \quad \alpha - \omega < 0 . \]
Multiplying the numerator of \eqref{eq:rdnull} by $t>0$ and rearranging gives the following expression:
\[ (\omega t - \alpha t) (\exp[-\omega t] - \exp[-\omega e^\gamma t]) - (\omega e^\gamma t - \omega t) (\exp[-\alpha t] - \exp[-\omega t] ) \]
By Lemma \ref{lem1}, the numerator of \eqref{eq:rdnull} is negative for any $t>0$, so $RR<1$.\\ 

\noindent\textbf{Sub-case 2.2:} Suppose $0 < \omega < \alpha < \omega e^\gamma$. The denominator of \eqref{eq:rdnull} is negative, and the expressions in the numerator have the following signs:
\[  \omega(1 - e^\gamma) < 0, \quad \omega e^\gamma - \alpha > 0, \quad \text{and} \quad \alpha - \omega > 0 . \]
Multiplying the numerator of \eqref{eq:rdnull} by $t>0$ and rearranging gives the following expression:
\[ (\omega e^\gamma t - \alpha t) (\exp[-\omega t] - \exp[-\alpha t]) - (\alpha t - \omega t) (\exp[-\alpha t] - \exp[-\omega e^\gamma t] ) \]
By Lemma \ref{lem1}, the numerator of \eqref{eq:rdnull} is positive for any $t>0$, so $RR<1$. \\ 

\noindent\textbf{Sub-case 2.3:} Suppose $0 < \omega < \omega e^\gamma < \alpha$. The denominator of \eqref{eq:rdnull} is positive, and the expressions in the numerator have the following signs:
\[  \omega(1 - e^\gamma) < 0, \quad \omega e^\gamma - \alpha < 0, \quad \text{and} \quad \alpha - \omega > 0 . \]
Multiplying the numerator of \eqref{eq:rdnull} by $t>0$ and rearranging gives the following expression:
\[ (\omega e^\gamma t - \omega t) (\exp[-\omega e^\gamma t] - \exp[-\alpha t]) - (\alpha t - \omega e^\gamma t) (\exp[-\omega t] - \exp[-\omega e^\gamma t] ) \]
By Lemma \ref{lem1}, the numerator of \eqref{eq:rdnull} is negative for any $t>0$, so $RR<1$.\\ 

\noindent\textbf{Sub-case 2.4:} Suppose $\alpha - \omega = 0$. Since $e^\gamma > 1$, the denominator of \eqref{eq:rdnull} is negative. Dividing the numerator by $\exp[-\alpha t]$ and rearranging gives: 
\[ \exp[\alpha t (1-e^\gamma)] - (1+ \alpha t (1-e^\gamma)),\]
which is positive for any $t>0$, since $\exp[a] > 1+a$ for $a \ne 0$, so $RR<1$.\\

\noindent\textbf{Sub-case 2.5:} Suppose $\alpha - \omega e^\gamma = 0$. Since $e^\gamma > 1$, the denominator of \eqref{eq:rdnull} is positive. Dividing the numerator by $\exp[-\omega e^\gamma t]$ and rearranging gives: 
\[ (1+ \omega t (e^\gamma - 1)) - \exp[\omega t (e^\gamma - 1)] ,\]
which is negative for any $t>0$, since $\exp[a] > 1+a$ for $a \ne 0$, so $RR<1$. 
\end{proof}

\subsubsection*{Result \ref{prop:biasnoinfect}: Homogeneous infectiousness}
\small
\begin{proof}
  Suppose $\gamma=0$. The sign of \eqref{eq:rd} is the same as the sign of $RD_{\gamma=0}^*$, where 
\begin{equation}
 RD_{\gamma=0}^* = 
 \begin{cases}
   {\frac{\omega(e^{2 \beta} - 1) \exp[-\alpha(e^\beta+1)t] + (\omega - \alpha e^\beta) \exp[-e^\beta(\alpha + \omega)t] + e^\beta (\alpha - \omega e^\beta) \exp[-(\alpha + \omega)t] }{(\alpha - \omega e^\beta)(\alpha e^\beta - \omega)}}, & e^\beta \omega \neq \alpha,  \alpha e^\beta \neq \omega \\[1em]
  {\frac{e^\beta \exp[- \omega t]  - (e^\beta + t \omega (e^{2 \beta} - 1)) \exp[- \omega e^{2\beta} t] }{\omega (e^{2\beta} - 1)}}, & e^\beta \omega = \alpha,  \alpha e^\beta \neq \omega \\[1em] 
  { \frac{(1 + t \alpha e^\beta (1 - e^{2\beta})) \exp[-\alpha t] - \exp[-\alpha e^{2\beta} t]  }{\alpha (1 - e^{2\beta})}}, & e^\beta \omega \neq \alpha,  \alpha e^\beta = \omega \\[1em]
 {t (e^\beta - 1)}, & e^\beta \omega = \alpha,  \alpha e^\beta = \omega . \\
 \end{cases}
\label{eq:rdnoinf}
\end{equation}
\normalsize
First note that when $\gamma=0$, $RD_{\beta=0}^*=0$, so $RR=1$.  

Now suppose $\beta \neq 0$.
The proof is divided into cases for $\beta < 0$ and $\beta > 0$. These cases are further divided into several sub-cases defined by the relationship between the parameters of the model.\\[1em] 

\noindent \textbf{Case 1:} Suppose $\beta < 0$. We will show that for any $t>0$, expression in \eqref{eq:rdnoinf} is negative.\\
\noindent\textbf{Sub-case 1.1:} Suppose $0 < \alpha < \omega e^\beta$. It follows from this condition that $\alpha e^\beta < \omega e^{2\beta} < \omega$, $\alpha e^\beta < \omega$ and $\exp[-e^\beta(\alpha + \omega)t] < \exp[-(\alpha + \omega e^{2\beta})t]$. The denominator of \eqref{eq:rdnoinf} is positive, and the expressions in the numerator have the following signs:
\[  \omega(e^{2 \beta} - 1) < 0, \quad \omega - \alpha e^\beta > 0, \quad \text{and} \quad e^\beta(\alpha - \omega e^\beta) < 0 . \]
The numerator of \eqref{eq:rdnoinf} is less than
	\begin{equation}
	\omega(e^{2 \beta} - 1) \exp[-\alpha(e^\beta+1)t] + (\omega - \alpha e^\beta) \exp[-(\alpha + \omega e^{2\beta})t] + e^\beta (\alpha - \omega e^\beta) \exp[-(\alpha + \omega)t]
	\label{eq:ni1}
	\end{equation}
which has the same sign as 
	\begin{equation}
	\omega(e^{2 \beta} - 1) \exp[-\alpha e^\beta t] + (\omega - \alpha e^\beta) \exp[- \omega e^{2\beta} t] + e^\beta (\alpha - \omega e^\beta) \exp[- \omega t]. 
	\label{eq:ni2}
	\end{equation}
Multiplying \eqref{eq:ni2} by $t>0$ and rearranging gives the following expression:
\[ (\omega e^{2\beta} t - \alpha e^\beta t) (\exp[-\omega e^{2\beta} t] - \exp[-\omega t] ) - (\omega t - \omega e^{2\beta} t) (\exp[-\alpha e^\beta t] - \exp[-\omega e^{2\beta} t] ) . \]
By Lemma \ref{lem1}, \eqref{eq:ni2} is negative for any $t>0$, so $RR<1$.\\

\noindent\textbf{Sub-case 1.2:} Suppose $0 < \omega e^{2\beta} < \alpha e^\beta < \omega$. It follows from this condition that $\alpha e^\beta < \omega$, $\omega e^\beta < \alpha$, and $\exp[-e^\beta(\alpha + \omega)t] > \exp[-(\alpha + \omega e^{2\beta})t]$. The denominator of \eqref{eq:rdnoinf} is negative, and the expressions in the numerator have the following signs:
\[  \omega(e^{2 \beta} - 1) < 0, \quad \omega - \alpha e^\beta > 0, \quad \text{and} \quad e^\beta(\alpha - \omega e^\beta) > 0 . \]
The numerator of \eqref{eq:rdnoinf} is greater than \eqref{eq:ni1}, which has the same sign as \eqref{eq:ni2}. Multiplying \eqref{eq:ni2} by $t>0$ and rearranging gives the following expression:
\[ (\omega t - \alpha e^\beta t) (\exp[-\omega e^{2\beta} t] - \exp[-\alpha e^\beta t] ) - (\alpha e^\beta t - \omega e^{2\beta} t) (\exp[-\alpha e^\beta t] - \exp[-\omega t] ) . \]
By Lemma \ref{lem1}, \eqref{eq:ni2} is positive for any $t>0$, so $RR<1$.\\

\noindent\textbf{Sub-case 1.3:} Suppose $0 < \omega < \alpha e^\beta$. It follows from this condition that $\omega e^{2\beta} < \omega < \alpha e^\beta$, $\omega e^\beta < \alpha$ and $\exp[-e^\beta(\alpha + \omega)t] > \exp[-(\alpha + \omega e^{2\beta})t]$. The denominator of \eqref{eq:rdnoinf} is positive, and the expressions in the numerator have the following signs:
\[  \omega(e^{2 \beta} - 1) < 0, \quad \omega - \alpha e^\beta < 0, \quad \text{and} \quad e^\beta(\alpha - \omega e^\beta) > 0 . \]
The numerator of \eqref{eq:rdnoinf} is less than \eqref{eq:ni1}, which has the same sign as \eqref{eq:ni2}. Multiplying \eqref{eq:ni2} by $t>0$ and rearranging gives the following expression:
\[ (\omega t - \omega e^{2\beta} t) (\exp[-\omega t] - \exp[-\alpha e^\beta t] ) - (\alpha e^\beta t - \omega t) (\exp[-\omega e^{2\beta} t] - \exp[-\omega t] ) . \]
By Lemma \ref{lem1}, \eqref{eq:ni2} is negative for any $t>0$, so $RR<1$.\\

\noindent\textbf{Sub-case 1.4:} Suppose $\alpha - \omega e^\beta = 0$ and $\alpha e^\beta - \omega \ne 0$. Since $e^\beta < 1$, the denominator of \eqref{eq:rdnoinf} is negative.  Dividing the numerator by $e^\beta \exp[- \omega e^{2\beta} t]$ and rearranging gives:
\[ \exp[\omega t (e^{2\beta}-1) ] -(1 + \omega \frac{t}{e^\beta} (e^{2\beta}-1)) > \exp[\omega t (e^{2\beta}-1) ] -(1 + \omega t (e^{2\beta}-1)) .\]
The right hand side of this expression is positive for any $t>0$, since $\exp[a] > 1+a$ for $a \ne 0$, so $RR<1$.\\

\noindent\textbf{Sub-case 1.5:} Suppose $\alpha - \omega e^\beta \ne 0$ and $\alpha e^\beta - \omega = 0$. Since $e^\beta < 1$, the denominator of \eqref{eq:rdnoinf} is positive. 
Dividing the numerator by $\exp[- \alpha t]$ and rearranging gives:
\[ 1 + \alpha e^\beta t (1 - e^{2\beta}) -\exp[\alpha t (1 - e^{2\beta})] < 1 + \alpha t (1 - e^{2\beta}) -\exp[\alpha t (1 - e^{2\beta})]  .\]
The right hand side of this expression is negative for any $t>0$, since $\exp[a] > 1+a$ for $a \ne 0$, so $RR<1$.\\

\noindent\textbf{Sub-case 1.6:} Suppose $\alpha - \omega e^\beta = 0$ and $\alpha e^\beta - \omega = 0$. Since $e^\beta < 1$, \eqref{eq:rdnoinf} is negative for any $t>0$, so $RR<1$. \\[1em]

\noindent \textbf{Case 2:} Suppose $\beta > 0$.  We will show that for any $t>0$, expression in \eqref{eq:rdnoinf} is positive, and hence $RR>1$.\\
\noindent\textbf{Sub-case 2.1:} Suppose $0 < \alpha e^\beta < \omega$. It follows from this condition that $\alpha e^\beta < \omega < \omega e^{2\beta}$, $\alpha < \omega e^\beta$ and $\exp[-e^\beta(\alpha + \omega)t] > \exp[-(\alpha + \omega e^{2\beta})t]$. The denominator of \eqref{eq:rdnoinf} is positive, and the expressions in the numerator have the following signs:
\[  \omega(e^{2 \beta} - 1) > 0, \quad \omega - \alpha e^\beta > 0, \quad \text{and} \quad e^\beta(\alpha - \omega e^\beta) < 0 . \]
The numerator of \eqref{eq:rdnoinf} is greater than \eqref{eq:ni1}, which has the same sign as \eqref{eq:ni2}. Multiplying \eqref{eq:ni2} by $t>0$ and rearranging gives the following expression:
\[ (\omega e^{2\beta} t - \omega t) (\exp[-\alpha e^\beta t] - \exp[-\omega t] ) - (\omega t - \alpha e^\beta t) (\exp[-\omega t] - \exp[-\omega e^{2\beta} t] ) . \]
By Lemma \ref{lem1}, \eqref{eq:ni2} is positive for any $t>0$, so $RR>1$.\\

\noindent\textbf{Sub-case 2.2:} Suppose $0 < \omega < \alpha e^\beta < \omega e^{2\beta}$. It follows from this condition that $\alpha < \omega e^\beta$, $\alpha e^\beta > \omega$ and $\exp[-e^\beta(\alpha + \omega)t] > \exp[-(\alpha + \omega e^{2\beta})t]$. The denominator of \eqref{eq:rdnoinf} is negative, and the expressions in the numerator have the following signs:
\[  \omega(e^{2 \beta} - 1) > 0, \quad \omega - \alpha e^\beta < 0, \quad \text{and} \quad e^\beta(\alpha - \omega e^\beta) < 0 . \]
The numerator of \eqref{eq:rdnoinf} is less than \eqref{eq:ni1}, which has the same sign as \eqref{eq:ni2}. Multiplying \eqref{eq:ni2} by $t>0$ and rearranging gives the following expression:
\[ (\alpha e^\beta t - \omega t) (\exp[-\alpha e^\beta t] - \exp[-\omega e^{2\beta} t] ) - (\omega e^{2\beta}  t - \alpha e^\beta t) (\exp[-\omega t] - \exp[-\alpha e^\beta t] ) . \]
By Lemma \ref{lem1}, \eqref{eq:ni2} is negative for any $t>0$, so $RR>1$.\\

\noindent\textbf{Sub-case 2.3:} Suppose $0 < \omega e^\beta < \alpha$. It follows from this condition that $\omega < \omega e^{2\beta} < \alpha e^\beta$, $\alpha e^\beta > \omega$ and $\exp[-e^\beta(\alpha + \omega)t] < \exp[-(\alpha + \omega e^{2\beta})t]$. The denominator of \eqref{eq:rdnoinf} is positive, and the expressions in the numerator have the following signs:
\[  \omega(e^{2 \beta} - 1) > 0, \quad \omega - \alpha e^\beta < 0, \quad \text{and} \quad e^\beta(\alpha - \omega e^\beta) > 0 . \]
The numerator of \eqref{eq:rdnoinf} is greater than \eqref{eq:ni1}, which has the same sign as \eqref{eq:ni2}. Multiplying \eqref{eq:ni2} by $t>0$ and rearranging gives the following expression:
\[ (\alpha e^\beta t - \omega e^{2\beta} t) (\exp[-\omega t] - \exp[-\omega e^{2\beta} t] ) - (\omega e^{2\beta} t - \omega t) (\exp[-\omega e^{2\beta} t] - \exp[-\alpha e^\beta t] ) . \]
By Lemma \ref{lem1}, \eqref{eq:ni2} is positive for any $t>0$, so $RR<1$.\\

\noindent\textbf{Sub-case 2.4:} Suppose $\alpha - \omega e^\beta = 0$ and $\alpha e^\beta - \omega \ne 0$. Since $e^\beta > 1$, the denominator of \eqref{eq:rdnoinf} is positive.  Dividing the numerator by $e^\beta \exp[- \omega e^{2\beta} t]$ and rearranging gives:
\[ \exp[\omega t (e^{2\beta}-1) ] -(1 + \omega \frac{t}{e^\beta} (e^{2\beta}-1)) > \exp[\omega t (e^{2\beta}-1) ] -(1 + \omega t (e^{2\beta}-1)) .\]
The right hand side of this expression is positive for any $t>0$, since $\exp[a] > 1+a$ for $a \ne 0$, so $RR<1$.\\

\noindent\textbf{Sub-case 2.5:} Suppose $\alpha - \omega e^\beta \ne 0$ and $\alpha e^\beta - \omega = 0$. Since $e^\beta > 1$, the denominator of \eqref{eq:rdnoinf} is negative.  Dividing the numerator by $\exp[- \alpha t]$ and rearranging gives:
\[ 1 + \alpha e^\beta t (1 - e^{2\beta}) -\exp[\alpha t (1 - e^{2\beta})] < 1 + \alpha t (1 - e^{2\beta}) -\exp[\alpha t (1 - e^{2\beta})]  .\]
The right hand side of this expression is negative for any $t>0$, since $\exp[a] > 1+a$ for $a \ne 0$, so $RR<1$.\\

\noindent\textbf{Sub-case 2.6:} Suppose $\alpha - \omega e^\beta = 0$ and $\alpha e^\beta - \omega = 0$. Since $e^\beta > 1$, \eqref{eq:rdnoinf} is positive for any $t>0$, so $RR<1$.  
\end{proof}

\subsubsection*{Result \ref{prop:biasacrossnull}: Direction bias}
\begin{proof}
The proof is divided into cases for $\beta < 0$ and $\beta > 0$. These cases are further divided into several sub-cases defined by the relationship between the parameters of the model. \\

\noindent \textbf{Case 1:} Suppose $\beta<0$ and $e^\gamma < \min\{e^{2\beta} , e^\beta + \frac{\alpha}{\omega}(e^\beta - 1)\}$. It follows that $e^\beta > \alpha / (\alpha+\omega)$, and that the equalities $e^\beta \omega = \alpha$ and $\alpha e^\beta = \omega e^\gamma$ cannot hold simultaneously. 
We will show that for every combination of other parameters there exists $t^* > 0$, such that for all $t > t^*$, the corresponding expression in \eqref{eq:rd} is positive, and hence $RR>1$.

\noindent\textbf{Sub-case 1.1:} Suppose $\alpha / \omega < e^\beta < 1$, 
$\alpha e^\beta \ne \omega e^\gamma$ and consider the first line of \eqref{eq:rd}.  
When $\alpha e^\beta - \omega e^\gamma < 0$, the denominator 
is positive, and the expressions in the numerator have the following signs:
\[  \omega(e^{2 \beta} - e^\gamma) > 0, \quad \omega e^\gamma - \alpha e^\beta > 0, \quad \text{and} \quad e^\beta (\alpha - \omega e^\beta) < 0 . \]
Therefore for any $t>0$ the numerator of \eqref{eq:rd} is greater than 
\begin{equation}
 \omega(e^{2 \beta} - e^\gamma) \exp[-\alpha(e^\beta+1)t] + e^\beta (\alpha - \omega e^\beta) \exp[-(\alpha + \omega e^\gamma)t] 
 \label{eq:case111a}
\end{equation}
which has positive sign whenever
\begin{equation}
 \omega(e^{2 \beta} - e^\gamma) \exp[-\alpha(e^\beta+1)t] > - e^\beta (\alpha - \omega e^\beta) \exp[-(\alpha + \omega e^\gamma)t].
 \label{eq:case111b}
\end{equation}
This inequality holds for any $t > \frac{\log[e^\beta(\omega e^\beta - \alpha)] - \log[\omega(e^{2 \beta} - e^\gamma)]}{\omega e^\gamma - \alpha e^\beta}$. 
Note that this threshold for $t$ is positive and finite.

When $\alpha e^\beta - \omega e^\gamma > 0$, the denominator of \eqref{eq:rd} is negative, and the coefficients in the numerator have the following signs:
\[ \omega(e^{2 \beta} - e^\gamma) > 0, \quad \omega e^\gamma - \alpha e^\beta < 0, \quad \text{and}\quad e^\beta (\alpha - \omega e^\beta) < 0 . \]
Therefore for any $t>0$ the numerator of \eqref{eq:rd} is less than 
\begin{equation}
  \omega(e^{2 \beta} - e^\gamma) \exp[-\alpha(e^\beta+1)t] + e^\beta (\alpha - \omega e^\beta) \exp[-(\alpha + \omega e^\gamma)t] 
 \label{eq:case112a}
\end{equation}
which is negative whenever
\begin{equation}
 \omega(e^{2 \beta} - e^\gamma) \exp[-\alpha(e^\beta+1)t]  < - e^\beta (\alpha - \omega e^\beta) \exp[-(\alpha + \omega e^\gamma)t].
 \label{eq:case112b}
\end{equation}
This inequality holds for any $t > \frac{\log[\omega(e^{2 \beta} - e^\gamma)] - \log[e^\beta(\omega e^\beta - \alpha)]}{\alpha e^\beta - \omega e^\gamma}$. Note that this threshold for $t$ is positive and finite.

\noindent \textbf{Sub-case 1.2:} Suppose $e^\beta < \alpha / \omega$
and $\alpha e^\beta \ne \omega e^\gamma$. It follows that $\alpha e^\beta - \omega e^\gamma > 0$, the denominator of \eqref{eq:rd} is positive, and the expressions in the numerator have the following signs:
\[ \omega(e^{2 \beta} - e^\gamma) > 0,  \quad \omega e^\gamma - \alpha e^\beta < 0, \quad\text{and}\quad e^\beta (\alpha - \omega e^\beta) > 0 .\]
Therefore for any $t>0$, the numerator of \eqref{eq:rd} is greater than 
\begin{equation}
 (\omega e^\gamma - \alpha e^\beta) \exp[-e^\beta(\alpha + \omega)t] + e^\beta (\alpha - \omega e^\beta) \exp[-(\alpha + \omega e^\gamma)t] ,
 \label{eq:case12a}
\end{equation}
which is positive whenever
\begin{equation}
 (\omega e^\gamma - \alpha e^\beta) \exp[-e^\beta(\alpha + \omega)t] > - e^\beta (\alpha - \omega e^\beta) \exp[-(\alpha + \omega e^\gamma)t].
 \label{eq:case12b}
\end{equation}
This inequality holds for any $t > \frac{\log[\alpha e^\beta - \omega e^\gamma] - \log[e^\beta(\alpha - \omega e^ \beta)]}{e^\beta(\alpha + \omega) - (\alpha + \omega e^\gamma) }$. Note that this threshold for $t$ is positive and finite.

\noindent \textbf{Sub-case 1.3:}
Suppose that $\alpha = \omega e^\beta$ and $\alpha e^\beta \ne \omega e^\gamma$.  It follows that $\alpha / \omega <1$ and $\alpha e^\beta - \omega e^\gamma > 0$.  The denominator of \eqref{eq:rd} is positive, and the expressions in the numerator have the following signs:
\[ e^\beta > 0 \quad\text{and}\quad e^\beta + t (\alpha e^\beta - \omega e^\gamma) > 0. \]
Therefore \eqref{eq:rd} has positive sign whenever
\begin{equation}
 e^\beta \exp[-(\alpha + \omega e^\gamma)t]  > (e^\beta + t (\alpha e^\beta - \omega e^\gamma)) \exp[-\alpha(e^\beta+1)t].
 \label{eq:case13}
\end{equation}
Since $t > 0$ and $a>0$, $\log(1+at)$ is a monotonic function of $t$ that grows more slowly than $t$. Therefore there exists $t^* > 0$ such that for any $t > t^*$, $t > \frac{\log(1+\frac{t}{e^\beta} (\alpha e^\beta - \omega e^\gamma))}{\alpha e^\beta - \omega e^\gamma}$. Therefore \eqref{eq:case13} holds for $t > t^*$.

\noindent \textbf{Sub-case 1.4:} Suppose $\alpha \ne \omega e^\beta$ and $\alpha e^\beta = \omega e^\gamma$. It follows that $\alpha / \omega <1$, and $\alpha - \omega e^\beta < 0$.
The denominator of \eqref{eq:rd} is negative, and \eqref{eq:rd} is positive when 
$1 + t e^\beta (\alpha - \omega e^\beta) < 0$. 
This inequality holds for any $t > [e^\beta (\omega e^\beta - \alpha)]^{-1}$. Note that this threshold value for $t$ is positive and finite. \\[1em]

\noindent \textbf{Case 2:} Suppose $\beta > 0$ and $e^\gamma > \max \{e^{2\beta}, e^\beta + \frac{\alpha}{\omega}(e^\beta - 1)\}$. It follows that the equalities $e^\beta \omega = \alpha$ and $\alpha e^\beta = \omega e^\gamma$ cannot hold simultaneously. We will show that for every combination of other parameters there exists $t^* > 0$, such that for all $t > t^*$, the corresponding expression in \eqref{eq:rd} is negative.

\noindent \textbf{Sub-case 2.1:} Suppose $1 < e^\beta < \alpha / \omega$ 
and $\alpha e^\beta \ne \omega e^\gamma$.  
When $\alpha e^\beta - \omega e^\gamma < 0$, the denominator of \eqref{eq:rd} is negative, and the coefficients in the numerator have the following signs:
\[ \omega(e^{2 \beta} - e^\gamma) < 0, \quad \omega e^\gamma - \alpha e^\beta > 0, \quad\text{and}\quad e^\beta (\alpha - \omega e^\beta) > 0 . \]
Therefore for any $t>0$ the numerator of \eqref{eq:rd} is greater than 
\begin{equation}
 \omega(e^{2 \beta} - e^\gamma) \exp[-\alpha(e^\beta+1)t] + (\omega e^\gamma - \alpha e^\beta) \exp[-e^\beta(\alpha + \omega)t] ,
 \label{eq:case211a}
\end{equation}
which has positive sign whenever
\begin{equation}
 (\omega e^\gamma - \alpha e^\beta) \exp[-e^\beta(\alpha + \omega)t] > - \omega(e^{2 \beta} - e^\gamma) \exp[-\alpha(e^\beta+1)t].
 \label{eq:case211b}
\end{equation}
This inequality holds for any $t > \frac{\log[\omega(e^\gamma - e^{2 \beta})] - \log[\omega e^\gamma - \alpha e^\beta]}{\alpha - \omega e^\beta}$. Note that this threshold for $t$ is positive and finite.

When $\alpha e^\beta - \omega e^\gamma > 0$, the denominator of \eqref{eq:rd} is positive, and the coefficients in the numerator have the following signs:
\[ \omega(e^{2 \beta} - e^\gamma) < 0, \quad \omega e^\gamma - \alpha e^\beta < 0, \quad\text{and}\quad e^\beta (\alpha - \omega e^\beta) > 0. \]
Therefore for any $t>0$ the numerator of \eqref{eq:rd} is less than 
\begin{equation}
(\omega e^\gamma - \alpha e^\beta) \exp[-e^\beta(\alpha + \omega)t] + e^\beta (\alpha - \omega e^\beta) \exp[-(\alpha + \omega e^\gamma)t]
 \label{eq:case212a}
\end{equation}
which is negative whenever
\begin{equation}
 e^\beta (\alpha - \omega e^\beta) \exp[-(\alpha + \omega e^\gamma)t]  < - (\omega e^\gamma - \alpha e^\beta) \exp[-e^\beta(\alpha + \omega)t].
 \label{eq:case212b}
\end{equation}
This inequality holds for any $t > \frac{\log[e^\beta(\alpha - \omega e^\beta)] - \log[\alpha e^\beta - \omega e^\gamma]}{(\alpha - \omega e^\beta) - (\alpha e^\beta - \omega e^\gamma)}$. Note that this threshold for $t$ is positive and finite. \\

\noindent \textbf{Sub-case 2.2:} Suppose $e^\beta > \alpha / \omega$
and $\alpha e^\beta \ne \omega e^\gamma$. It follows that $\alpha e^\beta - \omega e^\gamma < 0$.  The denominator of \eqref{eq:rd} is positive, and the expressions in the numerator have the following signs:
\[ \omega(e^{2 \beta} - e^\gamma) < 0, \quad \omega e^\gamma - \alpha e^\beta > 0, \quad\text{and}\quad e^\beta (\alpha - \omega e^\beta) < 0. \]
Therefore for any $t>0$ the numerator of \eqref{eq:rd} is less than 
\begin{equation}
 \omega(e^{2 \beta} - e^\gamma) \exp[-\alpha(e^\beta+1)t] + (\omega e^\gamma - \alpha e^\beta) \exp[-e^\beta(\alpha + \omega)t] 
 \label{eq:case22a}
\end{equation}
which is negative whenever
\begin{equation}
 (\omega e^\gamma - \alpha e^\beta) \exp[-e^\beta(\alpha + \omega)t] < - \omega(e^{2 \beta} - e^\gamma) \exp[-\alpha(e^\beta+1)t].
 \label{eq:case22b}
\end{equation}
This inequality holds for any $t > \frac{\log[\omega e^\gamma - \alpha e^\beta] - \log[\omega (e^\gamma - e^{2 \beta})]}{\omega e^\beta - \alpha }$. Note that this threshold for $t$ is positive and finite. \\

\noindent \textbf{Sub-case 2.3:} Suppose $\alpha = \omega e^\beta$ and $\alpha e^\beta \ne \omega e^\gamma$.  It follows that $\alpha / \omega >1$ and $(\alpha e^\beta - \omega e^\gamma) < 0$.  The denominator of \eqref{eq:rd} is negative, and \eqref{eq:rd} is negative when 
$ e^\beta + t (\alpha e^\beta - \omega e^\gamma) < 0 $. 
This inequality holds for any $t > \frac{e^\beta}{\omega e^\gamma - \alpha e^\beta}$. Note that this threshold for $t$ is positive and finite.

\noindent \textbf{Sub-case 2.4:} Suppose that $\alpha \ne \omega e^\beta$ and $\alpha e^\beta = \omega e^\gamma$. It follows from this condition that $\alpha / \omega >1$, and $\alpha - \omega e^\beta > 0$.  The denominator of \eqref{eq:rd} is positive, and the expression in the numerator has the following sign:
\[ 1 + t e^\beta (\alpha - \omega e^\beta) > 0\]
Therefore \eqref{eq:rd} has negative sign whenever
\begin{equation}
 (1+ t e^\beta (\alpha - \omega e^\beta)) \exp[-\alpha(e^\beta+1)t] < \exp[-e^\beta(\alpha + \omega)t].
 \label{eq:case24}
\end{equation}
Since $t > 0$ and $a>0$, $\log(1+at)$ is a monotonic function of $t$ that grows more slowly than $t$. Therefore there exists $t^* > 0$ such that for any $t > t^*$, $t > \frac{\log(1+ t e^\beta (\alpha - \omega e^\beta))}{\alpha - \omega e^\beta}$. Therefore \eqref{eq:case24} holds for $t > t^*$.
\end{proof}

\subsection*{General clusters}

  We begin with notation that will simplify exposition.  Let $H_i=(\alpha_i(t), \omega_{ikj}(t), n_i, T_i)$ represent cluster-level variables. 
  Let $\E_{t_{ij}}[\cdot]$ denote expectation with respect to the infection time of $j$, and let $\E_{\T_{i,-j}}[\cdot]$ denote expectation with respect to infection times $t_{ik}$ for $k\neq j$ (and implicitly, outcomes $Y_{ik}(T_i))$. 
  Since $Y_{ij}(t) = \indicator{t_{ij}<t}$, we will employ expectation with respect to $Y_{ij}(t)$ and $t_{ij}$ interchangeably, so $\E_{t_{ij}}[Y_{ij}(t)] = \E_{Y_{ij}(t)}[Y_{ij}(t)]$.  
  Let $\x_i = (x_{i1}, \ldots, x_{in_i})$ be the vector of covariate $x$ in cluster $i$, and $\E_{\x_{i,-j}}[\cdot]$ denote expectation with respect to $x_{ik}$ for $k\neq j$. 
  By iterating expectations, we can decompose the conditional expectations that comprise the risk ratio as follows,
  \small \[ \E[Y_{ij}(T_i)|x_{ij}=x] = \E_{H_i}\Bigg[ \E_{\x_{i,-j}} \Big[ \E_{\T_{i,-j}}\big[ \E_{t_{ij}}[ Y_{ij}(T_i) \mid x_{ij}=x, \T_{i,-j}, \x_{i,-j}, H_i] \mid x_{ij}=x, \x_{i,-j}, H_i \big] \mid x_{ij}=x, H_i \Big] \mid x_{ij}=x \Bigg] .\]
  
  \normalsize
  \noindent At time $T_i$, the innermost expectation is 
  \[ \E_{t_{ij}}[Y_{ij}(T_i) \mid x_{ij}=x,\x_{i,-j}, \T_{i,-j},H_i] = 1 - \exp\left(-e^{x\beta} \int_0^{T_i} \left( \alpha_i(t) + \sum_{k=1}^{n_i} \indicator{t_{ik}<t} \omega_{ikj}(t-t_{ik}) e^{x_{ik}\gamma}\right)\dx{t} \right) . \]

\begin{lem}
Let $X$ be a non-negative random variable that takes at least some positive values, and let $a$ be a non-negative constant. Then
\[ \frac{\E_X [1 - \exp(-aX)]}{\E_X [1 - \exp(-X)]} < 1 \text{  iff    } a<1 \]
\[ \frac{\E_X [1 - \exp(-aX)]}{\E_X [1 - \exp(-X)]} > 1 \text{  iff    } a>1 \]
\[ \frac{\E_X [1 - \exp(-aX)]}{\E_X [1 - \exp(-X)]} = 1 \text{  iff    } a=1 \]
\begin{proof}
Let $a<1$. 
\[ \frac{\E_X [1 - \exp(-aX)]}{\E_X [1 - \exp(-X)]} < 1 \Leftrightarrow \E_X [\exp(-aX)] - \E_X [\exp(-X)] > 0 \]
\[ \E_X [\exp(-aX)] - \E_X [\exp(-X)] = \int_{0} ^{\infty} [\exp(-ax) - \exp(-x)] f(x) \dx{x}  >0 \Leftrightarrow a< 1.\]
The proof for $a>1$ and $a=1$ is similar.
\end{proof}
\label{lem:gen}
\end{lem}

\subsubsection*{Result \ref{prop:biasnoclustercontgen}: No within-cluster contagion}
\begin{proof}
  Suppose $\omega_{ikj}(t)=0$ for all $t$ and $\x_i \indep \{\alpha_i(t), n_i, T_i\}$.  Then 
  \begin{equation*}
    \begin{split}
      RR &= \frac{\E[Y_{ij}(T_i)\mid x_{ij}=1]}{\E[Y_{ij}(T_i)\mid x_{ij}=0]}  \\
      &= \frac{\E_{H_i}[\E_{\x_{i,-j}}[\E_{\T_{i,-j}}[1-\exp(-e^\beta \int_0^{T_i} \alpha_i(t)\dx{t}) \mid x_{ij}=1, \x_{i,-j}, H_i ] \mid x_{ij}=1, H_i ] \mid x_{ij}=1 ] }{\E_{H_i}[\E_{\x_{i,-j}}[\E_{\T_{i,-j}}[1-\exp(-\int_0^{T_i} \alpha_i(t)\dx{t}) \mid x_{ij}=0, \x_{i,-j},  H_i ] \mid x_{ij}=0, H_i ] \mid x_{ij}=0] }   \\
      &= \frac{\E_{H_i}[\E_{\x_{i,-j}}[\E_{\T_{i,-j}}[1-\exp(-e^\beta \int_0^{T_i} \alpha_i(t)\dx{t}) \mid  H_i ] \mid x_{ij}=1, H_i ] \mid x_{ij}=1 ] }{\E_{H_i}[\E_{\x_{i,-j}}[\E_{\T_{i,-j}}[1-\exp(-\int_0^{T_i} \alpha_i(t)\dx{t}) \mid  H_i ] \mid x_{ij}=0, H_i ] \mid x_{ij}=0] }   \\
      &= \frac{\E_{H_i}[\E_{\x_{i,-j}}[1-\exp(-e^\beta \int_0^{T_i} \alpha_i(t)\dx{t}) \mid x_{ij}=1, H_i] \mid x_{ij}=1] }{\E_{H_i}[\E_{\x_{i,-j}}[1-\exp(-\int_0^{T_i} \alpha_i(t)\dx{t}) \mid x_{ij}=0, H_i ] \mid x_{ij}=0 ] }   \\
      &= \frac{\E_{H_i}[1-\exp(-e^\beta \int_0^{T_i} \alpha_i(t)\dx{t}) \mid x_{ij}=1 ]}{\E_{H_i}[1-\exp(-\int_0^{T_i} \alpha_i(t)\dx{t}) \mid x_{ij}=0]} \\
      & = \frac{\E_{H_i}[1-\exp(-e^\beta \int_0^{T_i} \alpha_i(t)\dx{t})]}{\E_{H_i}[1-\exp(-\int_0^{T_i} \alpha_i(t)\dx{t})]},
    \end{split}
  \end{equation*}
  where the third line follows because the distribution of $Y_{ij}(T_i)$ does not depend on $\x_{i,-j}$, and only depends on $x_{ij}$ via multiplicative constant $e^\beta$; the fourth line follows because $Y_{ij}(T_i)$ does not depend on $\T_{i,-j}$, and the fifth line follows because $Y_{ij}(T_i)$ is independent of $\x_{i,-j}$ and $\x_i$ is independent of $H_i$.  Since the only difference between the numerator and denominator is the multiplicative constant $e^\beta$, by Lemma \ref{lem:gen} the risk ratio is direction-unbiased. 
\end{proof}

\subsubsection*{Result \ref{prop:biasindependentx}: Independent $\x$}
  \begin{proof}[Proof of Result \ref{prop:biasindependentx}]
  Suppose the covariates $\x_i$ are jointly independent and $\x_i \indep \{\alpha_i(t), \omega_{ikj}(t), n_i, T_i\}$. 
For any time $t>0$, we can write the cumulative hazard to subject $j$ in cluster $i$ as
\[ \Lambda_{ij}(t) = e^{x_{ij}\beta} \int_0^{T_i} (1-Y_{ij}(s)) \left( \alpha_i(s) + \sum_{k=1}^{n_i} Y_{ik}(s) \omega_{ikj}(s-t_{ik}) e^{x_{ik}\gamma}\right) \dx{s}  \]
For ease of exposition, let 
\[   \xi_{i}(t) = \alpha_i(t) + \sum_{k=1}^{n_i} Y_{ik}(t) \omega_{ikj}(t-t_{ik}) e^{x_{ik}\gamma}. \] 
Then
  \begin{equation*}
    \begin{split}
      RR & = \frac{\E[Y_{ij}(T_i)\mid x_{ij}=1]}{\E[Y_{ij}(T_i)\mid x_{ij}=0]} \\
      & = \frac{\E_{H_i}[\E_{\x_{i,-j}}[\E_{\T_{i,-j}}[1-\exp(-e^\beta \int_0^{T_i} (1-Y_{ij}(t)) \xi_i(t) \dx{t}) \mid \x_{i,-j}, x_{ij}=1, H_i ] \mid x_{ij}=1, H_i ] \mid x_{ij}=1 ] }{\E_{H_i}[\E_{\x_{i,-j}}[\E_{\T_{i,-j}}[1-\exp(-\int_0^{T_i} (1-Y_{ij}(t)) \xi_i(t) \dx{t}) \mid  \x_{i,-j},  x_{ij}=0, H_i ] \mid x_{ij}=0, H_i ] \mid x_{ij}=0 ] }   \\
      &= \frac{\E_{H_i}[\E_{\x_{i,-j}}[\E_{\T_{i,-j}}[1-\exp(-e^\beta \int_0^{T_i} (1-Y_{ij}(t)) \xi_i(t) \dx{t}) \mid  \x_{i,-j}, H_i ] \mid  H_i ] ] }{\E_{H_i}[\E_{\x_{i,-j}}[\E_{\T_{i,-j}}[1-\exp(-\int_0^{T_i} (1-Y_{ij}(t)) \xi_i(t) \dx{t}) \mid  \x_{i,-j},  H_i ] \mid  H_i ] ] }, 
    \end{split}
  \end{equation*}
  because the distribution of $\T_{i,-j}$ is invariant to conditioning on $x_{ij}=1$ or $x_{ij}=0$, when subject $j$ is susceptible, and because by joint independence of $\x_i$, the expectation $\E_{\x_{i,-j}}[\cdot]$ is also invariant to conditioning on $x_{ij}=1$ or $x_{ij}=0$, and $x_{ij}$ is independent of $H_i$. By Lemma \ref{lem:gen}
  the risk ratio is direction-unbiased.
\end{proof}

\section*{Risk ratio maps}

\subsection *{Exact risk ratio maps for clusters of size two}

Figures \ref{supfig:exact_est} and \ref{supfig:exact_dbias} provide the plots that illustrate analytic result \eqref{eq:exactrr} for different combinations of force of infection parameters $\alpha$ and $\omega$ as a function of susceptibility ($\beta$) and infectiousness ($\gamma$) parameters.  Figure \ref{supfig:exact_est} shows the exact expected value of the $log[RR]$, and Figure \ref{supfig:exact_dbias} shows the regions of the directional bias of the risk ratio as an approximation of the hazard ratio for the same combinations of parameters $\alpha$, $\omega$ and observation time $T_i$. We have demonstrated in Result \ref{prop:biasacrossnull} that for a given combination of $(\beta, \gamma)$ directional bias depends on the observation time and the ratio of $\omega / \alpha$. In Figures \ref{supfig:exact_est} and \ref{supfig:exact_dbias} observation time is chosen such that cumulative incidence when $\beta=0$ and $\gamma=0$ is kept constant around 0.15 for a given ratio of $\omega / \alpha$. With observation time chosen this way, the behavior of the bias is exactly the same for the same ratio of $\omega / \alpha$ regardless of the absolute values of these two parameters. 

\begin{figure}
\centering
\includegraphics[scale=0.85]{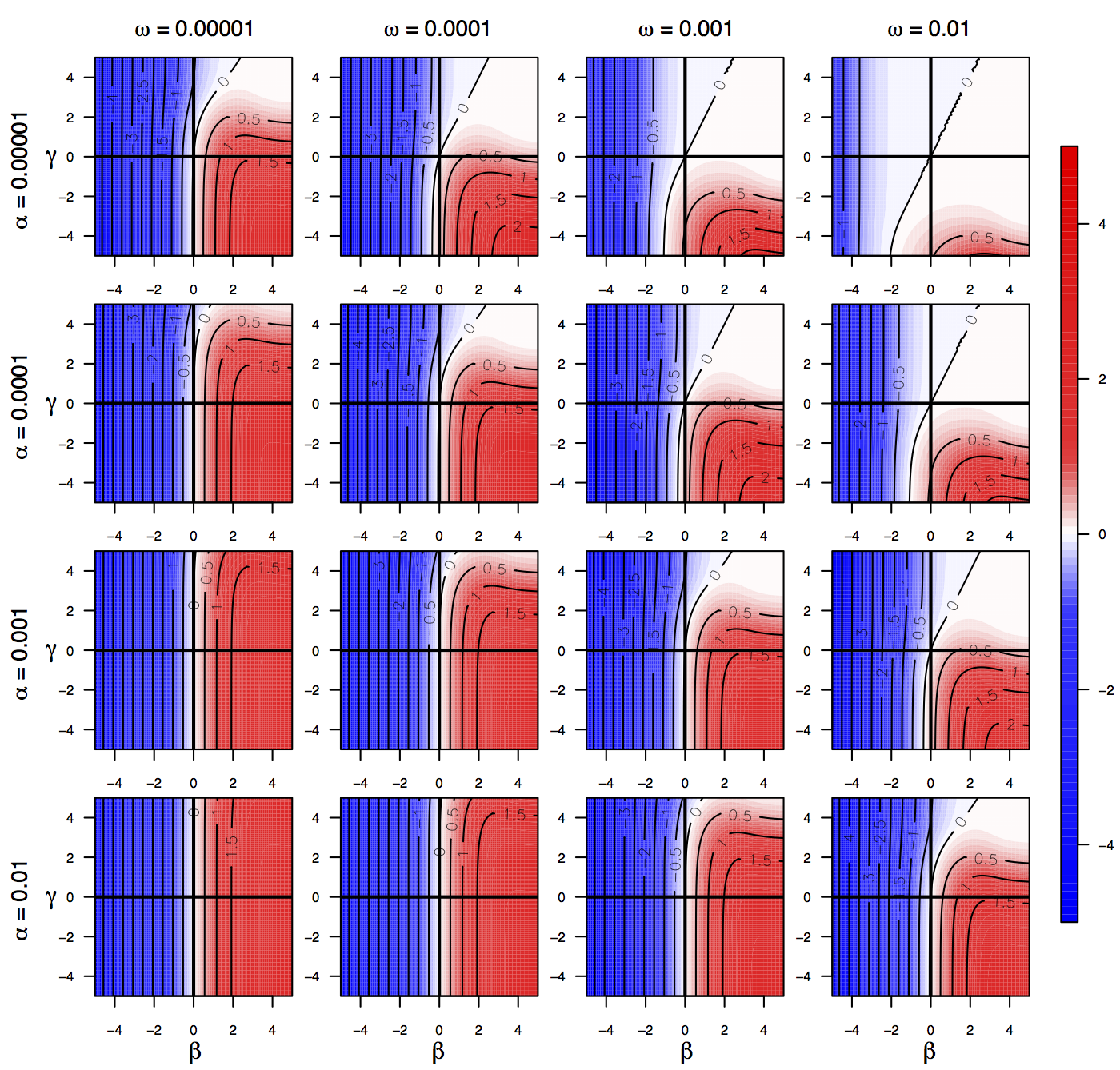}
\caption{Computed $\log[RR]$ as a function of $\beta$ and $\gamma$ in clusters of size two, when exactly one subject per cluster has a value of $x=1$. Observation time is constant and chosen such that the cumulative incidence when $\beta=0$ and $\gamma=0$ is approximately 0.15 for a given combination of $\alpha$ and $\omega$.} 
\label{supfig:exact_est}
\end{figure}

\begin{figure}
\centering
\includegraphics[scale=0.85]{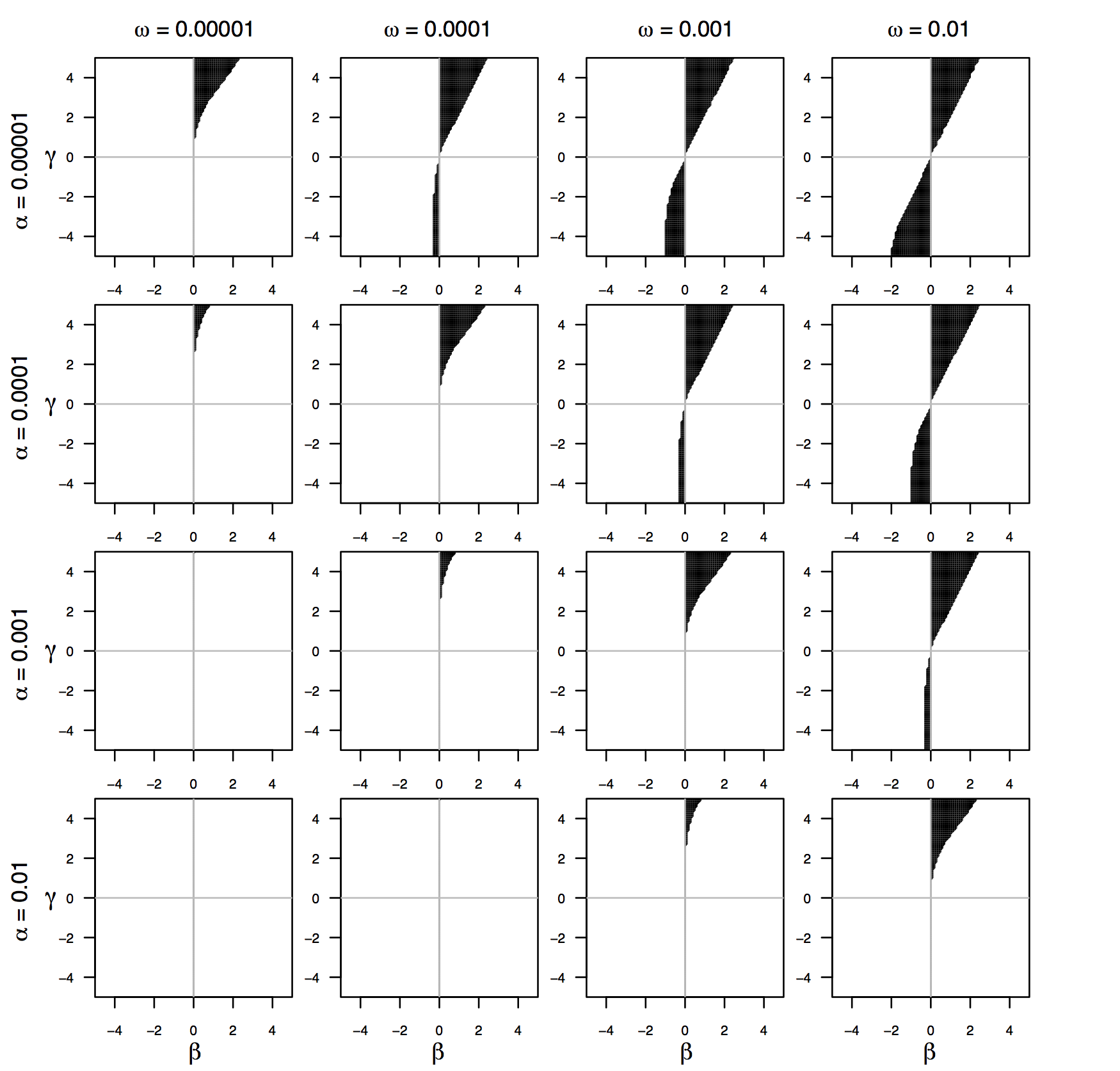}
\caption{Regions of direction bias of computed $\log[RR]$ as a function of $\beta$ and $\gamma$ in clusters of size two, when exactly one subject per cluster has a value of $x=1$.  Observation time is constant and chosen such that the cumulative incidence when $\beta=0$ and $\gamma=0$ is approximately 0.15 for a given combination of $\alpha$ and $\omega$.} 
\label{supfig:exact_dbias}
\end{figure}

\subsection*{Simulation results}
Exact expression for the expectation of the risk ratio is intractable in general case. This section provides a summary of the simulation results for different study designs and combinations of epidemiologic parameters. In simulations we vary: 
\begin{itemize} 
\setlength\itemsep{0em}
	\item Distribution of covariate $x$: block randomization, independent Bernoulli, cluster randomization;
	\item Cluster size distribution: fixed size, Poisson distributed; 
	\item Observation period: constant at different values, variable; 
	\item Subjects infected at baseline: $Pr[Y(0)=1]=0$; $Pr[Y(0)=1]>0$.
	\item Ratio $\omega / \alpha $. 
\end{itemize}

\subsubsection *{Distribution of covariate $x$ when cluster size is constant}

\begin{figure}
\centering
\includegraphics[scale=0.7]{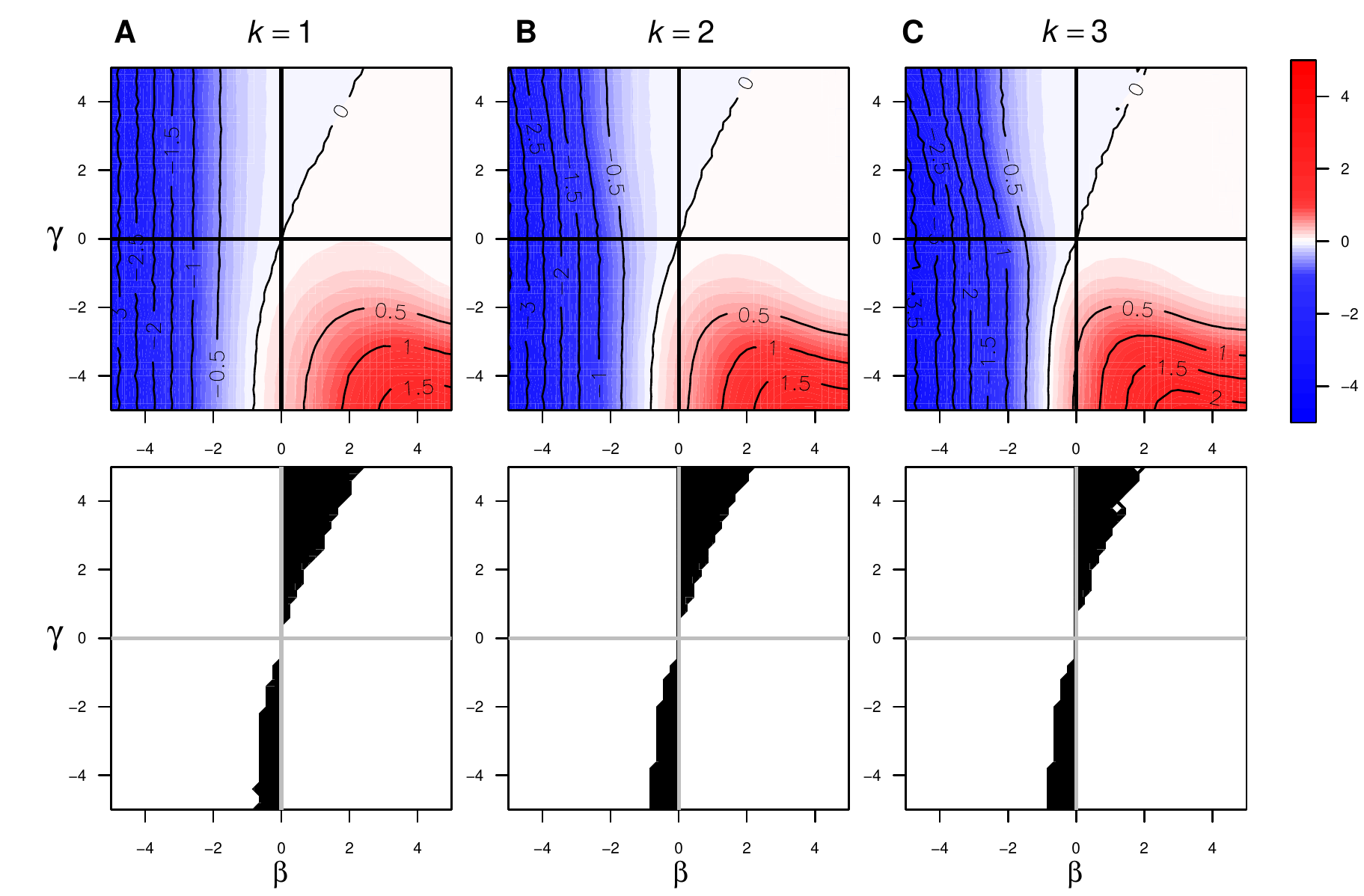}
\caption{$\log[RR]$ (top row) and region of direction bias (bottom row) as a function of $\beta$ and $\gamma$ when cluster size is constant and $x$ is block randomized: $\sum_{j=1}^{n_i} x_{ij}= k$.} 
\label{fig:x_block_fix_ni}
\end{figure}

\begin{figure}
\centering
\includegraphics[scale=0.7]{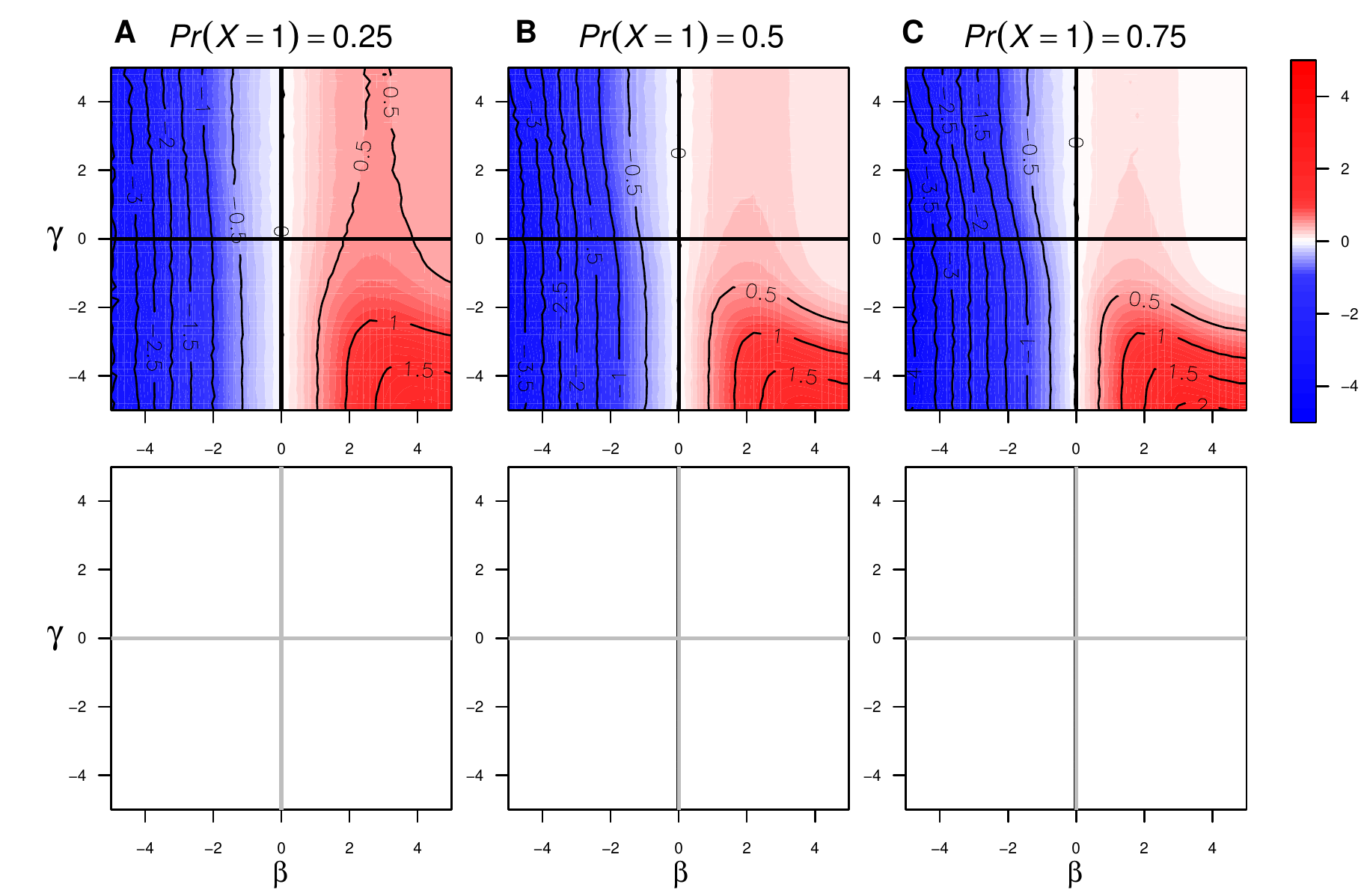}
\caption{$\log[RR]$ (top row) and region of direction bias (bottom row) as a function of $\beta$ and $\gamma$ when cluster size is constant and $x$ has independent Bernoulli distribution with varying $Pr[x=1]$.} 
\label{fig:x_bernoulli_fix_ni}
\end{figure}

\begin{figure}
\centering
\includegraphics[scale=0.7]{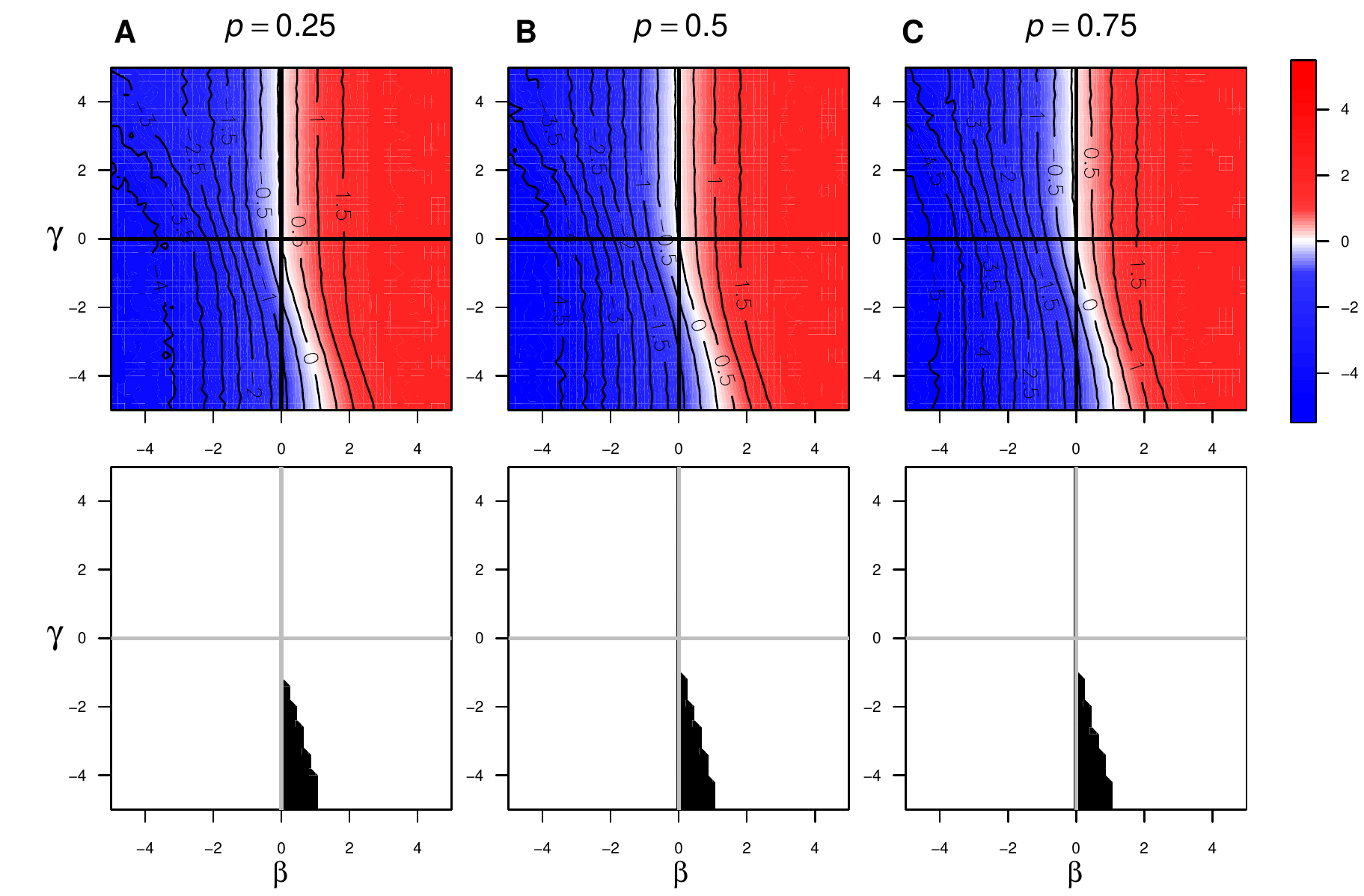}
\caption{$\log[RR]$ (top row) and region of direction bias (bottom row) as a function of $\beta$ and $\gamma$ when cluster size is constant and $x$ is cluster randomized: proportion $p$ of clusters have $\sum_{j=1}^{n_i} x_{ij}= 4$, and remaining $1-p$ have $\sum_{j=1}^{n_i} x_{ij}= 0$.} 
\label{fig:x_cluster_fix_ni}
\end{figure}

First, we look at the behavior of the bias of the risk ratio as an approximation of the hazard ratio for different types of the distribution of covariate $x$ when cluster size is constant. Figure \ref{fig:x_block_fix_ni} shows simulation results for block randomized distribution of $x$, Figure \ref{fig:x_bernoulli_fix_ni} - for independent Bernoulli distribution of $x$, and Figure \ref{fig:x_cluster_fix_ni} - for cluster randomized distribution of $x$. 
In all simulations presented in this subsection (Figures \ref{fig:x_block_fix_ni} - \ref{fig:x_cluster_fix_ni}) the following parameters are held constant:
\begin{itemize} 
\setlength\itemsep{0em}
	\item Force of infection parameters: $\alpha = 0.0001$, $\omega = 0.01$;
	\item Cluster size: $n_i = 4$ for $i=1,\ldots, N$;
	\item Observation time: $T_i = 450$ for $i=1,\ldots, N$ (giving cumulative incidence of approximately 0.15 when $\beta=0$ and $\gamma=0$);
	\item All subjects uninfected at baseline: $Y_{ij}(0) = 0$ for $i = 1,\ldots, N$ and $j = 1, 2, 3, 4$;
	\item Simulation parameters: number of clusters $N=500$, number of simulations per combination of parameters $=200$. 
\end{itemize} 
As demonstrated analytically in the Result \ref{prop:biasindependentx}, independent Bernoulli distribution of $x$ results in the direction-unbiased risk ratio (Figure \ref{fig:x_bernoulli_fix_ni}). Lack of joint independence in the distribution of $x$, however, generally results in the risk ratio exhibiting direction bias in some regions of the $(\beta, \gamma)$ parameter space. Figure \ref{fig:x_block_fix_ni} shows that bias patterns under block randomization and constant cluster size mimic that of the simple two-person cluster case. Figure \ref{fig:x_cluster_fix_ni} shows that cluster randomized distribution of $x$ leads to the direction bias in the regions where $\beta$ and $\gamma$ have opposite signs, and when the risk ratio is direction unbiased, it is not necessarily biased towards the null, but may be biased away from the null.

\subsubsection* {Variable cluster size under different distributions of covariate $x$}

\begin{figure}
\centering
\includegraphics[scale=0.85]{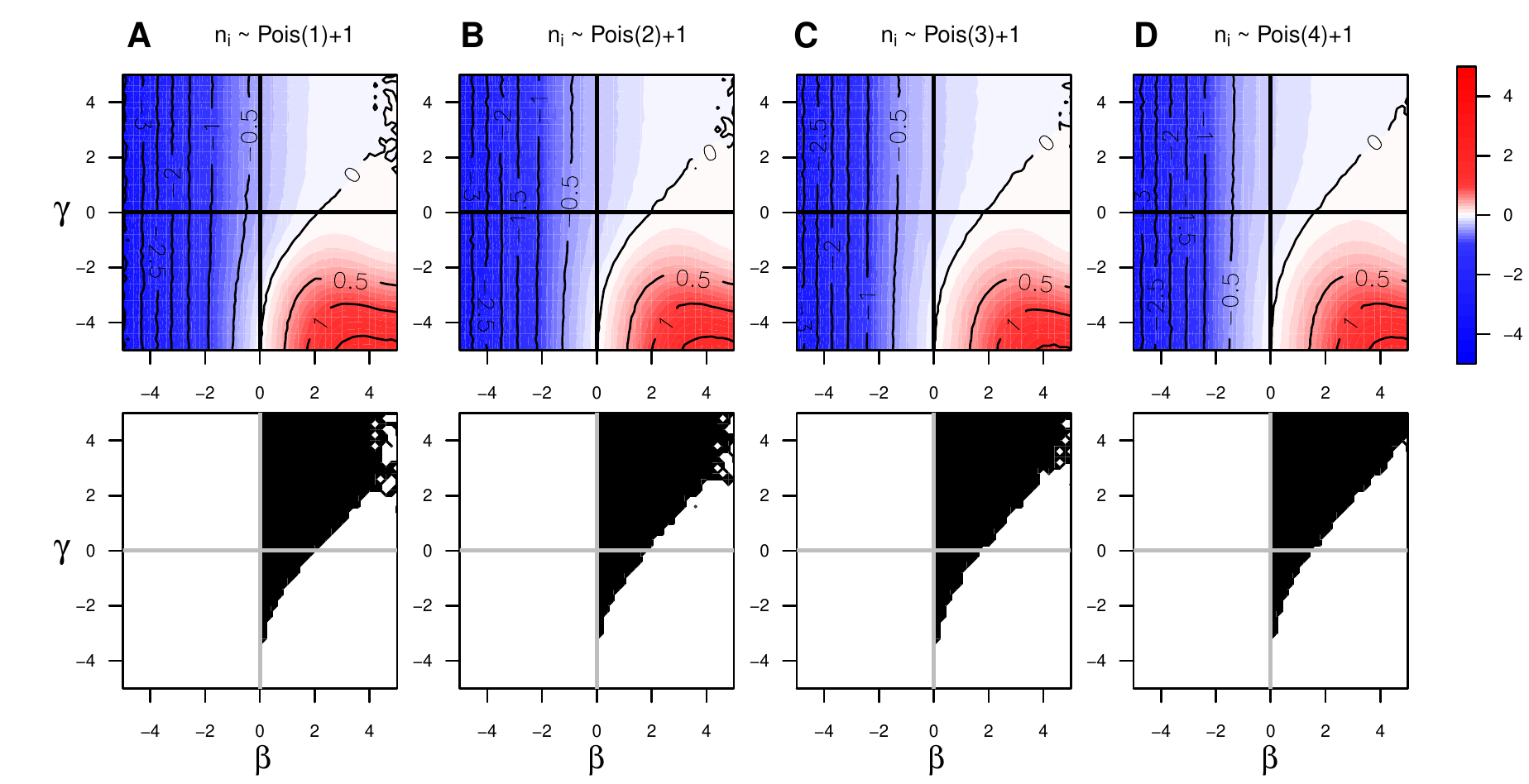}
\caption{$\log[RR]$ (top row) and region of direction bias (bottom row) as a function of $\beta$ and $\gamma$ when cluster size $n_i \sim \text{Pois}(\mu)+1$ and $x$ is block randomized such that $\sum_{j=1}^{n_i} x_{ij}=1$ for all $i$.} 
\label{fig:x_block1_var_ni}
\end{figure}

\begin{figure}
\centering
\includegraphics[scale=0.85]{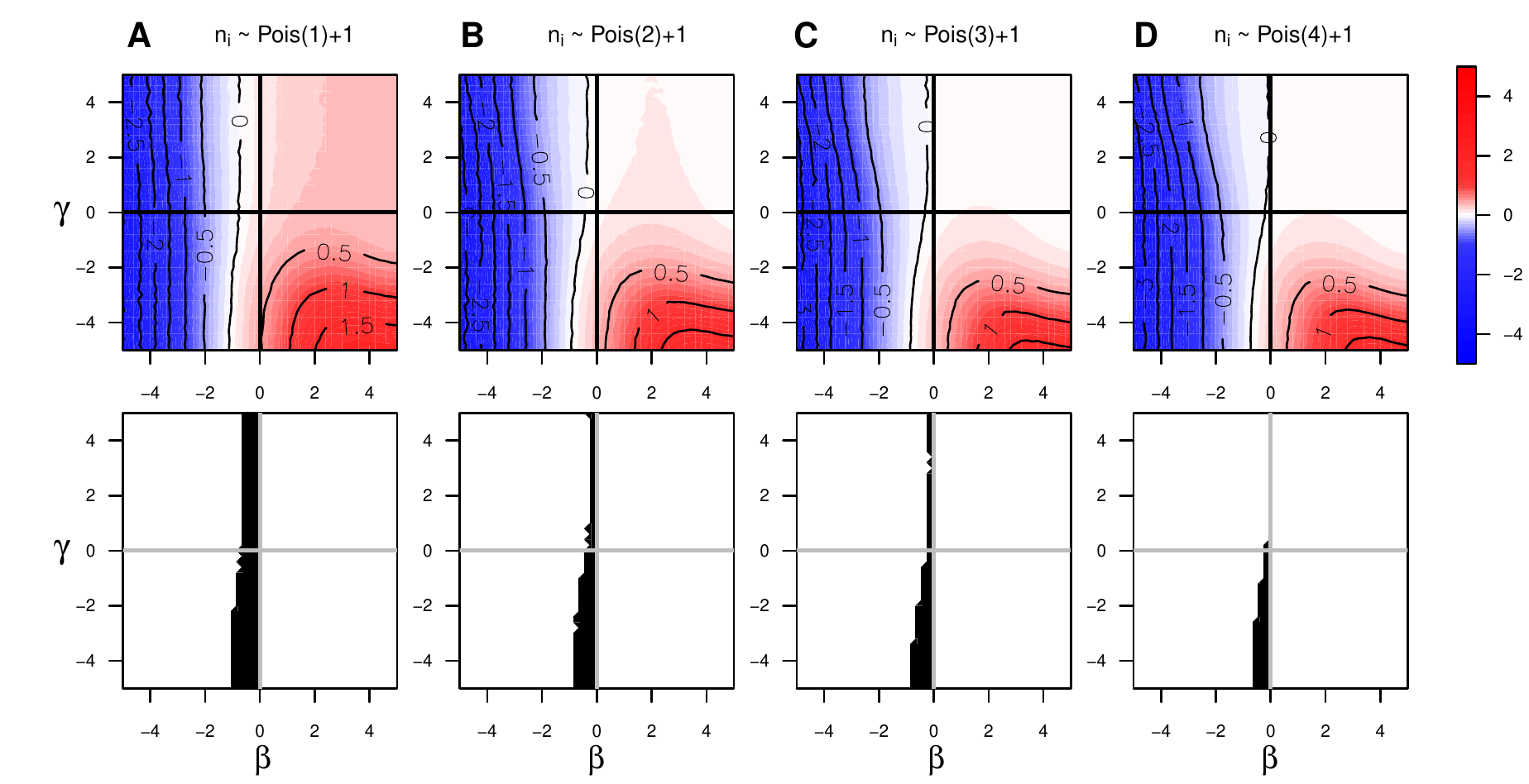}
\caption{$\log[RR]$ (top row) and region of direction bias (bottom row) as a function of $\beta$ and $\gamma$ when cluster size $n_i \sim \text{Pois}(\mu)+1$ and $x$ is block randomized such that $\sum_{j=1}^{n_i} x_{ij}=\lfloor n_i/2 \rfloor$ for all $i$.} 
\label{fig:x_block_half_var_ni}
\end{figure}

\begin{figure}
\centering
\includegraphics[scale=0.85]{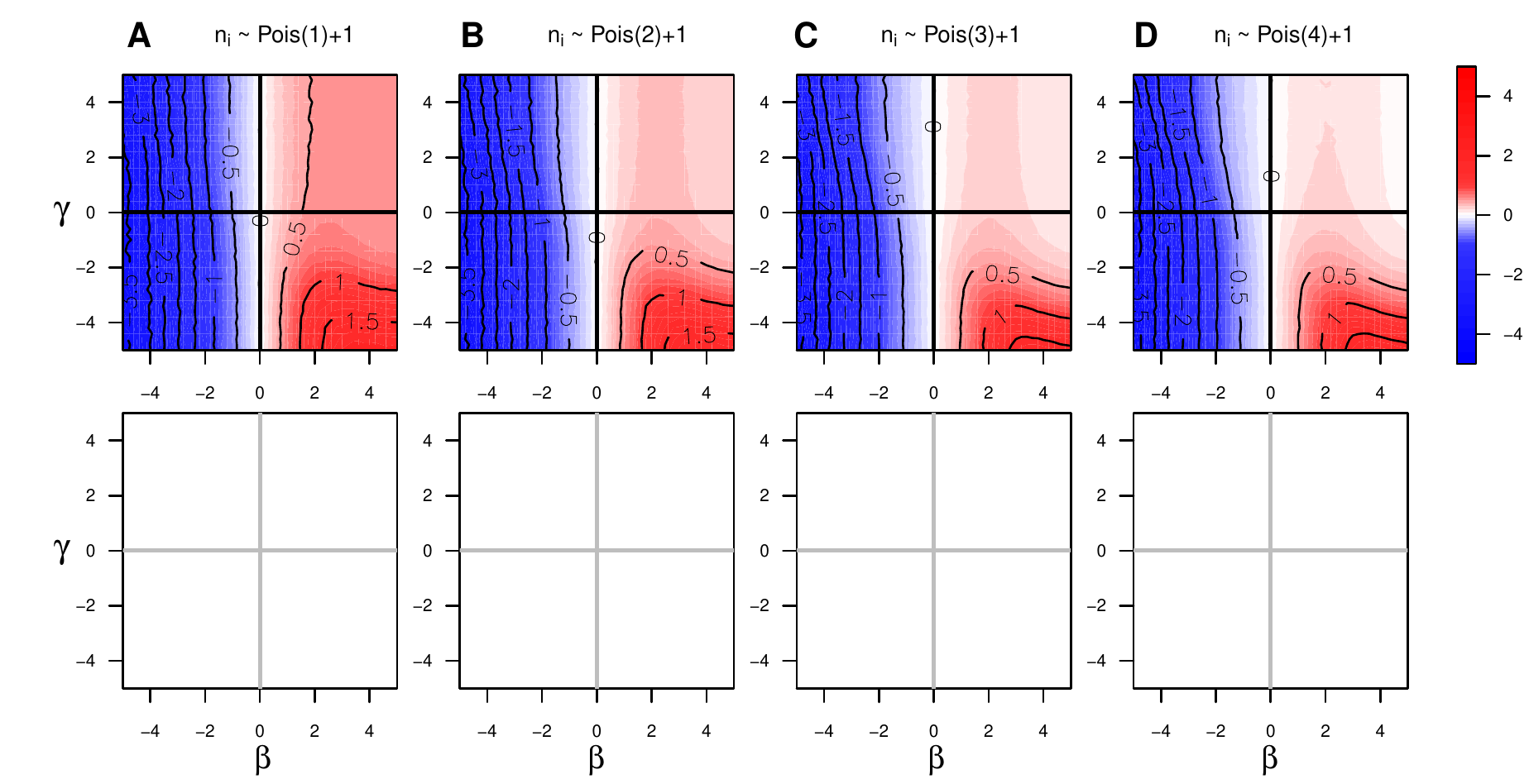}
\caption{$\log[RR]$ (top row) and region of direction bias (bottom row) as a function of $\beta$ and $\gamma$ when cluster size $n_i \sim \text{Pois}(\mu)+1$ and $x$ has independent Bernoulli distribution with $Pr[x=1]=0.5$. } 
\label{fig:x_bernoulli_var_ni}
\end{figure}

\begin{figure}
\centering
\includegraphics[scale=0.85]{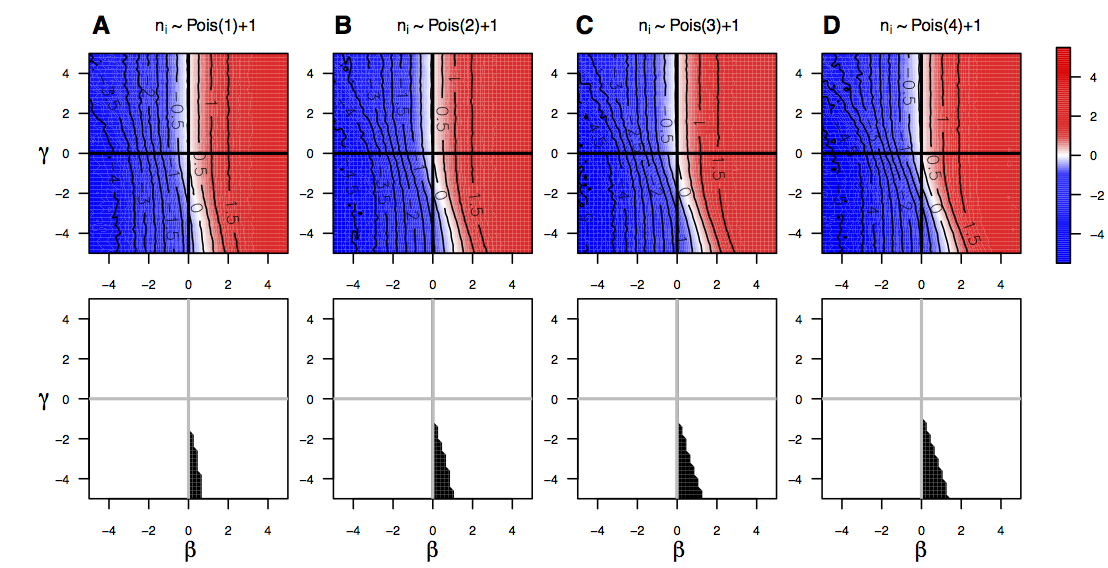}
\caption{$\log[RR]$ (top row) and region of direction bias (bottom row) as a function of $\beta$ and $\gamma$ when cluster size $n_i \sim \text{Pois}(\mu)+1$ and $x$ is cluster randomized: half of clusters have $\sum_{j=1}^{n_i} x_{ij}= n_i$, and remaining half have $\sum_{j=1}^{n_i} x_{ij}= 0$. } 
\label{fig:x_cluster_var_ni}
\end{figure}

In this subsection we explore the behavior of the risk ratio bias under variable cluster size, which follows Poisson distribution with different means. Figures \ref{fig:x_block1_var_ni} - \ref{fig:x_cluster_var_ni} show simulation results for the average cluster size between two and five, under different distributions of covariate $x$. In Figure \ref{fig:x_block1_var_ni} covariate $x$ is block randomized such that for any cluster size only one subject per cluster has $x=1$; Figure \ref{fig:x_block_half_var_ni} shows block randomization of $x$, when in any cluster half of the subjects have $x=1$; in Figure \ref{fig:x_bernoulli_var_ni} covariate $x$ has Bernoulli distribution with $Pr[x=1]=0.5$, and Figure \ref{fig:x_cluster_var_ni} shows the results for cluster randomized distribution of $x$ such that in half of clusters all subjects have $x=1$, and in the remaining half everyone has $x=0$.
In all simulations presented in this subsection (Figures \ref{fig:x_block1_var_ni} - \ref{fig:x_cluster_var_ni}) the following parameters are held constant:
\begin{itemize} 
\setlength\itemsep{0em}
	\item Force of infection parameters: $\alpha = 0.0001$, $\omega = 0.01$;
	\item Observation time: $T_i = 750$, when $n_i \sim \text{Pois}(1)+1$; $T_i = 525$, when $n_i \sim \text{Pois}(2)+1$; $T_i = 450$, when $n_i \sim \text{Pois}(3)+1$; and $T_i = 330$, when $n_i \sim \text{Pois}(4)+1$ (giving cumulative incidence of approximately 0.15 when $\beta=0$ and $\gamma=0$);
	\item All subjects uninfected at baseline: $Y_{ij}(0) = 0$ for $i = 1,\ldots, N$ and $j = 1, \ldots, n_i$;
	\item Simulation parameters: number of clusters $N=500$, number of simulations per combination of parameters $=200$. 
\end{itemize} 

When covariate $x$ is block randomized, the behavior of risk ratio bias changes substantially when we allow cluster size to vary compared to holding it constant. Figures  \ref{fig:x_block1_var_ni}, \ref{fig:x_block_half_var_ni} and \ref{fig:x_block_fix_ni} demonstrate very different patterns of bias behavior, while all having block randomized distribution of $x$. When cluster size is fixed (Figure \ref{fig:x_block_fix_ni}), risk ratio is direction-unbiased when $\gamma=0$, and bias in direction requires $\gamma$ being more extreme than and having the same sign as $\beta$. However, when cluster size varies under block randomized $x$, the risk ratio is not necessarily direction-unbiased when $\gamma=0$, or when $\gamma$ and $\beta$ have opposite signs.  Figures  \ref{fig:x_block1_var_ni} and \ref{fig:x_block_half_var_ni} show that under variable cluster size and block randomized $x$, bias behaves very differently depending on proportion of subjects with $x=1$ per cluster. Increasing imbalance in the distribution of $x$ generally makes things worse under such study design (compare Figure \ref{fig:x_block1_var_ni} to Figure \ref{fig:x_block_half_var_ni}).

When $x$ has independent Bernoulli or cluster randomized distribution,  bias of the risk ratio as an approximation of hazard ratio generally behaves similarly for constant or variable cluster size (compare Figure \ref{fig:x_bernoulli_var_ni} to Figure \ref{fig:x_bernoulli_fix_ni} for Bernoulli distributed $x$, and Figure \ref{fig:x_cluster_var_ni} to Figure \ref{fig:x_cluster_fix_ni} for cluster randomized distribution of $x$).

\subsubsection *{Duration and variability of observation time $T_i$}

\begin{figure}
\centering
\includegraphics[scale=0.85]{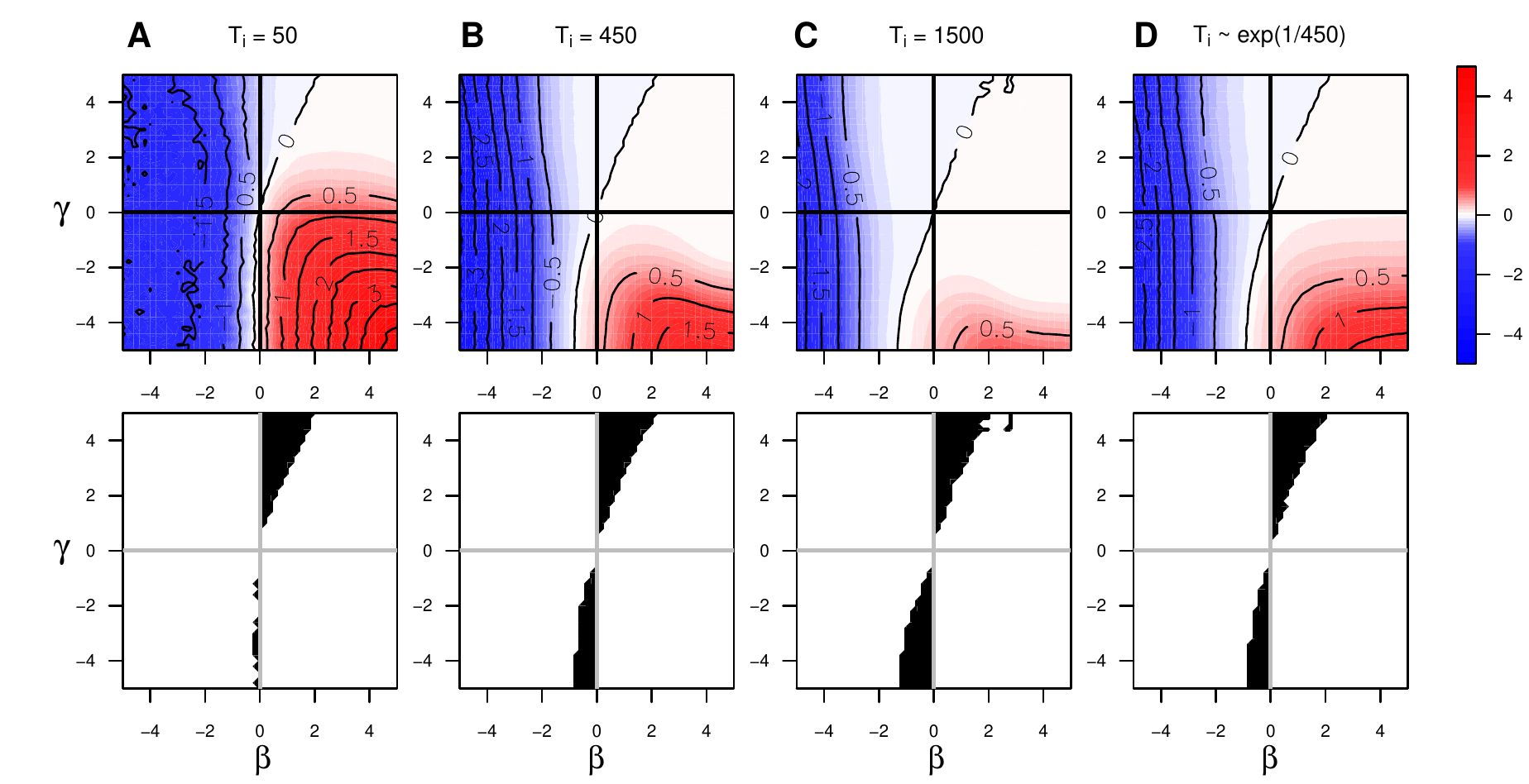}
\caption{$\log[RR]$ (top row) and region of direction bias (bottom row) as a function of $\beta$ and $\gamma$ for different observation time $T_i$, when cluster size is constant ($n_i=4$ for all $i$), and $x$ is block randomized such that $\sum_{j=1}^{n_i} x_{ij}= 2$.}
\label{fig:x_block_fix_ni_by_Ti}
\end{figure}

\begin{figure}
\centering
\includegraphics[scale=0.85]{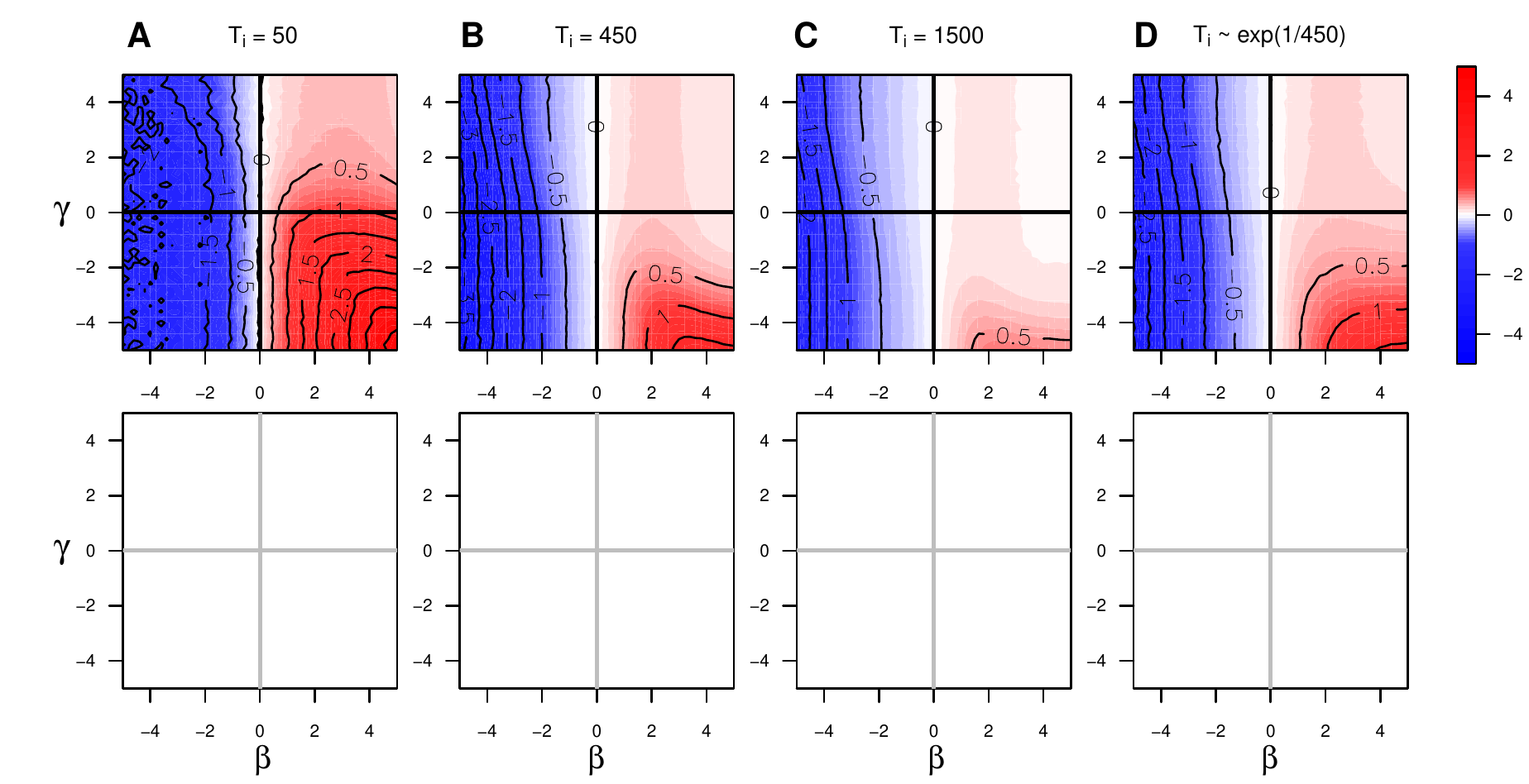}
\caption{$\log[RR]$ (top row) and region of direction bias (bottom row) as a function of $\beta$ and $\gamma$ for different observation time $T_i$, when cluster size $n_i \sim \text{Pois}(3)+1$ and $x$ has independent Bernoulli distribution with $Pr[x=1]=0.5$.} 
\label{fig:x_bernoulli_var_ni_by_Ti}
\end{figure}

\begin{figure}
\centering
\includegraphics[scale=0.85]{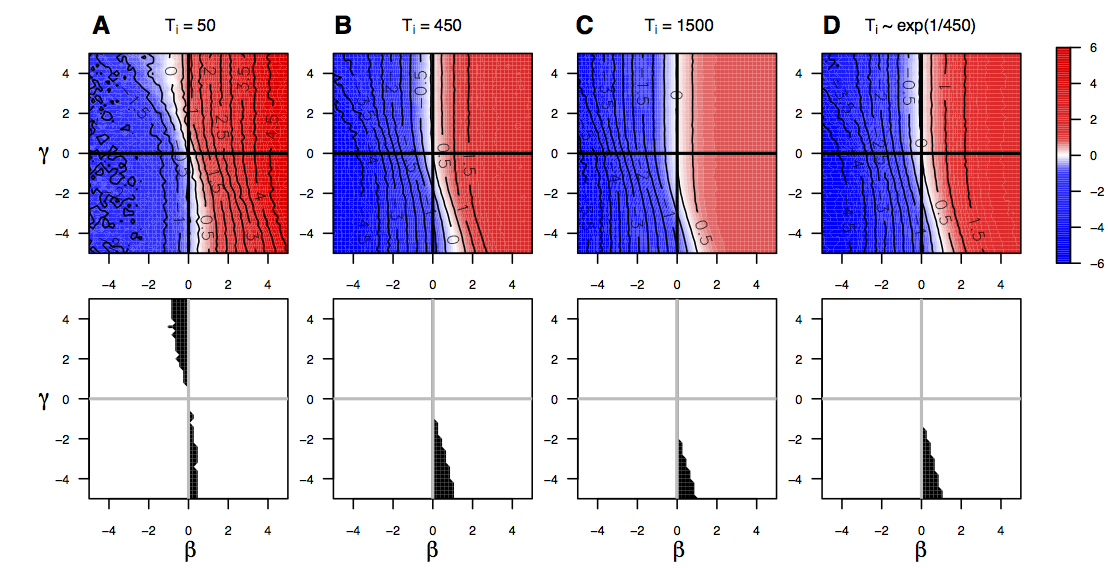}
\caption{$\log[RR]$ (top row) and region of direction bias (bottom row) as a function of $\beta$ and $\gamma$ for different observation time $T_i$, when cluster size is constant ($n_i=4$ for all $i$), and $x$ is cluster randomized: half of clusters have $\sum_{j=1}^{n_i} x_{ij}= 4$, and remaining half have $\sum_{j=1}^{n_i} x_{ij}= 0$.} 
\label{fig:x_cluster_fix_ni_by_Ti}
\end{figure}

This subsection looks at the impact of duration and variability of observation time on the risk ratio bias under different distributions of covariate $x$. Figures \ref{fig:x_block_fix_ni_by_Ti} - \ref{fig:x_cluster_fix_ni_by_Ti} show simulation results for three constant durations of observation ($T_i = 50$, $450$ and $1500$) and one, where observation time is exponentially distributed with rate $=1/450$.  Figure \ref{fig:x_block_fix_ni_by_Ti} shows the results for constant cluster size and block randomized distribution of $x$, Figure \ref{fig:x_bernoulli_var_ni_by_Ti} - for variable cluster size and independent Bernoulli distribution of $x$, and Figure \ref{fig:x_cluster_fix_ni_by_Ti} - for constant cluster size and cluster randomized distribution of $x$.
In all simulations presented in this subsection (Figures \ref{fig:x_block_fix_ni_by_Ti} - \ref{fig:x_cluster_fix_ni_by_Ti}) the following parameters are held constant:
\begin{itemize} 
\setlength\itemsep{0em}
	\item Force of infection parameters: $\alpha = 0.0001$, $\omega = 0.01$;
	\item All subjects uninfected at baseline: $Y_{ij}(0) = 0$ for $i = 1,\ldots, N$ and $j = 1, \ldots, n_i$;
	\item Simulation parameters: number of clusters $N=500$, number of simulations per combination of parameters $=200$. 
\end{itemize} 
With all other parameters held the same, increasing duration of observation leads to the increase in cumulative incidence. Under block randomized (Figure \ref{fig:x_block_fix_ni_by_Ti}) and independent Bernoulli (Figure \ref{fig:x_bernoulli_var_ni_by_Ti}) distribution of $x$ higher cumulative incidence increases the bias of the risk ratio as an approximation of the hazard ratio. However, under cluster randomized distribution of $x$ (Figure \ref{fig:x_cluster_fix_ni_by_Ti}) increasing duration of observation reduces the region, where the risk ratio exhibits direction bias, but does not necessarily reduce the bias in absolute value. Under any of the three distributions of $x$ variable duration of observation does not appreciably change the behavior of risk ratio bias compared to constant $T_i$.

\subsubsection*{Infections at time zero}

\begin{figure}
\centering
\includegraphics[scale=0.85]{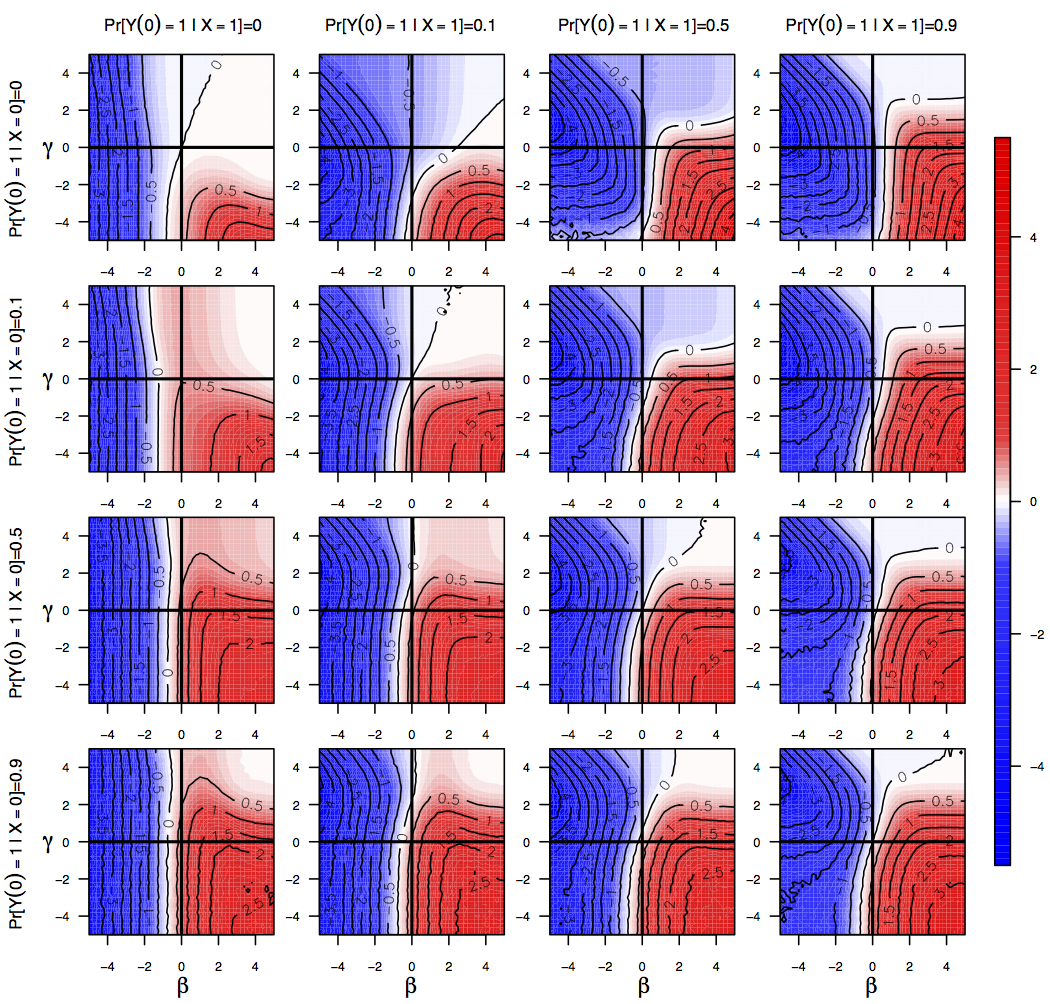}
\caption{$\log[RR]$ as a function of $\beta$ and $\gamma$ for a range of $Pr[Y(0)=1|X=1]$ and $Pr[Y(0)=1|X=0]$, when cluster size is constant ($n_i=4$ for all $i$), and $x$ is block randomized such that $\sum_{j=1}^{n_i} x_{ij}= 2$. }
\label{fig:x_block_fix_ni_blinf1_est}
\end{figure}

\begin{figure}
\centering
\includegraphics[scale=0.85]{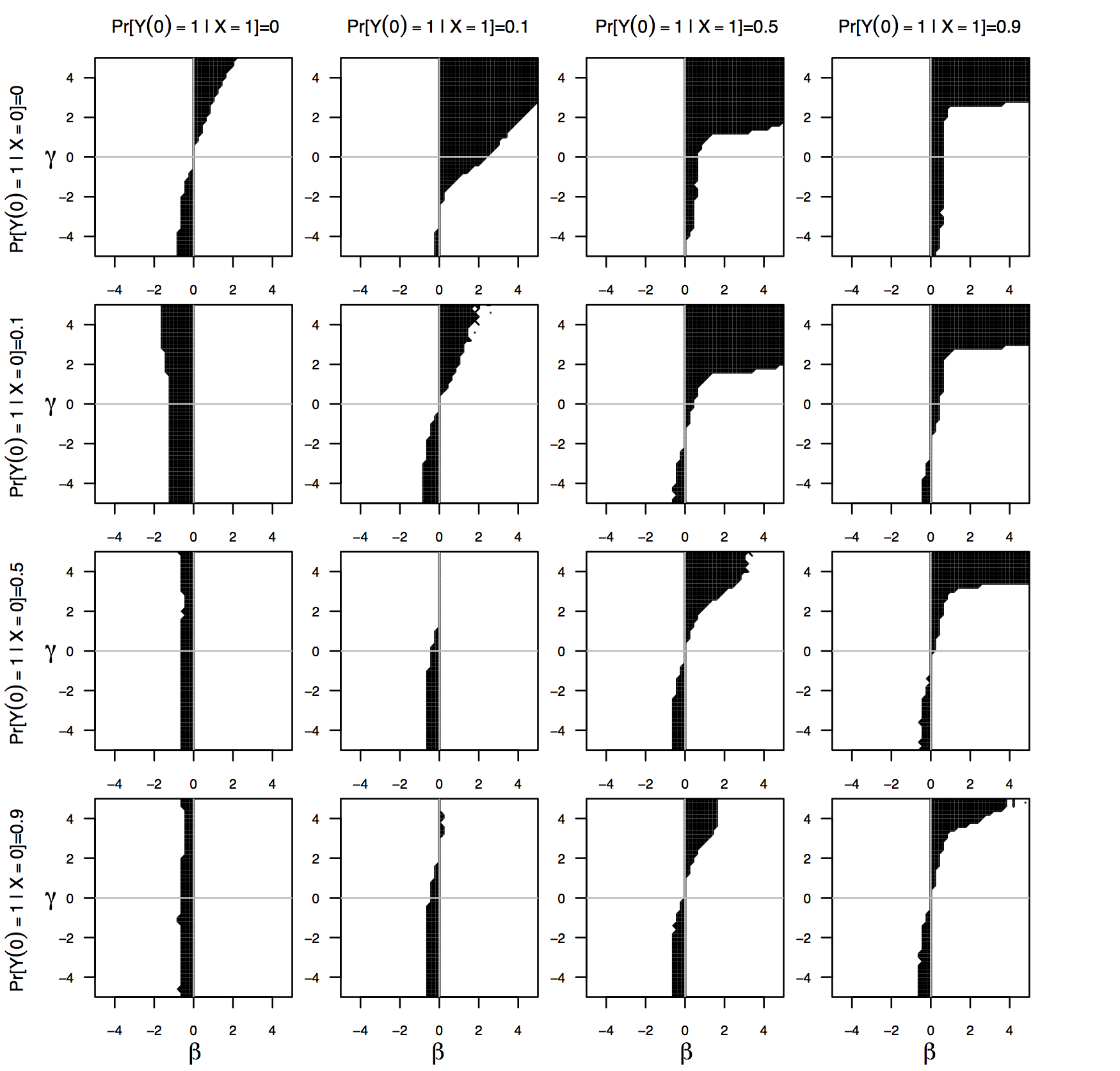}
\caption{Regions of direction bias of $\log[RR]$ as a function of $\beta$ and $\gamma$ for a range of $Pr[Y(0)=1|X=1]$ and $Pr[Y(0)=1|X=0]$, when cluster size is constant ($n_i=4$ for all $i$), and $x$ is block randomized such that $\sum_{j=1}^{n_i} x_{ij}= 2$.}
\label{fig:x_block_fix_ni_blinf1_dbias}
\end{figure}

\begin{figure}
\centering
\includegraphics[scale=0.85]{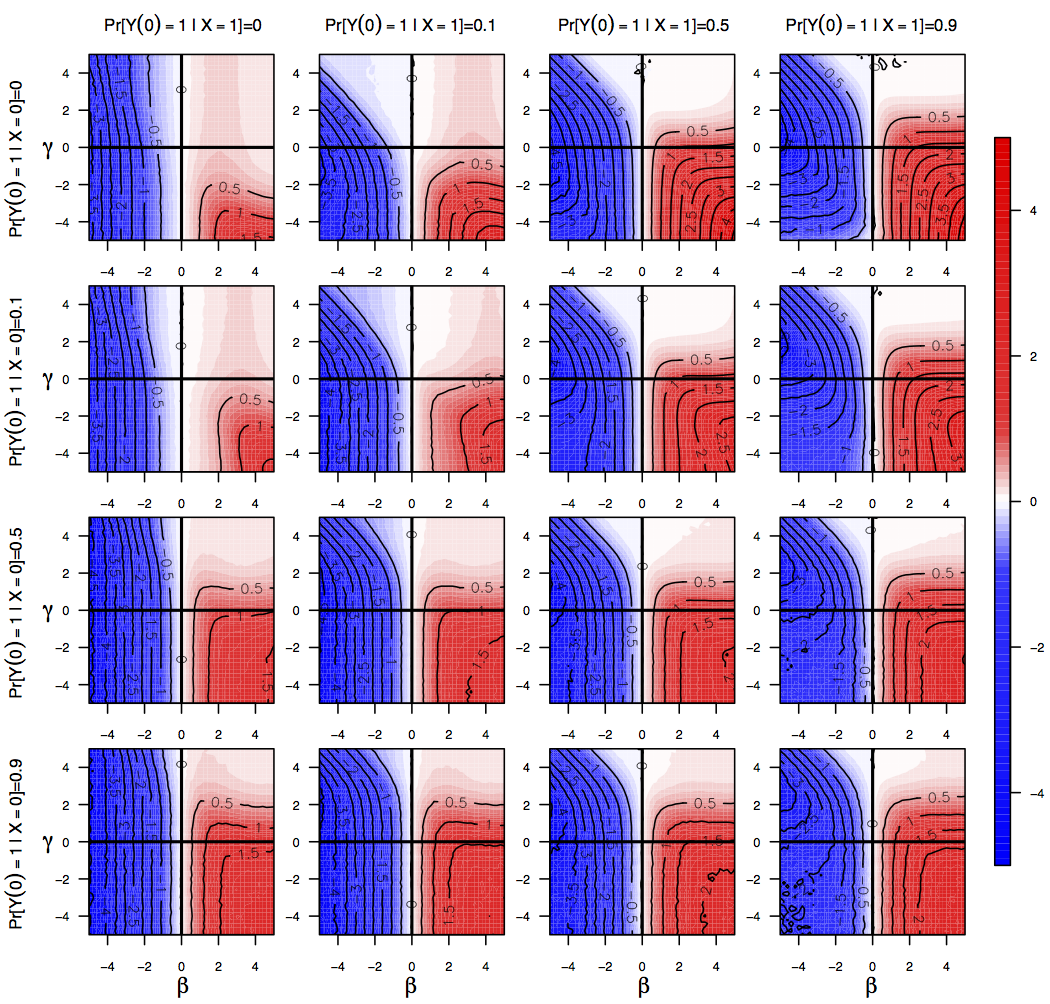}
\caption{$\log[RR]$ as a function of $\beta$ and $\gamma$ for a range of $Pr[Y(0)=1|X=1]$ and $Pr[Y(0)=1|X=0]$, when cluster size $n_i \sim \text{Pois}(3)+1$ and $x$ has independent Bernoulli distribution with $Pr[x=1]=0.5$.} 
\label{fig:x_bernoulli_var_ni_blinf1_est}
\end{figure}

\begin{figure}
\centering
\includegraphics[scale=0.85]{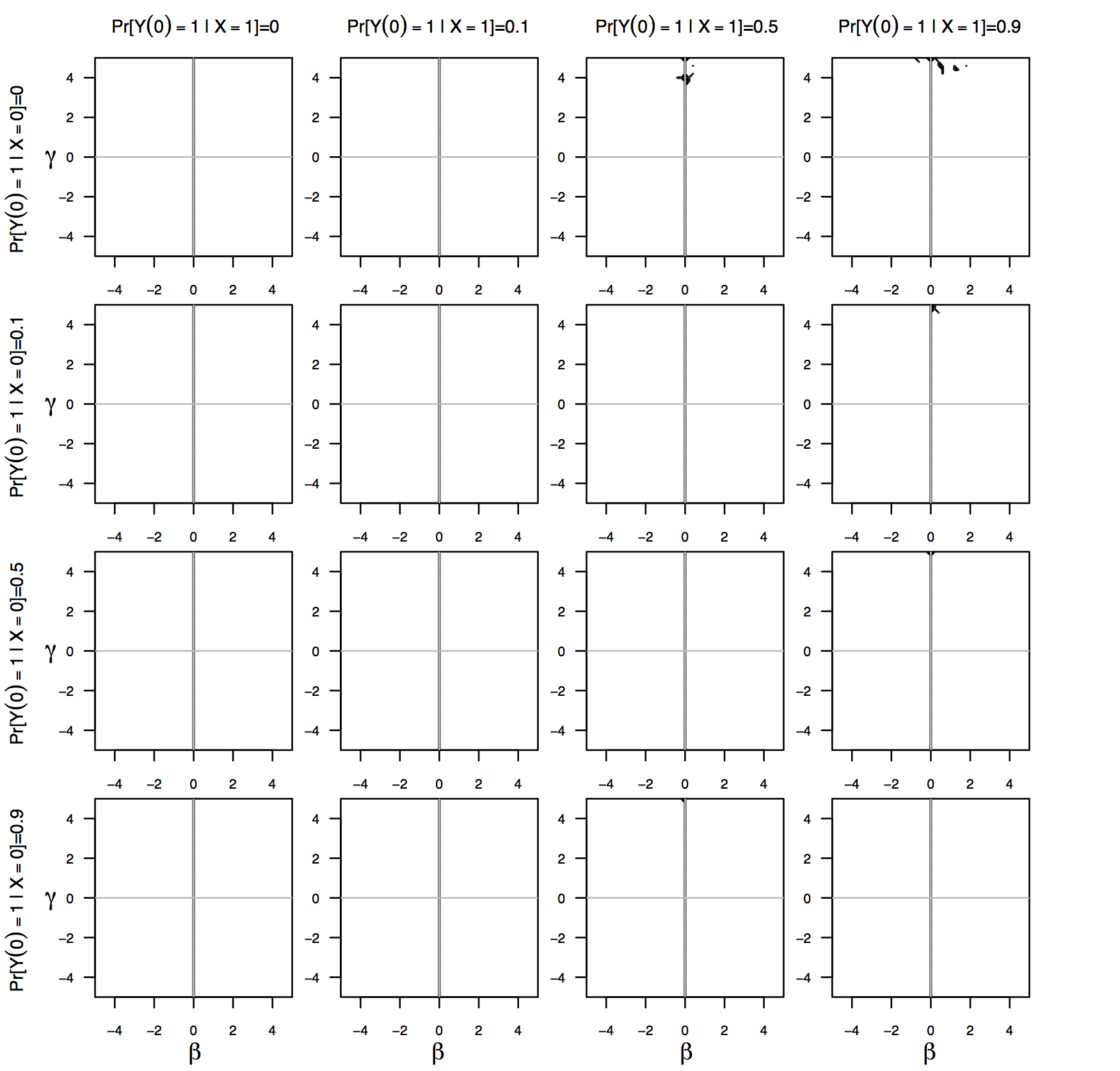}
\caption{Regions of direction bias of $\log[RR]$ as a function of $\beta$ and $\gamma$ for a range of $Pr[Y(0)=1|X=1]$ and $Pr[Y(0)=1|X=0]$, when cluster size $n_i \sim \text{Pois}(3)+1$ and $x$ has independent Bernoulli distribution with $Pr[x=1]=0.5$.} 
\label{fig:x_bernoulli_var_ni_blinf1_dbias}
\end{figure}

\begin{figure}
\centering
\includegraphics[scale=0.6]{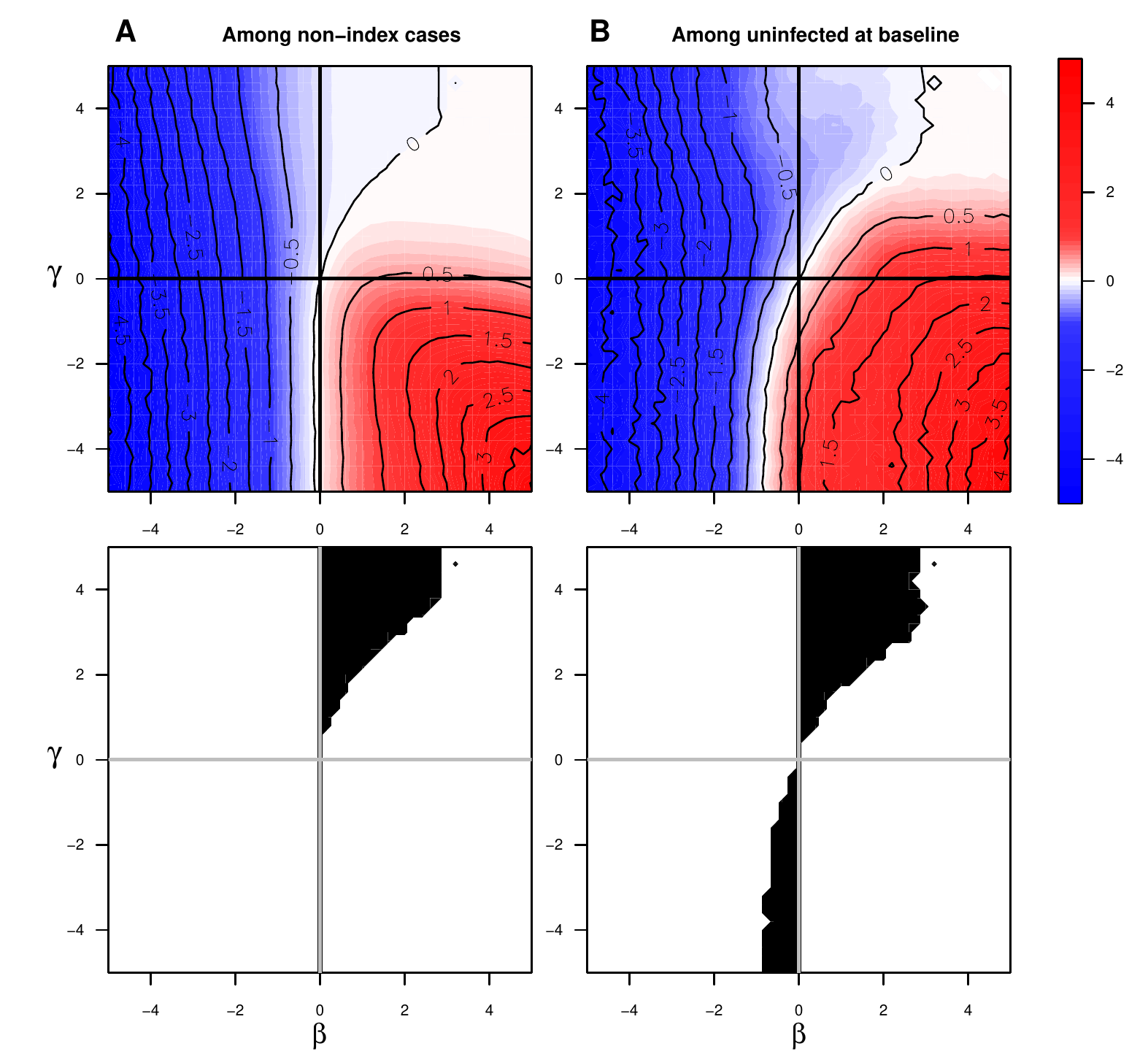}
\caption{$\log[RR]$ (top row) and region of direction bias (bottom row) as a function of $\beta$ and $\gamma$ for clusters selected based on having at least one infection at ``baseline", when cluster size is constant ($n_i=4$ for all $i$), and $x$ is block randomized such that $\sum_{j=1}^{n_i} x_{ij}= 2$. Risk ratio is calculated among all ``non-index" cases (A), and among cases uninfected at ``baseline" (B).} 
\label{fig:x_block_fix_ni_blinf2}
\end{figure}

\begin{figure}
\centering
\includegraphics[scale=0.6]{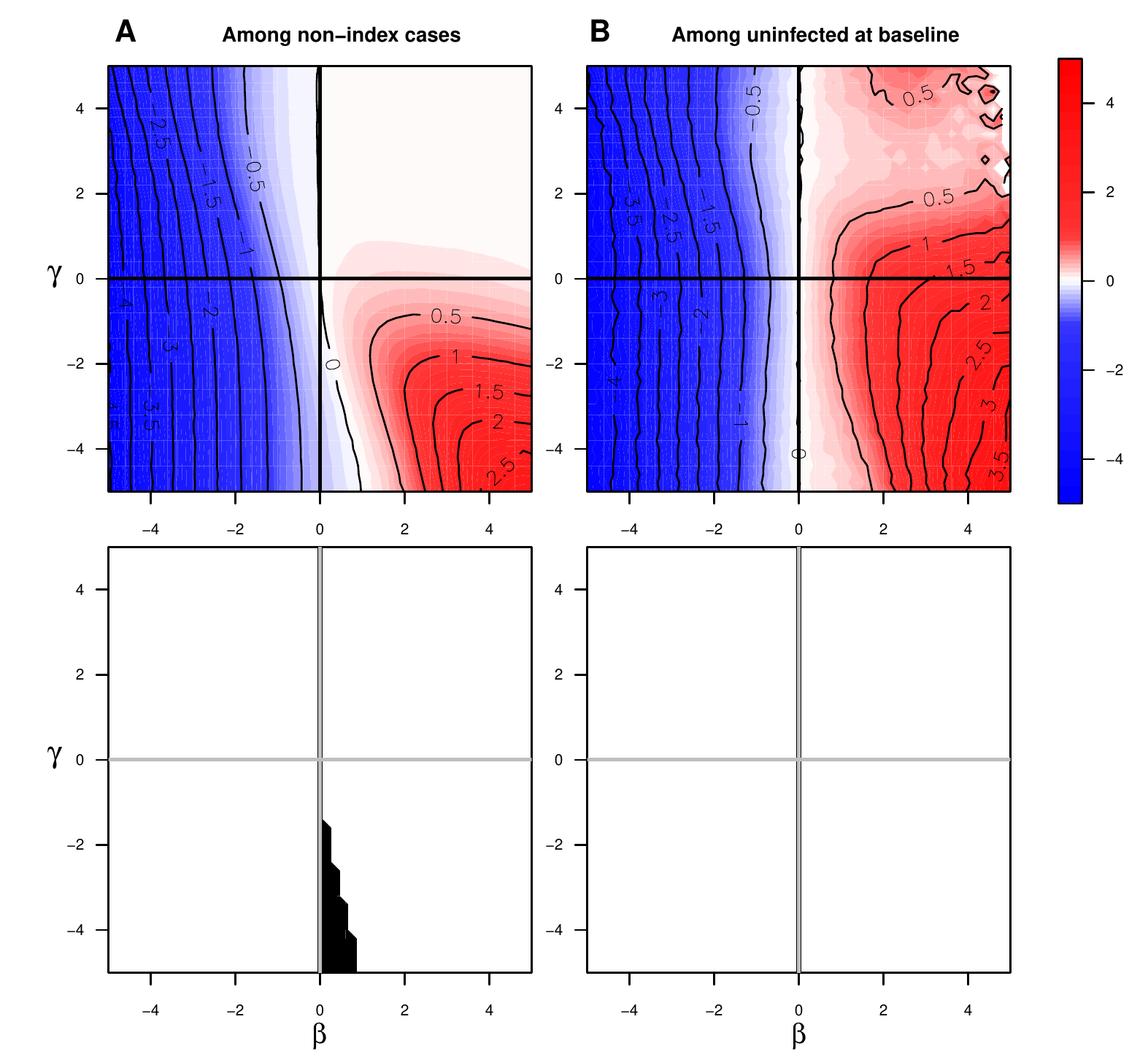}
\caption{$\log[RR]$ (top row) and region of direction bias (bottom row) as a function of $\beta$ and $\gamma$ for clusters selected based on having at least one infection at ``baseline", when cluster size $n_i \sim \text{Pois}(2)+2$ and $x$ has independent Bernoulli distribution with $Pr[x=1]=0.5$. Risk ratio is calculated among all ``non-index" cases (A), and among cases uninfected at ``baseline" (B).} 
\label{fig:x_bernoulli_var_ni_blinf2}
\end{figure}

In the simple case of clusters of size two and in all previous simulations we assumed that all subjects are uninfected at time zero (baseline). In practice, however, such study design is rarely a case. When researchers study infection transmission within clusters, they often select these clusters based on having at least one infected subject per cluster at baseline assessment (often called ``index" case). Sometimes studies would include a mix of clusters with and without infected subjects at baseline. In observational studies the distribution of covariate $x$ is given, and if $\beta$ and/or $\gamma$ is not zero, then the distribution of infections at baseline assessment is not independent of $x$. In experimental studies baseline distribution of infections may be independent of treatment $x$, and researchers can choose, whether subjects infected at baseline may or may not be assigned to treatment ($x=1$). In this subsection we explore the behavior of the risk ratio bias under the presence of infections at time zero. 

Figures \ref{fig:x_block_fix_ni_blinf1_est} and \ref{fig:x_block_fix_ni_blinf1_dbias} show the estimate of $\log[RR]$ and regions of direction bias for a range of values of $Pr[Y(0)=1|X=1]$ and $Pr[Y(0)=1|X=0]$ under block randomized distribution of $x$ and constant cluster size, and Figures \ref{fig:x_bernoulli_var_ni_blinf1_est} and \ref{fig:x_bernoulli_var_ni_blinf1_dbias} - under independent Bernoulli distribution of $x$ and variable cluster size. For every combination of parameters in these plots observation time $T_i$ was chosen such that cumulative incidence when $\beta=0$ and $\gamma=0$ is approximately 0.15. The risk ratio was computed among subjects uninfected at time zero. In most of the simulations presented in Figures \ref{fig:x_block_fix_ni_blinf1_est} - \ref{fig:x_bernoulli_var_ni_blinf1_dbias} number of clusters $N=500$. In some of the plots we increased $N$ to 1,000 and 5,000 to ensure convergence of the averages to expectations. For the same reason the number of simulations per combination of parameters varies between 100 and 1,000.

Figures \ref{fig:x_block_fix_ni_blinf2}-\ref{fig:x_bernoulli_var_ni_blinf2} summarize simulations that represent observational study design, which includes clusters based on having at least one ``index" case at baseline. These simulations were conducted as follows. We started with all subjects being uninfected and ran simulation for $T_i = 75$ (Figure \ref{fig:x_block_fix_ni_blinf2}) or $T_i = 150$ (Figure \ref{fig:x_bernoulli_var_ni_blinf2}). This time point then became the time of ``baseline'' assessment, at which we selected clusters with at least one infected subject. For different values of $\beta$ the initial number of clusters $N$ was chosen such that the number of clusters with at least one infected at ``baseline'' assessment was approximately 500. If there were more than one subject per cluster infected at baseline, an ``index" case was selected randomly from among them. We then ran simulation for $T_i=10$ (resulting in cumulative incidence of approximately 0.15 among subjects uninfected at ``baseline" when $\beta=0$ and $\gamma=0$) and calculated the risk ratio in two ways: among all subjects uninfected at ``baseline", and among ``non-index" cases. In Figure \ref{fig:x_block_fix_ni_blinf2} number of simulations per combination of parameters $=50$, and in Figure \ref{fig:x_bernoulli_var_ni_blinf2} - $200$.

In all simulations presented in this subsection (Figures \ref{fig:x_block_fix_ni_blinf1_est} - \ref{fig:x_bernoulli_var_ni_blinf2}) force of infection parameters are held constant at the following values: $\alpha = 0.0001$, $\omega = 0.01$.

Introducing subjects infected at baseline with different probabilities conditional on the value of covariate $x$ may result in substantial direction bias that generally increases with the increase of the difference in these conditional probabilities. Under constant cluster size and block randomized distribution of $x$, when $Pr[Y(0)=1|X=1] = Pr[Y(0)=1|X=0]$, bias generally behaves in the way similar to the same study design with no subjects infected at baseline (Figures \ref{fig:x_block_fix_ni_blinf1_est} and \ref{fig:x_block_fix_ni_blinf1_dbias}). Under variable cluster size and independent Bernoulli distribution of $x$ (Figures \ref{fig:x_bernoulli_var_ni_blinf1_est} and \ref{fig:x_bernoulli_var_ni_blinf1_dbias}), the risk ratio is direction-unbiased.

When study clusters are selected based on having at least one subject per cluster infected at baseline (``index" case), bias behavior under constant cluster size and block randomized distribution of $x$ is similar to having no subjects infected at baseline. Whether risk ratio is calculated among subjects uninfected at baseline, or excluding only ``index" cases the risk ratio exhibits direction bias in the same regions of the $(\beta, \gamma)$ parameter space. Under independent Bernoulli distribution of $x$, when we start with no subjects infected at time zero, the risk ratio is always direction-unbiased (Result \ref{prop:biasindependentx}). When we include clusters based on infections at ``baseline", and calculate the risk ratio excluding all subjects infected at the start of observation, we still have this nice property of the risk ratio. However, when the risk ratio is calculated excluding only the ``index" cases under the same conditions, direction-unbiasedness does not necessarily hold (Figure  \ref{fig:x_bernoulli_var_ni_blinf2}).

\subsubsection *{Ratio $\omega / \alpha $} 

\begin{figure}
\centering
\includegraphics[scale=0.7]{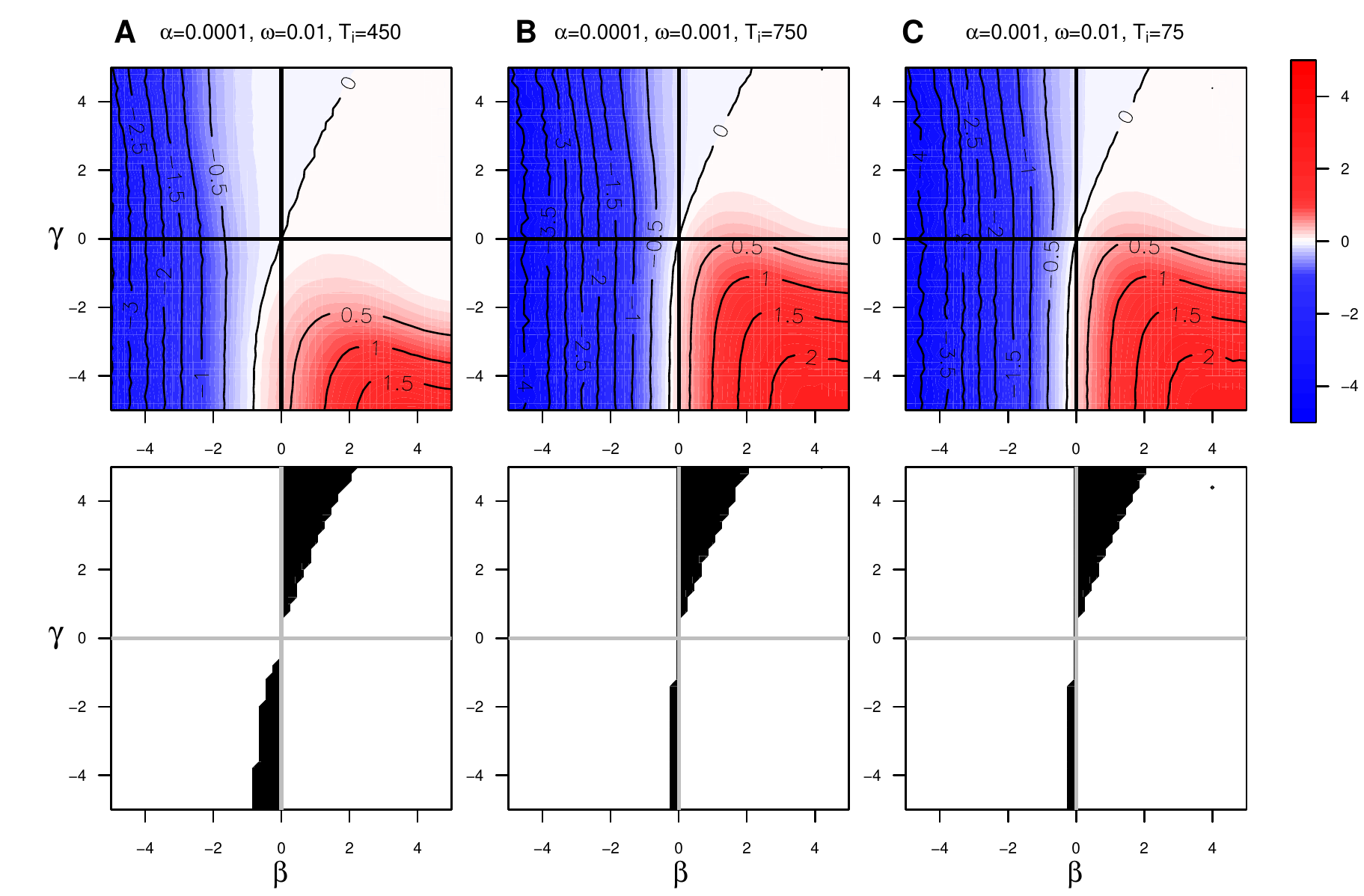}
\caption{$\log[RR]$ (top row) and region of direction bias (bottom row) as a function of $\beta$ and $\gamma$ for different combinations of ratio $\omega / \alpha$ and observation time $T_i$, when cluster size is constant ($n_i=4$ for all $i$), and $x$ is block randomized such that $\sum_{j=1}^{n_i} x_{ij}= 2$. In all plots observation time is constant and chosen such that the cumulative incidence when $\beta=0$ and $\gamma=0$ is approximately 0.15 for a given combination of $\alpha$ and $\omega$.} 
\label{fig:x_block_fix_ni_by_foi}
\end{figure}

\begin{figure}
\centering
\includegraphics[scale=0.7]{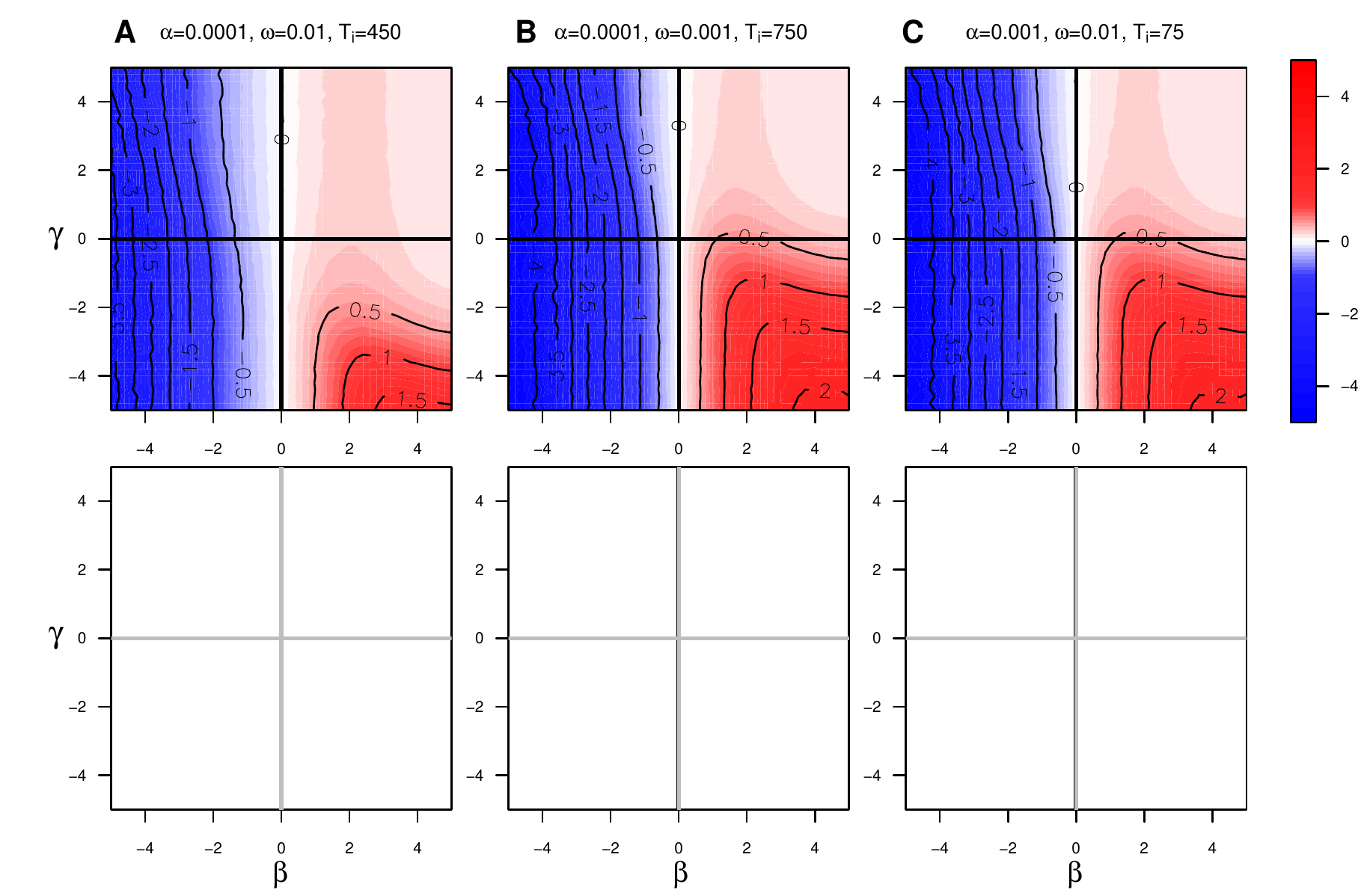}
\caption{$\log[RR]$ (top row) and region of direction bias (bottom row) as a function of $\beta$ and $\gamma$ for different combinations of ratio $\omega / \alpha$ and observation time $T_i$, when cluster size $n_i \sim \text{Pois}(3)+1$ and $x$ has independent Bernoulli distribution with $Pr[x=1]=0.5$. In all plots observation time is constant and chosen such that the cumulative incidence when $\beta=0$ and $\gamma=0$ is approximately 0.15 for a given combination of $\alpha$ and $\omega$.} 
\label{fig:x_bernoulli_var_ni_by_foi}
\end{figure}

\begin{figure}
\centering
\includegraphics[scale=0.9]{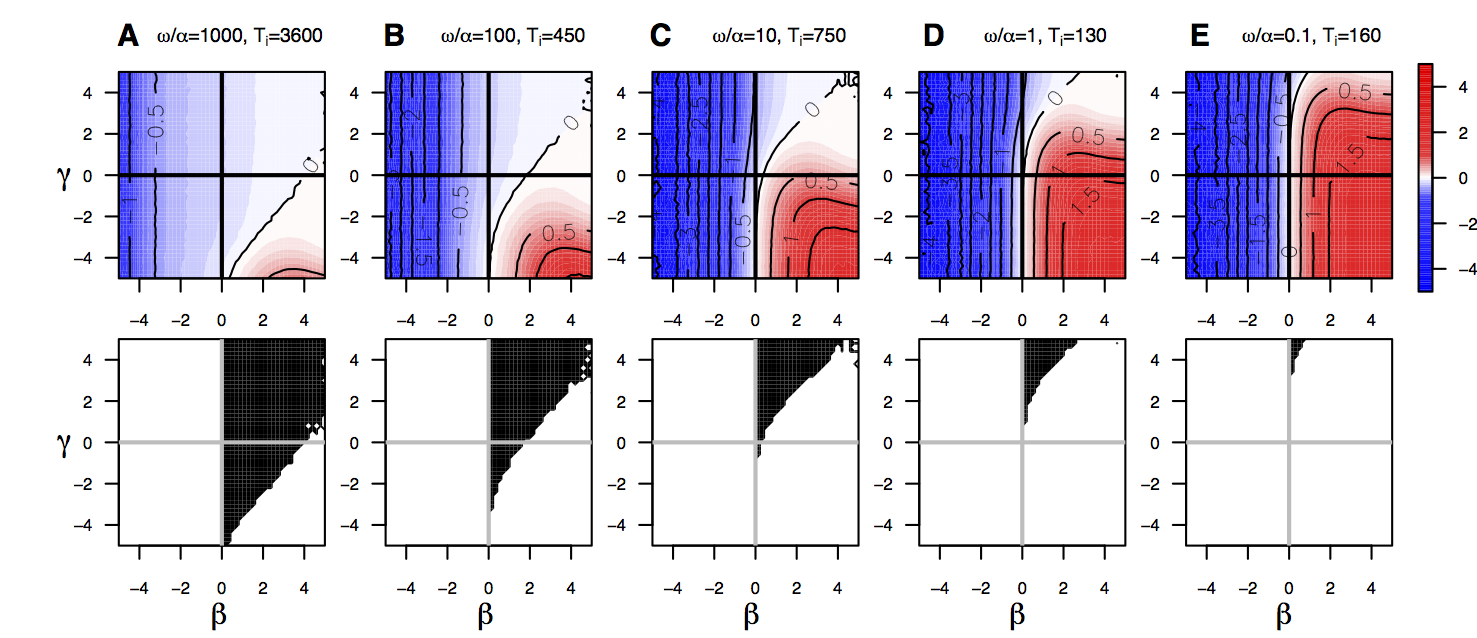}
\caption{$\log[RR]$ (top row) and region of direction bias (bottom row) as a function of $\beta$ and $\gamma$ for different combinations of ratio $\omega / \alpha$ and observation time $T_i$, when cluster size $n_i \sim \text{Pois}(3)+1$ and $x$ is block randomized such that $\sum_{j=1}^{n_i} x_{ij}=1$ for all $i$. In all plots observation time is constant and chosen such that the cumulative incidence when $\beta=0$ and $\gamma=0$ is approximately 0.15 for a given combination of $\alpha$ and $\omega$.} 
\label{fig:x_block1_var_ni_by_foi}
\end{figure}

\begin{figure}
\centering
\includegraphics[scale=0.9]{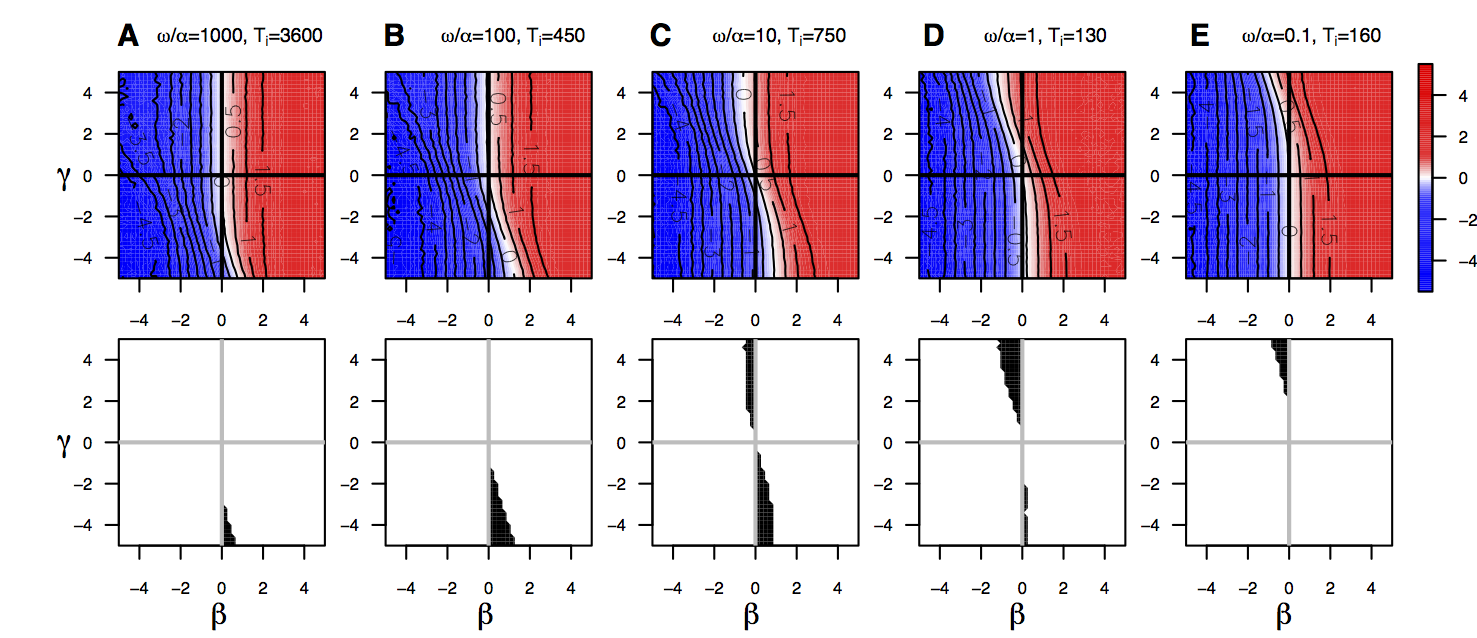}
\caption{$\log[RR]$ (top row) and region of direction bias (bottom row) as a function of $\beta$ and $\gamma$ for different combinations of ratio $\omega / \alpha$ and observation time $T_i$, when cluster size $n_i \sim \text{Pois}(3)+1$ and $x$ is cluster randomized: half of clusters have $\sum_{j=1}^{n_i} x_{ij}= n_i$, and remaining half have $\sum_{j=1}^{n_i} x_{ij}= 0$. In all plots observation time is constant and chosen such that the cumulative incidence when $\beta=0$ and $\gamma=0$ is approximately 0.15 for a given combination of $\alpha$ and $\omega$.} 
\label{fig:x_cluster_var_ni_by_foi}
\end{figure}

This subsection looks at the influence of the ratio $\omega / \alpha$ of per-subject within-cluster to exogenous force of infection. Figure \ref{fig:x_block_fix_ni_by_foi} shows simulation results for different values of $\omega$ and $\alpha$ under constant cluster size and block randomized distribution of $x$; Figure \ref{fig:x_bernoulli_var_ni_by_foi} - under variable cluster size and independent Bernoulli distribution of $x$, Figure \ref{fig:x_block1_var_ni_by_foi} - under variable cluster size and block randomized distribution of $x$, when exactly one subject per cluster has a value of $x=1$, and Figure \ref{fig:x_cluster_var_ni_by_foi} - under variable cluster size and cluster randomized distribution of $x$. Similarly to the previous subsection, in all plots the observation time $T_i$ was chosen such that the cumulative incidence when $\beta=0$ and $\gamma=0$ is approximately 0.15. In all simulations presented in this subsection (Figures \ref{fig:x_block_fix_ni_by_foi} - \ref{fig:x_cluster_var_ni_by_foi}) the following parameters are held constant:
\begin{itemize} 
	\item All subjects uninfected at baseline: $Y_{ij}(0) = 0$ for $i = 1,\ldots, N$ and $j = 1, \ldots, n_i$;
	\item Simulation parameters: number of clusters $N=500$, number of simulations per combination of parameters $=200$. 
\end{itemize} 

In the simple case of clusters of size two, for which we have derived analytic expression for the risk ratio bias, we have demonstrated that bias behavior is exactly the same for the same ratio of $\omega / \alpha$ when observation time $T_i$ is chosen such that it keeps cumulative incidence the same (Figures \ref{supfig:exact_est} and \ref{supfig:exact_dbias}). Figures  \ref{fig:x_block_fix_ni_by_foi} and \ref{fig:x_bernoulli_var_ni_by_foi} show that this property holds for more complex study designs. Figure \ref{fig:x_block1_var_ni_by_foi} shows that under the same conditions on $T_i$ and block randomized distribution of $x$, the region of the $(\beta, \gamma)$ space, where risk ratio exhibits direction bias increases with the increase of the ratio $\omega / \alpha$ as proportionally more infections get attributed to within-cluster transmission. However, under cluster randomized distribution of $x$ (Figure \ref{fig:x_cluster_var_ni_by_foi}) region, where the risk ratio is not direction-unbiased, is largest when ratio $\omega / \alpha$ gets closer to one.

\end{document}